\documentclass[a4paper,11pt]{article}
\usepackage[scale={0.75,0.9},centering,includeheadfoot]{geometry}

\usepackage[latin1]{inputenc} 
\usepackage{rotating}
\usepackage[T1]{fontenc}
\usepackage{verbatim}
\usepackage{amssymb,amsmath,amsfonts,amsthm}
\usepackage{graphics}
\usepackage{natbib}
\usepackage{pstricks,pst-node,psfrag}
\usepackage{epsfig}
\usepackage{paralist}
\usepackage{pstool}

\DeclareMathOperator{\R}{\mathbb{R}}

\DeclareMathOperator{\vep}{\varepsilon}

\DeclareMathOperator{\Cov}{Cov}

\DeclareMathOperator*{\argmax}{arg\,max}

\DeclareMathOperator{\trace}{tr}

\DeclareMathOperator{\diag}{diag}
\DeclareMathOperator{\chol}{chol}

\long\def\symbolfootnote[#1]#2{\begingroup%
	\def\thefootnote{\fnsymbol{footnote}}\footnote[#1]{#2}\endgroup}

\def\qed{\relax\ifmmode\hskip2em \Box\else\unskip\nobreak\hskip1em $\Box$\fi}

\newtheorem{thm}{Theorem}[section]
\newtheorem{lemma}[thm]{Lemma}
\newtheorem{prop}[thm]{Proposition}
\newtheorem{assumption}[thm]{Assumption}
\newtheorem{cor}[thm]{Corollary}

\newtheorem{defn}[thm]{Definition}
\newtheorem{example}{Example}
\newtheorem{rem}{Remark}

\newcommand{\proper}{\mathsf}

\newcommand{\pE}{\proper{E}}
\newcommand{\pV}{\proper{V}}

\newcommand{\pN}{\proper{N}}
\newcommand{\mv}[1]{{\boldsymbol{\mathrm{#1}}}}

\newcommand{\trsp}{\ensuremath{\top}}
\newcommand{\md}{\ensuremath{\,\mathrm{d}}}
\newcommand{\scal}[2]{\left\langle {#1},\,{#2} \right\rangle}

\usepackage{booktabs}
\usepackage{mathtools}

\definecolor{darkgreen}{rgb}{0, 0.8, 0.7}

\graphicspath{{figs/}}
\usepackage{tikz}

\usepackage{marginnote} 

\usepackage{algorithm}
\usepackage{algpseudocode}
\algnewcommand{\ffor}[1]{\State\algorithmicfor\ #1\ \algorithmicdo}
\algnewcommand{\Endffor}{\unskip\ \algorithmicend\ \algorithmicfor}

\makeatother

\newcommand{\materncorr}[3]{M\left({#1}\mid{#2},\,{#3}\right)}

\newcommand{\E}{{\bf \mathsf{E}}}

\DeclareMathOperator{\gnoise}{\dot{\mathcal{W}}}
\DeclareMathOperator{\noise}{\dot{\mathcal{M}}}
\newcommand{\proc}{x}
\newcommand{\weight}{w}
\newcommand{\var}{v}
\newcommand{\data}{y}
\newcommand{\BM}{D}
\newcommand{\BMi}{R}
\newcommand{\nignu}{\eta}

\usepackage{ulem}

\newtheorem{Theorem}{Theorem}

\allowdisplaybreaks

\title{\bf{Multivariate type G Mat\'ern stochastic partial differential equation random fields}}
\author{\scshape{David Bolin$^1$ and Jonas Wallin$^2$}\\
	{\small $^1$King Abdullah University of Science and Technology and University of Gothenburg}\\
	{\small $^2$Lund University}}
\date{}

\begin{document}
	\normalem
	\maketitle
	\begin{center}\begin{minipage}{0.9\textwidth}
			\noindent{\bf Abstract:} 
For many applications with multivariate data, 
random field models  
capturing departures 
from Gaussianity within realisations  
are appropriate. 
For this reason, 
we formulate a new class of multivariate non-Gaussian models 
based on systems of stochastic partial differential equations 
with additive type G noise 
whose marginal covariance functions 
are of Mat\'ern type.
We consider four increasingly flexible constructions of the noise, where the first two are similar to existing copula-based models. In contrast to these, the latter two constructions can model non-Gaussian spatial data without replicates. 
Computationally efficient methods for likelihood-based parameter estimation and probabilistic prediction are proposed, and the flexibility of the suggested models is illustrated by numerical examples and two statistical applications.

			\vspace{0.3cm}\noindent{\bf Key words:}
			Mat\'{e}rn covariances;	
			Multivariate random fields;
			Non-Gaussian models;
			Spatial statistics; 
			Stochastic partial differential equations.
		\end{minipage}
	\end{center}

\section{Introduction}
Motivated by an increasing number of spatial data sets with multiple measured variables, such as different climate variables from weather stations, various pollutants monitored in urban areas, or climate model outputs, the literature on models for multivariate random fields is growing rapidly. The majority of research in this area has focused on Gaussian random fields, and how to construct valid multivariate cross-covariance functions. 

Of particular interest has been multivariate extensions of the Mat\'{e}rn correlation function \citep{matern60}, 
$\materncorr{\mv{h}}{\kappa}{\nu} = 2^{1-\nu}\Gamma(\nu)^{-1}\left(\kappa \|\mv{h}\|\right)^{\nu}K_{\nu}\left(\kappa \|\mv{h}\|\right)$,  $\mv{h} \in \R^d$. Here $K_{\nu}$ is a modified Bessel function of the second kind and the positive parameters $\kappa$ and $\nu$ determine the practical correlation range and smoothness of the process respectively. \cite{gneiting2012matern} extended it to the multivariate setting by proposing a model with cross-correlation functions $\rho_{ij}\materncorr{\mv{h}}{\kappa_{ij}}{\nu_{ij}}$, where $\rho_{ij}$ are parameters determining the cross-correlations between the $i$th and $j$th component of the multivariate field. The parameters in this construction must be restricted to assure that it is a valid multivariate covariance function, and \cite{gneiting2012matern} proposed two models that satisfied this requirement: A parsimonious model, where $\kappa_{ij} \equiv \kappa$ and $\nu_{ij} = (\nu_{ii} + \nu_{jj})/2$, and a more general bivariate model that was later extended by \cite{apanasovich2012valid}. 

Even though most research has focused on Gaussian random fields, many data sets have features that cannot be captured by Gaussian models, such as exponential tails, non-Gaussian dependence, or asymmetric marginal distributions. There is thus a need for multivariate random fields that are more general than the Gaussian. Examples of such models in the literature are multivariate max-stable processes for spatial extremes \citep{genton2015multivariate} and Mittag-Leffler random fields \citep{ma2013mittag}. A common approach for constructing non-Gaussian fields is to multiply a Gaussian random field with a random scalar. Multivariate versions of this approach were explored by \cite{ma2013student} and \cite{du2012hyperbolic}.  Copula-based modelling is another popular method for non-Gaussian data, which has  been used for creating both univariate \citep{graler2014modelling,bardossy2006copula} and multivariate \citep{krupskii2016factor} random fields.

However, creating non-Gaussian multivariate random field models that allow for likelihood-based parameter estimation and probabilistic prediction is difficult, especially if they should be able to capture interesting departures from normality within realisations, and not just have non-Gaussian marginal distributions. This requirement excludes fields that are non-Gaussian only in the presence of repeated measurements, such as the factor-copula models \citep{krupskii2016factor} and the constructions based on multiplying Gaussian fields with random scalars. Many other copula-based approaches in geostatistics use Gaussian copulas. The resulting models are then equivalent to transformed Gaussian models \citep{kazianka2010copula}, which have many disadvantages \citep{Wallin15}. Thus, most existing approaches are either too limited, in the sense that they cannot capture essential features such as sample path asymmetry, or they lack methods for practical applications. For this reason, the recent review article on multivariate random fields by \cite{genton2015cross} listed creation of practically useful non-Gaussian multivariate random fields as an open problem. 

The main contribution of this work is to present a class of models that remedies this problem. The model class is constructed using systems of stochastic partial differential equations (SPDEs) driven by non-Gaussian noise. To facilitate computationally efficient likelihood-based inference, we use noise with normal-variance mixture distributions \citep{Barndorff1982}, which we refer to as type G noise. The restriction to normal-variance mixtures is not a big limitation, since several common distributions can be formulated in this way. Four increasingly flexible constructions are considered, where the simplest is closely related to factor copula models and the approach where a Gaussian field is multiplied with a random scalar. The more flexible constructions allow the fields to capture more complex dependency structures and departures from Gaussianity within realisations, while still allowing for likelihood-based inference. As an additional motivation for the more flexible constructions, we investigate the properties of spatial prediction based on the type G models, and in particular prove that distributions of spatial predictions for the simplest construction are asymptotically Gaussian. This means that if the goal is to use the model for spatial prediction, one might as well use a Gaussian model instead of the simple non-Gaussian constructions.

In the seminal work of \cite{lindgren10}, Gaussian random fields where formulated as solutions to SPDEs, which were apprroximated using a element (FE) discretization to allow for computationally efficient inference.
\cite{hu2013multivariate} and \cite{hu2016spatial} extended the work to multivariate Gaussian random field models based on systems of SPDEs. However, their models in general do not have explicit covariance functions, which can complicate the understanding of the effect each model parameter has, in particular in the non-Gaussian case. To avoid this problem, we formulate systems of SPDEs that result in models with marginal Mat\'ern covariance functions, having the parsimonious Mat\'ern model as a special case. We further discover a set of parameters in the model formulation that do not affect the covariance function, and therefore are unidentifiable for Gaussian models. These parameters, however, control the more complex dependence for non-Gaussian models.

As always with more general models than the Gaussian, there is an added computational cost for inference. However, using FE discretizations of our non-Gaussian models allows for the same computational complexity with respect the size of the discretized random field as for the corresponding Gaussian models. This makes the models applicable in scenarios where the data sets are so large that it prohibits the use of standard covariance-based models. As in the Gaussian case, the SPDE approach also facilitates extensions to non-stationary models by using spatially varying parameters. An important fact related to this is that the construction of the FE approximation is identical to that for Gaussian models, and thus as easy to compute.

The article is structured as follows. In Section \ref{sec:spde}, the link between systems of SPDEs and cross-covariances is studied. Section \ref{sec:nongauss} contains the definitions of the non-Gaussian models, as well as derivations of basic model properties. More details and examples of multivariate normal inverse Gaussian (NIG) fields, a special case of the type G models, are given in Section \ref{sec:nig}. In Section \ref{sec:model}, the type G fields are included in a geostatistical model for which we derive computationally efficient methods for likelihood-based parameter estimation and probabilistic prediction.
Section~\ref{sec:applications} presents two applications, and the article concludes with a discussion in Section~\ref{sec:discussion}. The article contains five appendices that present (\ref{sec:fem}) details on the FE discretizations; (\ref{sec:gradients}) gradients needed for the parameter estimation; (\ref{sec:pseudo}) sampling methods for the models;  (\ref{seq:parameter_estimates}) details on the applications; and (\ref{sec:proofs}) all proofs. The methods developed in this work have been implemented in the R package \texttt{ngme}.

\section{Multivariate Mat\'ern fields and systems of SPDEs}\label{sec:spde}

A Gaussian random field $x(\mv{s})$ on $\mathbb{R}^d$ with a Mat\'ern covariance function can be represented as a stationary solution to the stochastic partial differential equation
\begin{equation}\label{eq:model}
(\kappa^2-\Delta)^{\frac{\alpha}{2}} \proc = \gnoise, \quad \mbox{in $\mathcal{D}:=\R^d$},
\end{equation}
where $\alpha = \nu + d/2$, $\Delta$ is the Laplacian and $\gnoise$ is Gaussian white noise \citep{whittle63}. Extending equation \eqref{eq:model} to a system of SPDEs can be used to define more general covariance models \citep{bolin09b} and to define multivariate random fields. \cite{hu2013multivariate} and later \cite{hu2016spatial} proposed using systems of the form
\begin{align}\label{eq:spdesystem}
\begin{bmatrix}
\mathcal{K}_{11} & \mathcal{K}_{12}  & \cdots & \mathcal{K}_{1p}\\
\mathcal{K}_{21} & \mathcal{K}_{22}  & \cdots & \mathcal{K}_{2p}\\
\vdots & \vdots  & \ddots & \vdots \\
\mathcal{K}_{p1} & \mathcal{K}_{p2}  & \cdots & \mathcal{K}_{pp}\\
\end{bmatrix}
\begin{bmatrix}
\proc_1\\
\proc_2\\
\vdots \\
\proc_p
\end{bmatrix}
= \begin{bmatrix}
\gnoise_1 \\
\gnoise_2 \\
\vdots \\
\gnoise_p
\end{bmatrix},
\end{align}
to construct multivariate random fields, $\mv{x}(\mv{s}) = (x_1(\mv{s}), \ldots, x_p(\mv{s}))^{\trsp}$, where $\mathcal{K}_{ij}$ are pseudo-differential operators  such as $(\kappa^2-\Delta)^{\frac{\alpha}{2}}$ and $\gnoise_1, \ldots, \gnoise_p$ are mutually independent Gaussian white noise processes. 
\cite{hu2013multivariate} focused on the bivariate triangular system
\begin{align}\label{eq:uppertrisystem}
\begin{bmatrix}
\mathcal{K}_{11} & \mathcal{K}_{12} \\
&  \mathcal{K}_{22}
\end{bmatrix}
\begin{bmatrix}
\proc_1\\
\proc_2
\end{bmatrix}
= \begin{bmatrix}
\gnoise_1 \\
\gnoise_2
\end{bmatrix},
\end{align}
where $\mathcal{K}_{ij} = (\kappa_{ij}^2-\Delta)^{\frac{\alpha_{ij}}{2}}$. To better understand the  cross-covariance function for this model, one can informally invert the operator matrix to obtain
\begin{align}\label{eq:rep2}
\begin{bmatrix}
\proc_1 \\
\proc_2
\end{bmatrix}
=
\begin{bmatrix}
\mathcal{K}^{-1}_{11} & -\mathcal{K}^{-1}_{11} \mathcal{K}_{12}  \mathcal{K}^{-1}_{22} \\
&  \mathcal{K}^{-1}_{22}
\end{bmatrix}
\begin{bmatrix}
\gnoise_1 \\
\gnoise_2
\end{bmatrix}.
\end{align}
From this representation one can see that $\proc_2$ is marginally a Gaussian Mat\'ern field whereas $\proc_1$ is a sum of two Gaussian fields $\mathcal{K}^{-1}_{11}\gnoise_1$ and $-\mathcal{K}^{-1}_{11} \mathcal{K}_{12}  \mathcal{K}^{-1}_{22}\gnoise_2$ and thus has a more complicated covariance function. 

Although the full system \eqref{eq:spdesystem} may be of interest, the generality comes at the cost of a large number of parameters that are difficult to identify in practice, and equally hard to estimate. We therefore focus on the case when all marginal covariances are Mat\'ern, and on characterizing systems of SPDEs that result in models with this property.

\subsection{Multivariate Mat\'ern-SPDE fields}
To make the results in this section applicable beyond Gaussian models, we replace the right-hand side of \eqref{eq:spdesystem} by $\dot{\mv{\mathcal{M}}} = (\dot{\mathcal{M}}_1,\ldots, \dot{\mathcal{M}}_p)^{\trsp}$, where the components are mutually uncorrelated, but not necessarily independent,  $L_2$-valued independently scattered random measures (see Section \ref{sec:nig} and \citet{rajput1989spectral} for details). This includes Gaussian noise but also the non-Gaussian processes that we will study in the next section. We introduce the operator matrix $\mv{\mathcal{K}}$ with entries $\mv{\mathcal{K}}_{ij} = \mathcal{K}_{ij}$ and write \eqref{eq:spdesystem} more compactly as $\mv{\mathcal{K}}\mv{\proc}= \mv{\dot{\mathcal{M}}}$.
Investigating \eqref{eq:rep2}, we can note that $\proc_1$ has a Mat\'ern covariance function if $\mathcal{K}_{12} = \mathcal{K}_{22}$. This motivates the following definition of $p$-variate Mat\'ern-SPDE fields.

\begin{defn}\label{def:tri}
	A multivariate Mat\'ern-SPDE field on $\mathbb{R}^d$ is a solution to $\mv{\mathcal{K}}\mv{\proc} = \mv{\dot{\mathcal{M}}}$ where the operator matrix is of the form $\mv{\mathcal{K}} = \mv{D}\diag(\mathcal{L}_1,\cdots, \mathcal{L}_p)$.
	Here $\mv{\BM}$ is a real invertible $p\times p$ matrix and $\mathcal{L}_i = (\kappa_{i}^2-\Delta)^{\frac{\alpha_{i}}{2}}$ with $\kappa_i>0$ and $\alpha_i >d/2$ for $i=1,\ldots, p$.
\end{defn}

Since $\mv{\BM}$ defines the dependence structure of the process, we refer to it as a dependence matrix. 
That the multivariate Mat\'ern-SPDE model indeed has marginal Mat\'ern covariance functions is clarified in the following proposition.

\begin{prop}\label{thm1}
	Given that the driving noise in $\mv{\mathcal{K}}\mv{\proc} = \mv{\dot{\mathcal{M}}}$ has unit variance, the multivariate Mat\'ern-SPDE field $\mv{\proc}(\mv{s})$ on $\mathbb{R}^d$ has covariance function 
	$$
	\Cov(\proc_i(\mv{s}),\proc_j(\mv{t})) = \begin{cases}
	\frac{\Gamma(\nu_i)\sum_{j=1}^p \BMi_{ii}^2}{\Gamma(\alpha_i)(4\pi)^{d/2}\kappa_i^{2\nu_i}}\materncorr{\|\mv{s}-\mv{t}\|}{\kappa_i}{\nu_i} & i = j, \\
	\mathcal{F}^{-1}(S_{ij})(\|\mv{s}-\mv{t}\|) & i \neq j,
	\end{cases}
	$$
	where $\BMi_{ij}$ are the elements of $\mv{\BMi} = \mv{\BM}^{-1}$,  $\mathcal{F}^{-1}$ denotes the inverse Fourier transform, and 
	\begin{equation}\label{eq:cross_spec}
	S_{ij}(\mv{k}) = \frac{\sum_{l=1}^p\BMi_{il}\BMi_{jl}}{(2\pi)^d}\frac{1}{(\kappa_i^2+\|\mv{k}\|^2)^{\frac{\alpha_i}{2}}
		(\kappa_j^2+\|\mv{k}\|^2)^{\frac{\alpha_j}{2}}}.
	\end{equation}
\end{prop}
Note that $\mv{\BM}$ determines the strength of the cross-correlations, and that $\Cov(\proc_i(\mv{s}),\proc_j(\mv{t}))$ for $i\neq j$ is a Mat\'ern covariance function only if $\kappa_j = \kappa_j$. In the case when $\kappa_i=\kappa$ for all $i$, the model coincides with the parsimonious Mat\'ern model by \cite{gneiting2012matern}. Also note that the shapes of the cross-correlation functions are determined by the parameters of the marginal correlation functions. 
This is slightly more restrictive than the general covariance-based multivariate Mat\'ern models, but has the advantage that there are no difficult-to-check restrictions on the model parameters. Furthermore, both \cite{gneiting2012matern} and \cite{apanasovich2012valid} argued that the most important aspect of multivariate models is to allow for flexibility in the marginal covariances while still allowing for some degree of cross-covariance. Thus, the Mat\'ern-SPDE model should be a sufficiently flexible alternative to multivariate Mat\'ern fields for most applications.

\begin{rem}\label{cor2}
	An immediate consequence of Definition \ref{def:tri} is that $\mv{x}$ alternatively can be  obtained as a solution to a diagonal system of SPDEs driven by correlated noise:  $\diag(\mathcal{L}_1,\ldots, \mathcal{L}_p)\mv{\proc}(\mv{s})  = \mv{\noise}_{R}$, where $\mv{\noise}_{R} = \mv{R}\mv{\noise}$ and $\mv{R} = \mv{D}^{-1}$.  This means that the model can be viewed as a linear model of coregionalization. 
\end{rem}

\subsection{Parameterising the model}\label{sec:param}
An important question for practical applications of the multivariate Mat\'ern-SPDE fields is if the model parameters (the dependence matrix and the parameters of the operators) are identifiable. The following proposition shows that this is not the case in general.

\begin{prop}\label{thm2}
	Two multivariate Mat\'ern-SPDE fields, with the same operators $\mathcal{L}_1,\ldots, \mathcal{L}_p$ and with dependence matrices $\mv{\BM}$ and $\hat{\mv{\BM}}$ respectively, have equal covariance functions if and only if $\mv{\BM} = \mv{Q}\hat{\mv{\BM}}$ for an orthogonal matrix $\mv{Q}$. For any choice of $\mv{\BM}$, one can find a triangular matrix $\hat{\mv{\BM}}$ that gives the same covariance functions. In particular, $\hat{\mv{D}} = \chol(\mv{D}^{\trsp}\mv{D})$ is the unique upper-triangular choice with positive diagonal elements.
	
\end{prop}

We will refer to models with triangular dependence matrices as triangular Mat\'ern-SPDE fields. Since Gaussian fields are uniquely specified by the first two moments, the proposition implies that the matrix $\mv{\BM}$ is not completely identifiable from data for Gaussian models, so there is no point in considering non-triangular Gaussian models. This is however not the case for non-Gaussian models, where non-triangular dependence matrices can be used to define more general dependence structures. 

Since the dependence matrix is not completely identifiable for Gaussian models, a different model parametrization that separates the control of marginal variances, cross-correlations, and higher moments is preferable. To derive such a parametrization, we use Proposition \ref{thm2} to write $\mv{D} = \mv{Q}_p\mv{D}_l$, where $\mv{D}_l$ is a triangular matrix and $\mv{Q}_p$ is an orthogonal matrix. Then $\mv{D}_l$ and $\mv{Q}_p$ respectively determine the cross-covariances and  the higher moments. To separate the control of the variances and cross-correlations, we rescale the operators $\mathcal{L}_i$ by constants $c_i = \sqrt{\sigma_i^{-2}(4\pi)^{-d/2}\kappa_i^{-2\nu_i}\Gamma(\nu_i)/\Gamma(\alpha_i)}$ and parametrize $\mv{D}_l$ as, 
$$
\mv{\BM}_l(\mv{\rho}) = 
\begin{pmatrix}
1 		& 			& 		& 		& \\
\rho_{1,1} 	& 1 			& 		& 		& \\
\rho_{2,1} 	& \rho_{2,2}		& 1 		& 		& \\
\vdots 	& \vdots 		& \ddots 	& \ddots 	& \\
\rho_{p,1} 	& \rho_{p,2}	& \hdots	& \rho_{p,p-1}	& 1
\end{pmatrix}^{-1}
\diag\left(1,k_2(\mv{\rho}),k_3(\mv{\rho}),\ldots,k_p(\mv{\rho})\right),
$$
where $k_j(\mv{\rho}) = \sqrt{1+\sum_{i<j}\rho_{j,i}^2}$. With this parametrization, $\mv{\rho} \in \R^{p(p-1)/2}$ controls the cross-correlations and $\sigma_i^2 = \pV(X_i(\mv{s}))$.
Figure \ref{fig1} shows an example of the resulting covariance function for a bivariate model with $\rho = \rho_{1,1} = 0.5$.

\begin{figure}
	\begin{center}
		\includegraphics[width=0.7\textwidth]{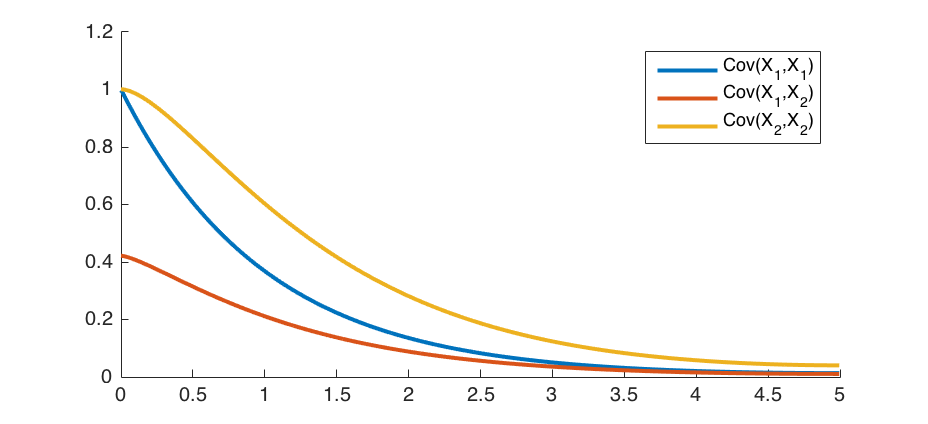}
	\end{center}
	\vspace*{-0.2cm}
	\caption{Example of covariance functions for the solution to the triangular Mat\'ern-SPDE with $\sigma_1 = \sigma_2 = 1$, $\rho = 0.5$, $\kappa_1 = \kappa_2 = 1$, $\alpha_1 = 1.5$, and $\alpha_2 =  2$.}
	\label{fig1}
\end{figure}

What remains is to find a parametrization of $\mv{Q}_p$. The determinant of an orthogonal matrix is $\pm 1$, where the sign is not identifiable in general. It is therefore enough to consider the subclass of special orthogonal matrices, which have determinant $1$. For a general $p$, it is difficult to parametrize such matrices. However, for $p=2$ and $p=3$ we can use the fact that they are equivalent to rotation matrices. We can therefore write 
$$
\mv{Q}_2(\theta) = 
\begin{pmatrix}
\cos(\theta) & -\sin(\theta) \\
\sin(\theta) & \cos(\theta) 
\end{pmatrix}, \quad \mv{Q}_3(\theta_1,\theta_2,\theta_3) = \mv{Q}_{3x}(\theta_1)\mv{Q}_{3y}(\theta_2)\mv{Q}_{3z}(\theta_3),
$$
where $\theta\in[0,2\pi]$, $\mv{Q}_{3x}(\theta) = \diag(\mv{Q}_2(\theta),1)$, $\mv{Q}_{3z}(\theta) = \diag(1,\mv{Q}_2(\theta))$, and
$$
\mv{Q}_{3y}(\theta) = 
\begin{pmatrix}
\cos(\theta) & 0 & -\sin(\theta) \\
0 & 1 & 0 \\
-\sin(\theta) & 0 & \cos(\theta) 
\end{pmatrix}.
$$
To summarize, we use the parametrization 
\begin{equation}\label{eq:general}
\mv{\BM}(\mv{\theta},\mv{\rho}) \diag(c_1\mathcal{L}_1,\cdots, c_p\mathcal{L}_p)\mv{\proc}(\mv{s}) = \mv{\noise},
\end{equation}
where $\mv{\BM}(\mv{\theta},\mv{\rho}) = \mv{Q}_p(\mv{\theta})\mv{\BM}_l(\mv{\rho})$ and $\mv{\theta}\in [0,2\pi]^{p(p-1)/2}$ will control higher moments for non-Gaussian models.  In the bivariate case, the dependence matrix simplifies to
\begin{equation}\label{eq:matern2}
\mv{\BM}(\theta,\rho) = 
\begin{bmatrix}
\cos(\theta)+\rho\sin(\theta) & -\sin(\theta)
\sqrt{1+\rho^2} \\
\sin(\theta)-\rho\cos(\theta) & \cos(\theta)
\sqrt{1+\rho^2}
\end{bmatrix}.
\end{equation}

\section{Type G Mat\'ern SPDE fields}\label{sec:nongauss}
In this section, the multivariate Mat\'ern-SPDE model is extended beyond Gaussianity by replacing the Gaussian noise with non-Gaussian noise. In Section \ref{sec:nonGauss-def}, four different constructions of noise for this approach are introduced and the resulting Mat\'ern-SPDE fields are discussed. The differences between the four constructions are illustrated using FE discretizations of the models in Section \ref{sec:discrete} and further properties of the models are stated in Section \ref{sec:nonGauss-prop}. Finally, asymptotic properties of spatial prediction based on the simplest type G models are derived in Section \ref{sec:nonGauss-outfill}.

\subsection{Four increasingly flexible constructions}\label{sec:nonGauss-def}
The four constructions are based on using different types of normal-variance mixtures 
\begin{equation}\label{eq:variancemixture}
\mv{\gamma} + v\mv{\mu}  + \sqrt{v}z, 
\end{equation}
where $\mv{\gamma}\in\R^p$ and $\mv{\mu}\in\R^p$ are parameters, $z\sim\pN(0,1)$, and $v$ is a non-negative random variable. Inspired by L\'{e}vy process, which are said to be of type G if their increments are normal-variance mixtures, we will refer to these models as type G Mat\'ern-SPDE fields.

The first two constructions are related to the approach where non-Gaussian fields are obtained by multiplying Gaussian fields with random scalars. 

\begin{defn}
	Let $v$ and $v_1,\ldots,v_p$ be independent non-negative infinitely divisible random variables and set $\mv{v}_1 = v \mv{1}_p$ and $\mv{v}_2 = (v_1,\ldots,v_p)^{\trsp}$, where $\mv{1}_p$ denotes a vector with $p$ ones. Further, let $\mathcal{W}(\mv{s}) = (\mathcal{W}_1(\mv{s}),\ldots,\mathcal{W}_p(\mv{s}))^{\trsp}$ be a vector of independent copies of Brownian sheets on $\mathbb{R}^d$. For $i\in \{1,2\}$, a type G$_i$ Mat\'ern-SPDE field is obtained by using $\mv{\noise}_i$ in \eqref{eq:general} where
	$\mv{\mathcal{M}}_i(\mv{s}) = \mv{\gamma} + \diag(\mv{v}_i)\mv{\mu} + \diag(\sqrt{\mv{v}_i})\mv{\mathcal{W}}(\mv{s})$.
\end{defn}

It should be noted that $\mv{\noise}_1$ and $\mv{\noise}_2$ are not independently scattered measures since they have common random scaling $\mv{v}_i$ for different spatial locations. Because of this, one cannot directly use the results from the previous section. However, the results can easily be extended by allowing for measures that are independently scattered conditionally on a random variable $\mv{v}$. 
In particular, if we restrict the distribution of $\mv{v}_i$ such that $\pE(\mv{v}_i) = \mv{1}$, then the result in Proposition~\ref{thm1} still holds in the symmetric case with $\mv{\mu} = \mv{0}$. In the non-symmetric case, the resulting fields have covariance functions given by the covariance function in Proposition~\ref{thm1} plus a constant factor depending on the variance of $\mv{v}$.  
With the restriction $\pE(\mv{v}_i) = \mv{1}$, the mean of the process is given by $\mv{\mu} + \mv{\gamma}$, and we therefore set $\mv{\gamma} = -\mv{\mu}$ to ensure that the process has zero mean as default.
In the type G$_1$ model, we can interpret $v$ as a random scaling of the variance of the entire process, whereas we scale the variance of each $x_i(\mv{s})$ separately with $v_i$ in the type G$_2$ case. When $\mv{\mu}\neq \mv{0}$, $v$ also decides the skewness of the marginal distributions of $\mv{x}(\mv{s})$ in in the type G$_1$ case, whereas $v_i$ controls the skewness of $x_i(\mv{s})$ in the type G$_2$ case. Hence, the type G$_2$ model gives more control of the marginal distributions of the process. 

From a Bayesian point of view, one could interpret $\pi(\mv{v})$ as a prior distribution on the mean and variance of a multivariate Gaussian random field. Thus, these models can be used in the same way as the Gaussian models in a Bayesian setting, but where we have a specific prior that links the mean and variance of the field. From this point of view, one would likely not refer to these models as non-Gaussian. 

The next two constructions are based on type G L\'evy noise. A random variable $x$ is said to be of type G if it can be written as $x \overset{d}{=} \sqrt{v}z$, where $z$ is a Gaussian variable and $v$ is an infinitely divisible non-negative random variable. A univariate type G L\'{e}vy process is a L\'evy process whose increments are of type G. \cite{rosinski91} showed that a type G process $\mathcal{M}(s), s\in [0,1]$, with $\mathcal{M}(1) \overset{d}{=} \sqrt{v}z$ can be represented as $\mathcal{M}(s) = \sum_{k=1}^{\infty}z_k g(e_k)^{\frac{1}{2}}\mathbb{I}(s \geq u_k)$, where $e_k$ are the points of a unit-rate Poisson process on $\mathbb{R}^+$, $z_k$ are iid $\pN(0,1)$ random variables, and  $u_k$ are iid uniform random variables on $(0,1)$. The function $g$ is the generalized inverse of the tail L\'{e}vy measure for $v$, defined as $g(u) = \inf\{x>0: M(x,\infty) \leq u\}$ where $M$ is the L\'evy measure of $v$. The non-decreasing L\'evy process $v(\mv{s}) = \sum_{k=1}^{\infty}g(e_k)\mathbb{I}(s \geq u_k)$ has the same Levy measure as $v$, and can be used to represent $\mathcal{M}$ as a subordinated Wiener process. We refer to \cite{rosinski91} for further technical details on the construction.  In the spatial case, a type G process $\mathcal{M}(\mv{s})$ on the unit square $D = [0, 1]\times [0, 1]$ with $M(\mv{1}) \overset{d}{=} \sqrt{v}z$ can similarly be represented as $\mathcal{M}(\mv{s}) = \sum_{k=1}^{\infty}z_k g(e_k)^{\frac{1}{2}}\mathbb{I}(\mv{s} \geq \mv{u}_k)$, where $\mv{u}_k$ now are uniform random variables on $D$ and $\mathbb{I}(\mv{s}\geq \mv{u}_k) = \mathbb{I}(s_1\geq u_{k,1})\mathbb{I}(s_2\geq u_{k,2})$ is a two-dimensional indicator function. In this case, the associated process $v(\mv{s}) = \sum_{k=1}^{\infty} g(e_k)\mathbb{I}(\mv{s} \geq \mv{u}_k)$ can no longer be seen as a subordinator, but could informally be thought of as a process that determines the variance of the noise. 
For multivariate processes, there are two natural extensions to vector valued noise that we use to define type G$_3$ and type G$_4$ fields.

\begin{defn}\label{def:noise}
	Let $\mathcal{M}(\mv{s})$ be a type G L\'evy processes with $v(\mv{s}) =  \sum_{k=1}^{\infty}g(e_k)\mathbb{I}(\mv{s} \geq \mv{u}_k) $ and let $\mv{\mathcal{M}}(\mv{s}) = (\mathcal{M}_1(\mv{s}), \ldots, \mathcal{M}_p(\mv{s}))^{\trsp}$ be a vector of independent type G L\'evy processes with corresponding variance processes $\mv{v}(\mv{s}) = (v_1(\mv{s}), \ldots v_p(\mv{s}))^{\trsp}$. For $i\in\{3,4\}$ a type G$_i$ Mat\'ern-SPDE field is obtained by using $\mv{\noise}_i$ in \eqref{eq:general} where 
	\begin{align*}
	\mv{\mathcal{M}}_{3}(\mv{s}) &= \mv{\gamma} + \mv{\mu}v(\mv{s}) + \sum_{k=1}^{\infty}g(e_k)^{\frac{1}{2}}\mathbb{I}(\mv{s} \geq \mv{u}_k) \mv{z}_k, & &
	\mv{\mathcal{M}}_4(\mv{s}) = \mv{\gamma} + \diag(\mv{\mu}) \mv{v}(\mv{s}) + \mv{\mathcal{M}}(\mv{s}).
	\end{align*}
\end{defn}

\begin{rem}
	In this section we have assumed a multivariate setting, i.e., $p>1$. However, in the univariate case, the type G$_1$ and type G$_2$ Mat\'ern-SPDE models are equivalent. Further if $\mv{\mu}=0$, the  type G$_1$ model is a Gaussian Mat\'ern field multiplied with a univariate positive random variable. Thus, models such as the t-distributed random fields by \citet{roislien06} belong to the class of type G$_1$ fields. Also, when $p=1$ the type G$_3$ and type G$_4$ Mat\'ern-SPDE models are also equivalent, and coincide with the models in \citep{Wallin15}.
\end{rem}

\subsection{Understanding the four constructions through their discretizations}\label{sec:discrete}
Although the Mat\'ern-SPDE models were formulated on the entire $\mathbb{R}^d$ in Section \ref{sec:spde}, we consider the system of SPDEs on a bounded domain $\mathcal{D}\subset \mathbb{R}^d$ when implementing them numerically. The operators are then equipped with suitable boundary conditions and the solution is approximated using a FE discretization derived in Appendix \ref{sec:fem}.
To understand the differences between the four different type G constructions, we now examine the properties of the discretized models in comparison to the corresponding Gaussian Mat\'ern-SPDE model.
In the FE approximation, the solution of \eqref{eq:general} is represented as a basis expansion $\mv{\proc}(\mv{s}) = \sum_{j=1}^n\sum_{k=1}^p w_{jk}\mv{\varphi}_j^k(\mv{s})$ using piecewise linear basis functions $\mv{\varphi}_j^k(\mv{s})$ obtained from a mesh over $\mathcal{D}$. The value of $\proc_k(\mv{s}_j)$ at the $j$th node in the mesh, $\mv{s}_j$, is then given by the stochastic weight $w_{jk}$.
Assuming $p=2$ and $\alpha_1 = \alpha_2 = 2$, the distribution of $\mv{w} = (\mv{w}_1^{\trsp},\mv{w}_2^{\trsp})^{\trsp} = (w_{11},\ldots, w_{n1},w_{12},\ldots, w_{n2})^{\trsp}$ for the case of Gaussian noise is
\begin{equation}\label{eq:Gausfem}
\mv{w} \sim \pN(\mv{0}, \mv{K}^{-1}\diag(\mv{h},\mv{h})\mv{K}^{-\trsp}), 
\end{equation}
where $\mv{K}$ is a discretization of the operator matrix and $\mv{h}$ is a vector with elements $h_i$ depending on the mesh.

For the corresponding type G$_3$ model, the distribution of the weights is
\begin{equation}\label{eq:G3fem}
\mv{w}|\mv{v} \sim \pN\left(\mv{K}^{-1}\left[\begin{matrix}\gamma_1\mv{h}+\mu_1\mv{v}\\ \gamma_2\mv{h}+\mu_2\mv{v}\end{matrix}\right], \mv{K}^{-1}\diag(\mv{v},\mv{v})\mv{K}^{-\trsp}\right), \quad \mv{v}\sim \pi(\mv{v}),
\end{equation}
where the elements of $\mv{v}\in\mathbb{R}_{+}^{n}$ are independent variables relating to the discretization of the variance process $v(\mv{s})$. 
For the type G$_4$ model, we have 
\begin{equation}\label{eq:G4fem}
\mv{w}|\mv{v}_1,\mv{v}_2 \sim \pN\left(\mv{K}^{-1}\left[\begin{matrix}\gamma_1\mv{h}+\mu_1\mv{v}_1\\ \gamma_2\mv{h}+\mu_2\mv{v}_2\end{matrix}\right], \mv{K}^{-1}\diag(\mv{v}_1,\mv{v}_2)\mv{K}^{-\trsp}\right), \quad \mv{v}_1, \mv{v}_2\sim \pi(\mv{v}),
\end{equation}
where $\mv{v}_1,\mv{v}_2\in\mathbb{R}_{+}^{n}$ have independent components relating to the discretisations of $v_1(\mv{s})$ and $v_2(\mv{s})$ repectively. 
Similarly, the discretization in the type G$_1$ and type G$_2$ cases can be written as \eqref{eq:G3fem} and \eqref{eq:G4fem} respectively, if we define $\mv{v} = v\mv{h}$ and $\mv{v}_i = v_i\mv{h}$. As we discussed for the first two cases, 
we set $\mv{\mu} = -\mv{\gamma}$ to ensure that the process has zero mean, and restrict the distribution of the variances to have mean one. We then have for all cases that $\pE(\mv{v}) = \pE(\mv{v}_1) = \pE(\mv{v}_2) = \mv{h}$. 
Thus, comparing \eqref{eq:Gausfem}, \eqref{eq:G3fem}, and \eqref{eq:G4fem}, we see that a difference between the type G processes and the Gaussian process is that we have replaced the deterministic vector $\mv{h}$ in the covariance matrix by a stochastic vector that has $\mv{h}$ as expected value. Furthermore, the difference between the four constructions lies in the flexibility of this stochastic vector. In the type G$_1$ case, we scale the entire field by a single stochastic variable, whereas we scaled each dimension separately in the type G$_2$ case. For the type G$_3$ case we have replaced the fixed scaling $h_i$ of the distribution of the weights $w_{i1},\ldots, w_{ip}$ for a given spatial location $\mv{s}_i$ by a common stochastic scaling $v_i$, which thus affect the sample path behaviour of the process. The type G$_4$ case is even more general where we have 
individual stochastic scalings $v_{ip}$ for each weight, and thus more control over the sample path behaviour.

\subsection{Properties of the four constructions}\label{sec:nonGauss-prop}
The four type G constructions provide random fields with increasing flexibility. All contain several interesting special cases depending on which distribution that is used for the variance components, such as generalised asymmetric Laplace distributions, normal-inverse gamma distributions, and Student's t-distributions. As an example, we will in the next section use NIG noise to highlight some properties of the constructions. 

Let $\mv{\Sigma}$ be the covariance matrix of the solution $\mv{x}(\mv{s})$ in \eqref{eq:general} in the case of Gaussian driving noise, for a fixed location $\mv{s}$. This matrix has diagonal elements  $\Sigma_{ii} = \sigma_i^2$ and off-diagonal elements $\Sigma_{ij}$ depending on $\sigma_i, \sigma_j$, and $\rho_{ij}$. For the type G$_1$ construction, we can then write the joint cumulative distribution function (CDF) $F^{(1)}$ of $\mv{x}(\mv{s})$, and the marginal CDFs $F^{(1)}_k$ of $x_k(\mv{s})$ for $k=1,\ldots, p$ as
\begin{equation*}
F^{(1)}(\mv{u}) = \int \Phi_{\mv{\Sigma}}\left(\frac{\mv{u}-\mv{\gamma}-\mv{\mu}v}{\sqrt{v}}\right)\md F_v(v), \quad  F^{(1)}_k(u) = \int \Phi\left(\frac{u-\mv{\gamma}-\mv{\mu}v}{\sigma_k\sqrt{v}}\right)\md F_v(v), 
\end{equation*}
where $\Phi_{\mv{\Sigma}}$ denotes the CDF of a $\pN(\mv{0},\mv{\Sigma})$ random variable and $F_v$ denotes the CDF of $v$. There are several choices of $F_v$ that result in fields with known marginal distributions. If for example $\mv{\mu} = \mv{0}$, the field has multivariate Student's $t$ marginals if $v$ is inverse-gamma distributed, and multivariate Laplace marginals if $v$ is gamma distributed. The copula of $\mv{x}(\mv{s})$ is
$
C^{(1)}(\mv{u}) = F^{(1)}[(F^{(1)}_1)^{-1}(u_1),\ldots, (F^{(1)}_p)^{-1}(u_p)],
$
which could be viewed as a generalization of the one-factor copulas in \citep{krupskii2015structured,krupskii2016factor}. However, despite the flexibility of the marginal distributions, the model is limited since it is non-ergodic for any non-singular distribution of $v$, and the sample paths are indistinguishable from sample paths of a Gaussian random field. If repeated realizations are available, one can estimate the distribution of $v$, but not the parameter $\mv{\theta}$ in the dependence matrix.

For the type G$_2$ construction, the joint CDF of $\mv{x}(\mv{s})$ is
\begin{equation*}
F^{(2)}(\mv{u}) = \int \Phi_{\mv{\Sigma}}\left(\diag\left(\frac1{\sqrt{v_1}},\ldots, \frac1{\sqrt{v_p}}\right)(\mv{u}-\mv{\gamma}-\mv{\mu}v)\right)\md F_{v_1}(v_1) \cdots\md F_{v_p}(v_p),
\end{equation*}
and the marginal CDF for $k=1,\ldots,p$ is 
$$
F^{(2)}_k(u) = \int \Phi\left(\frac{u-\gamma_k-\mu_kv_k}{\sigma_k\sqrt{v_k}}\right)\md F_{v_k}(v_k).
$$
The copula of $\mv{x}(\mv{s})$ is 
$
C^{(2)}(\mv{u}) = F^{(2)}[(F^{(2)}_1)^{-1}(u_1),\ldots, (F^{(2)}_p)^{-1}(u_p)],
$
which is similar to the $p$-factor copulas in \citep{krupskii2015structured}. Also fields obtained using the type G$_2$ construction are non-ergodic and have sample paths that are indistinguishable from Gaussian sample paths. However, it is possible to estimate all parameters of the model given multiple realizations. 

Since the type G$_1$ and the type G$_2$ constructions have copulas simular to factor copulas, one can compute their so-called tail dependence coefficients and derive conditions on the distribution of $v$ to study their asymptotic tail dependence similar to \cite{krupskii2016factor}. We leave this for future research as our main interest is in the more flexible type G$_3$ and type G$_4$ constructions. The reason for this is, as we will show in the next subsection, that the type G$_1$ and type G$_2$ models have asymptotically Gaussian conditional distributions. This greatly limits their flexibility for spatial data.

For the type G$_3$ and type G$_4$ constructions, we in general cannot derive closed-form expressions for the marginal distributions and copulas (we will discuss this further in the next section). However, if we use the representation of the process in Remark 1, and let $F^M_k$ and $\tilde{F}^M_{k}$ denote the distribution functions of the laws of $\mathcal{M}_k$ and $(\mv{R\mathcal{M}})_{k}$ respectively, the copula for the law of $\mv{\mathcal{M}}_R$ can be written as 
$$
C(\mv{u}) = \prod_{k=1}^p F^M_k(\mv{D}_k^{\trsp}((\tilde{F}^M_1)^{-1}(x_1), \ldots, (\tilde{F}_p^M)^{-1}(x_p))^{\trsp}),
$$
where $\mv{D}_k$ is the $k$th row of $\mv{D}$. This is a Gaussian copula only in the case when $\mv{\mathcal{M}}$ is Gaussian. Thus, also for these constructions, the dependence structure induced by the model can be made more flexible than simply using Gaussian copulas to model the dependence. 
The type G$_4$ construction is the most general but the type G$_3$ construction could be of interest for applications where one wish to capture dependence of the extreme values on different variables. It also has the following interesting feature.
\begin{prop}\label{thm3}
	Let $\mv{\proc}$  be a type G$_3$ Mat\'ern-SPDE field with $\mv{\rho} = \mv{0}$. Then, for $\mv{s},\mv{t}\in\mathcal{D}$ and $i\neq j$, $\proc_i(\mv{s})$ and $\proc_j(\mv{t})$ are dependent but uncorrelated random variables.
\end{prop}

\subsection{Increasing domain asymptotics for the type G$_1$ model}\label{sec:nonGauss-outfill}
In this section we explore the distributions of spatial predictions for the type G$_1$ models and show that they converge to Gaussian distributions as the number of observations goes to infinity. This implies that one might as well use the simpler Gaussian model for the purpose of prediction if the data set is sufficiently large. Similar issues with related non-Gaussian models have been noted in the mixed effect literature \citep{rubio2018flexible}. To simplify the notations, we focus on the mean-zero univariate case, but the results are easily extended to the general multivariate setting for both type G$_1$ and type G$_2$ models. 

Let $x_{i}=x(\mv{s}_i), i=1,\ldots,n$, be observations of  a mean-zero random field $x(\mv{s})$, for which we want to predict $x_{0}=x(\mv{s}_0)$.
Let $\mv{x}_{k:n}$ denote the vector $\left[x_k, x_{k+1}, \ldots, x_n\right]^{\trsp}$ and assume that the covariance function, $C(\mv{s},\mv{t})$, of $x$ and the locations $\mv{s}_0,\mv{s}_1,\ldots\mv{s}_n$ are such that covariance matrix of $\mv{x}_{0:n}$ is positive definite. 
Assuming that a mean-zero type G$_1$ model, with the same covariance function as $x$, is used for the prediction, the distribution of $x_0$ given $\mv{x}_{1:n}$ is
$$
\pi_{G_1,x_0}(x_0|\mv{x}_{1:n}) = \int N(x_0; \mv{c}_{0,1:n} \mv{C}_n^{-1} \mv{x}_{1:n} ,\, vc_{0}- v\mv{c}^{\trsp}_{0,1:n} \mv{C}_n^{-1} \mv{c}_{0,1:n})\pi(v) dv,
$$
where $c_0 = \pV(x_0)$, $\mv{c}_{0,1:n}$ is the cross-covariance between $\mv{x}_{1:n}$ and $x_0$, and $\mv{C}_n$ is the covariance matrix of $\mv{x}_{1:n}$.
To show that this distribution converges to a Gaussian distribution we need the following weak assumptions on the observed data. 

\begin{assumption}\label{ass1}
	The random field $x(\mv{s})$ and the observations satisfy, as $n\rightarrow \infty$,
	\begin{alignat}{4}
	&(\mv{x}_{1:n}^{\trsp} \mv{C}_n^{-1}\mv{x}_{1:n} )/n  		&&\overset{p}{\to}\, K_0, 	   \label{eq:mean_est}\\
	&\pV[(\mv{x}_{1:n}^{\trsp} \mv{C}_n^{-1}\mv{x}_{1:n} )/\sqrt{n} ] 	&&\to\, k_v, 		\label{eq:var_est}\\
	&\mv{c}_{0,1:n} \mv{C}_n^{-1} \mv{x}_{1:n}   			&&\overset{p}{\to}\, K_1, 	\label{eq:kriging_mean}\\
	&c_{0}- \mv{c}^{\trsp}_{0,1:n} \mv{C}_n^{-1} \mv{c}_{0,1:n} 	&&\to\, k_2, 			\label{eq:kriging_var}
	\end{alignat}
	where $K_0\geq 0$ and $K_1$ are random variables, $k_2 \in [0, c_0]$, and $k_v>0$.
\end{assumption}
The first two assumptions are satisfied for all models considered in this article given they have finite moments, and given that the sequence $\{\mv{s}_i\}$ does not result in a singular the covariance matrix (which for example is the case if $\mv{s}_i = \mv{s}_j$ for $i\neq j$). The last two assumptions assure that the linear predictor converges to a constant given the data. Assuming that $x$ has a Mat\'ern covariance function with $\nu<\infty$, this is also fulfilled as long as the sequence $\{\mv{s}_i\}$ is not chosen so that the covariance is degenerate. Given these assumptions, we have the following result. 

\begin{Theorem}\label{thm:krig}
	Let Assumption \ref{ass1} hold and assume that $\pi(v)$ is a bounded function which is absolutely continuous with respect to the Lebesgue measure, such that $\E[v]=1$. Then
	$\pi_{G_1,x_0}(\cdot|\mv{x}_{1:n})  \overset{p}{\to} N(\cdot; k_1, k_0 k_2)$ as $n \rightarrow \infty$.
	Here $k_0$ and $k_1$ are the realisations of $K_0$ and $K_1$ in \eqref{eq:mean_est} and \eqref{eq:kriging_mean} respectively.
\end{Theorem}
The theorem shows that the predictive distribution for a type G$_1$ model converges to a Gaussian distribution, and thus the predictor (the mean of the distribution) converges to the corresponding predictor for a Gaussian model, under quite general assumptions on the distribution for the data. In particular, it holds if the data comes from a type G$_1$ model. 

\begin{cor}\label{eq:Lemma}
	Let $x(\mv{s}), \mv{s}\in\mathbb{R}^d$, be a univariate type G$_1$ Mat\'ern-SPDE field let and $\mv{s}_0, \ldots, \mv{s}_n$ be locations in $\mathbb{R}^d$ such that  $i<\|\mv{s}_0-\mv{s}_i\|<i+1$ for $i=1\ldots, n$.  Assume that $\pi(v)$ is a bounded function and absolutely continuous with respect to the Lebesgue measure, such that $\E[v]=1$. 
	Then the predictive distribution for $x(s_0)$, $\pi_{x(s_0)}(\cdot|x(\mv{s}_1),\ldots,x(\mv{s}_n))$, converges in probability to a Gaussian distribution as $n\rightarrow \infty$.
\end{cor}

\section{Normal inverse Gaussian fields}\label{sec:nig}
The NIG distribution \citep{barndorff1997normal} is obtained by choosing $p=1$ and $v$ as an inverse gamma (IG) random variable in \eqref{eq:variancemixture}.
The IG distribution has density 
$$
IG(v;\nignu_1,\nignu_2) = \frac{\sqrt{\nignu_2} }{ \sqrt{2\pi v^3}} \exp\left( - \frac{\nignu_1}{2} v - \frac{\nignu_2}{2v} + \sqrt{\nignu_1\nignu_2} \right), \quad \nignu_1,\nignu_2>0.
$$
The resulting density for the NIG variable is 
\begin{align*}
NIG(x;\gamma,\mu,\nignu_1,\nignu_2) &=  \frac{ e^{\sqrt{\nignu_1\nignu_2} +\mu(x-  \gamma)}\sqrt{ \nignu_2\mu^2 + \nignu_1 \nignu_2}}{\pi \sqrt{ \nignu_2 + (x -  \gamma)^2}}  K_1\left( \sqrt{(\nignu_2  + (x -  \gamma)^2) (\mu^2 + \nignu_1)}  \right).
\end{align*}
In this form the $NIG$ density is overparameterized, and we therefore typically set $\nignu_1=\nignu_2=\nignu$ which results in $\E(v)=1$. If $\mu=0$, one has that the random variable defined in \eqref{eq:variancemixture} has variance 1, but for $\mu\neq 0$, the variance depends on $\eta$.
An important property of the NIG distribution is that its variance mixture distribution, the IG distribution, is closed under convolution. This simplifies inference as explained in later sections.

The simplest multivariate NIG Mat\'ern-SPDE field is obtained by using the type G$_1$ construction with $v\sim IG(\nignu,\nignu)$, resulting in a field with multivariate NIG marginal distributions. To construct the more flexible type G$_3$ and type G$_4$ models, we use IG random variables in the univariate type G L\'evy processes, which results in NIG processes. Since the NIG distribution has both the Gaussian the Cauchy distributions as limiting cases (as $\eta \rightarrow \infty$ and $\eta \rightarrow 0$ with suitable scalings of the other parameters), the NIG Mat\'ern-SPDE processes have both a Gaussian process and a L\'evy flight process as limiting cases. When using NIG noise in \eqref{eq:general}, it is convenient to note that the noise can be represented by an independently scattered random measure \citep{rajput1989spectral}. Specifically, for any Borel set $A$ in the domain, the measure is a univariate NIG random variable with probability density function
$f_{\dot{\mathcal{N}}(A)}(x) =   NIG(x;m(A)\gamma,\mu,\nignu,m(A)^2\nignu)$,
where $m(A)$ denotes the Lebesgue measure of $A$. Note that a random variable with density $f_{\dot{\mathcal{N}}(A)}(x)$ can be obtained from equation \eqref{eq:variancemixture} where $v\sim IG( \nignu,m(A)^2 \nignu )$ and thus $\E(v)=m(A)$.

We let $\mv{\mathcal{N}}_3$ and $\mv{\mathcal{N}}_4$ denote the vector-valued processes in Definition \ref{def:noise} when univariate NIG processes are used.  The density of $\mv{\proc}(\mv{s})$ in \eqref{eq:general} does not have an explicit form in this case but  one can derive the characteristic function (CF) of $\mv{\proc}(\mv{s})$. The following proposition provides the CF for the type G$_4$ case.

\begin{prop}\label{thm:charf}
	The CF of a stationary solution $\mv{\proc}$ to \eqref{eq:general}, evaluated at $\mv{s}$, where the driving noise is $\mv{\mathcal{N}}_4$, is $\phi_{\mv{\proc}(\mv{s})}(\mv{u}) = \prod_{k=1}^p \phi_k(\mv{u})$ where
	\begin{equation*}
	\phi_k(\mv{u}) = \exp\left[  - i\gamma_k  \int  \mv{u}^{\trsp} \mv{v}_{k,\mv{t}} d\mv{t} +    \sqrt{\nignu_k} \int  \nignu_k - \sqrt{\nignu_k   -2i \mu_k^2\mv{u}^{\trsp} \mv{v}_{k,\mv{t}}   +  (\mv{u}^{\trsp} \mv{v}_{k,\mv{t}})^2  } d\mv{t}\right].
	\end{equation*}
	Here  $\mv{v}_{k,\mv{t}} =  [\BMi_{1k} G_1(\mv{s},\mv{t}), \BMi_{2k} G_2(\mv{s},\mv{t}), \ldots, \BMi_{pk} G_p(\mv{s},\mv{t}) ]^{\trsp}$, $\mv{\BMi} = \mv{\BM}^{-1}$, and
	$$
	G_k(\mv{s},\mv{t})  = \frac{\Gamma\left(\frac{\alpha_k-d}{2}\right)}{c_k (4\pi)^{d/4}\Gamma(\frac{\alpha_k}{2})\kappa_k^{\alpha_k-d}} \materncorr{\|\mv{s}-\mv{t}\|}{\kappa_k}{\frac{\alpha_k-d}{2}}, \quad k=1,\ldots, p.
	$$
\end{prop}
The following example illustrates the effect of the shape parameter $\theta$ on the multivariate marginal distributions of the type G$_4$ model. 

\begin{figure}[t]
	\begin{center}
		\begin{minipage}[t]{0.25\linewidth}
			\begin{center}
				$\theta=0$\phantom{$\rho$}
				\includegraphics[width=\textwidth]{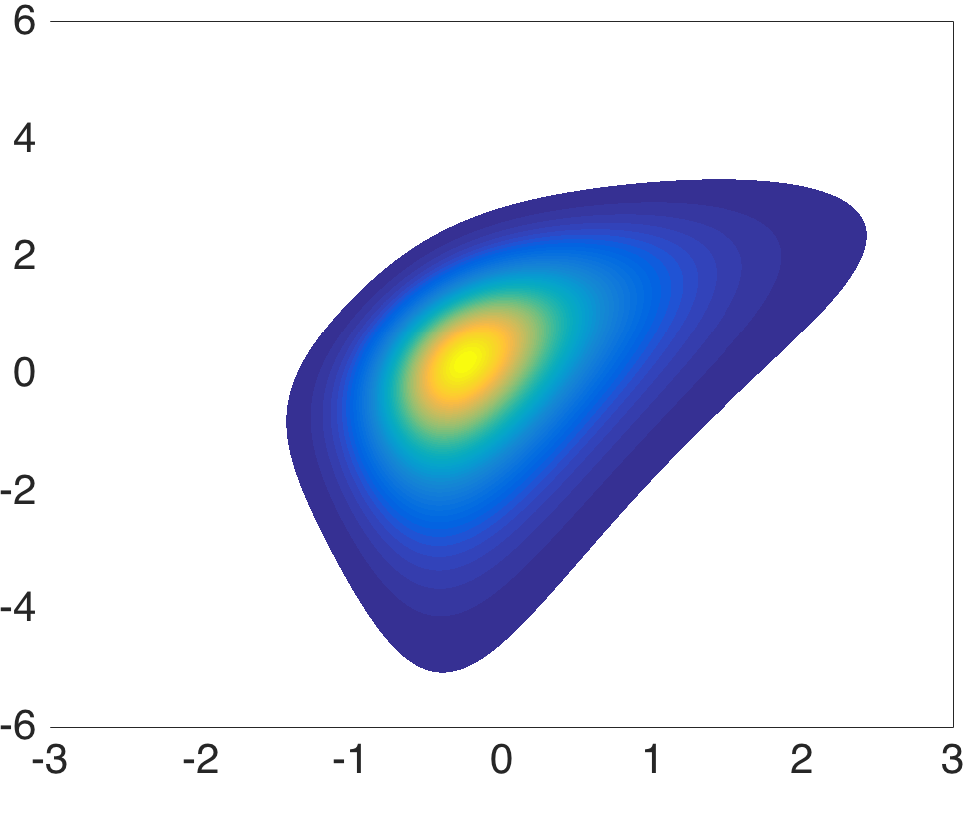}
			\end{center}
		\end{minipage}
		\begin{minipage}[t]{0.25\linewidth}
			\begin{center}
				$\theta=\arctan(\rho)$
				\includegraphics[width=\textwidth]{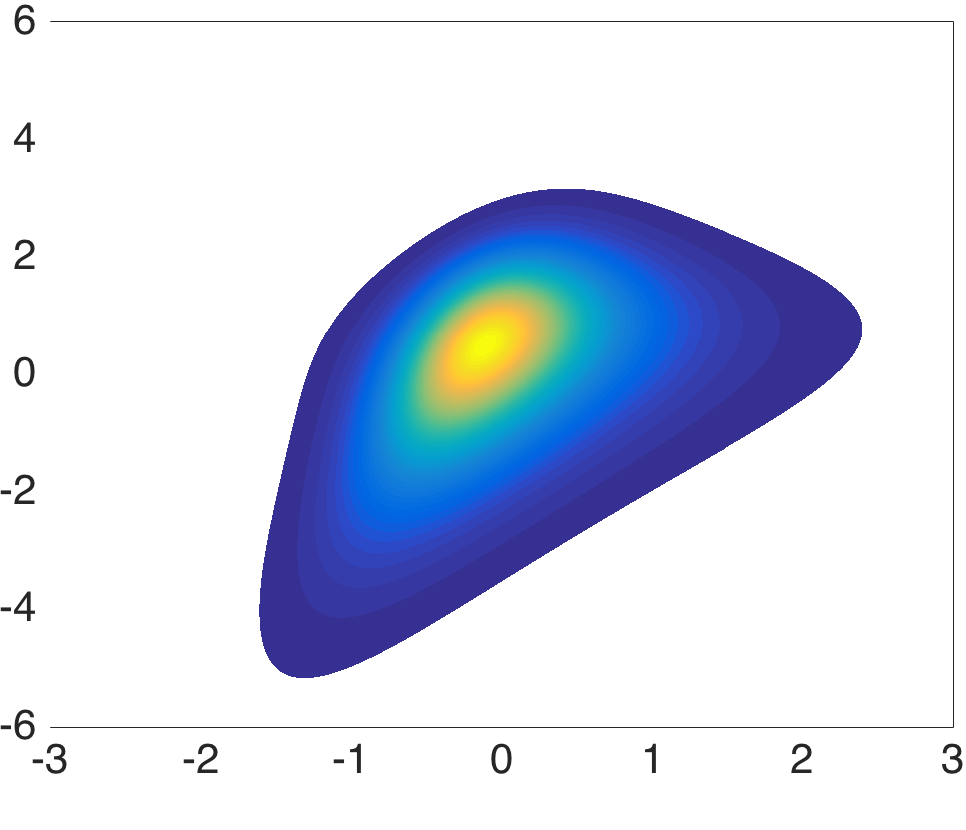}
			\end{center}
		\end{minipage}
		\begin{minipage}[t]{0.25\linewidth}
			\begin{center}
				$\theta=\pi/2$
				\includegraphics[width=\textwidth]{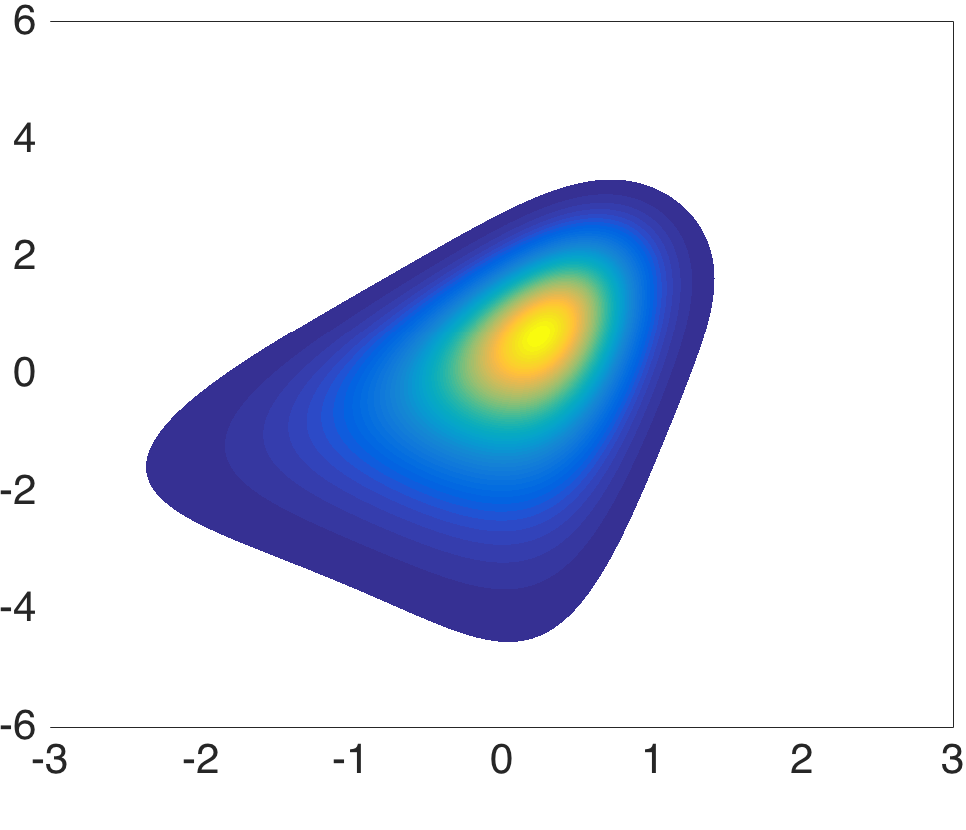}
			\end{center}
		\end{minipage}
		\begin{minipage}[t]{0.25\linewidth}
			\begin{center}
				$\theta=\pi$\phantom{$\rho$}
				\includegraphics[width=\textwidth]{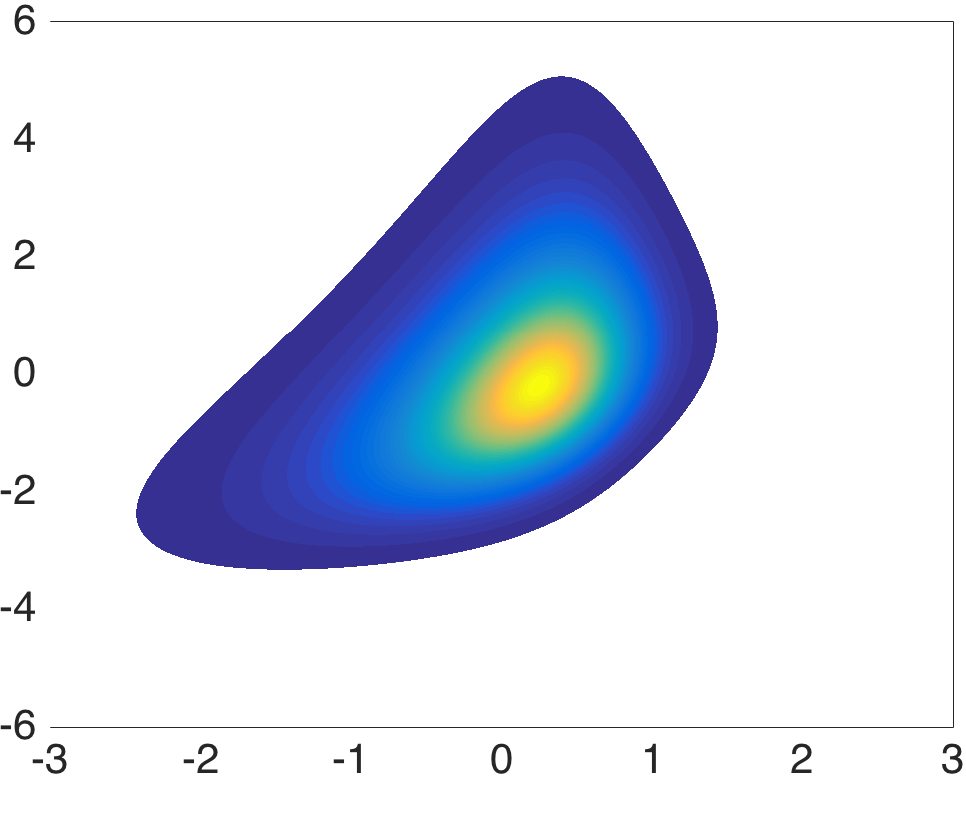}
			\end{center}
		\end{minipage}
		\begin{minipage}[t]{0.25\linewidth}
			\begin{center}
				$\theta=\pi+\arctan(\rho)$
				\includegraphics[width=\textwidth]{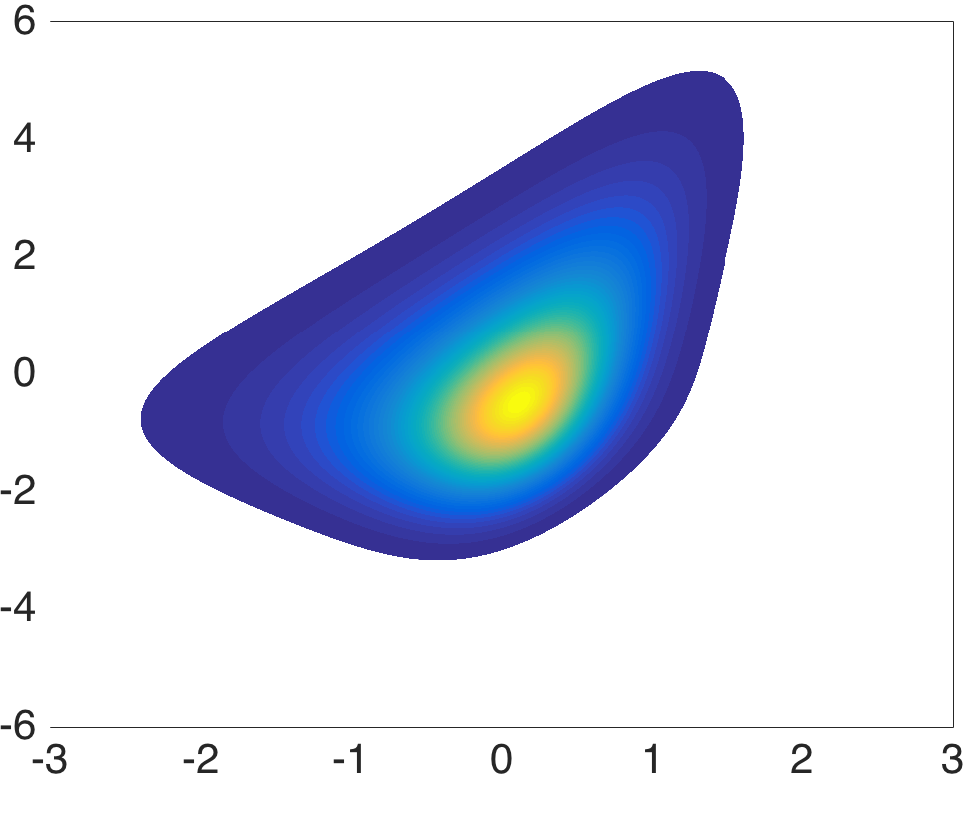}
			\end{center}
		\end{minipage}
		\begin{minipage}[t]{0.25\linewidth}
			\begin{center}
				$\theta=3\pi/2$
				\includegraphics[width=\textwidth]{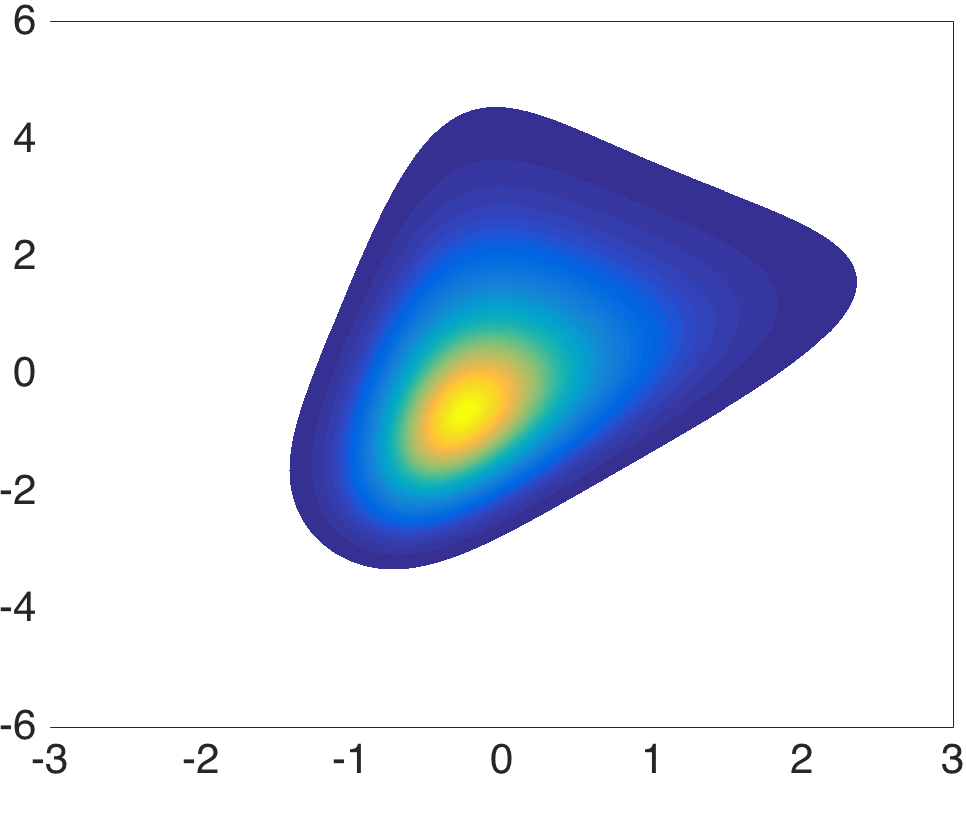}
			\end{center}
		\end{minipage}
	\end{center}
	\caption{Marginal distributions of a bivariate NIG Mat\'ern-SPDE field for different values of $\theta$. 
		All six cases have the correlation function shown in Figure \ref{fig1}.}
	\label{fig:simmarg}
\end{figure}

\begin{example}
	Let $\mv{x}(\mv{s})$ be a type G$_4$ bivariate NIG Mat\'ern-SPDE field with the same parameters as in Figure \ref{fig1}. For the driving noise, we let $\mu_1 = \gamma_2 = 1$, $\mu_2 = \gamma_1 = -1$, and $\nignu = 0.9$. Figure \ref{fig:simmarg} shows bivariate marginal distributions of the resulting field for different values of $\theta$ in the dependence matrix \eqref{eq:matern2}, computed using Proposition \ref{thm:charf}. Recall that $\rho$ determines the cross-correlations between $x_1(\mv{s})$ and $x_2(\mv{s})$ whereas $\theta$ determines the shape of the bivariate marginal distributions, but does not affect the covariance function. Thus, all six examples have the same correlation function, which is shown in Figure \ref{fig1}.  The case $\theta = 0$ corresponds to a lower-triangular operator matrix, and $\theta = \arctan(\rho)$ corresponds to an upper-triangular operator matrix. 
\end{example}

As discussed in Section \ref{sec:nongauss}, the simpler type G constructions have similar flexibility of the marginal distributions, but  lower flexibility in terms of conditional distributions. 
The following example illustrates how different the predictive distributions can be.

\begin{figure}[t]
	\begin{center}
		\begin{minipage}[t]{0.32\linewidth}
			\begin{center}
				\includegraphics[width=\textwidth]{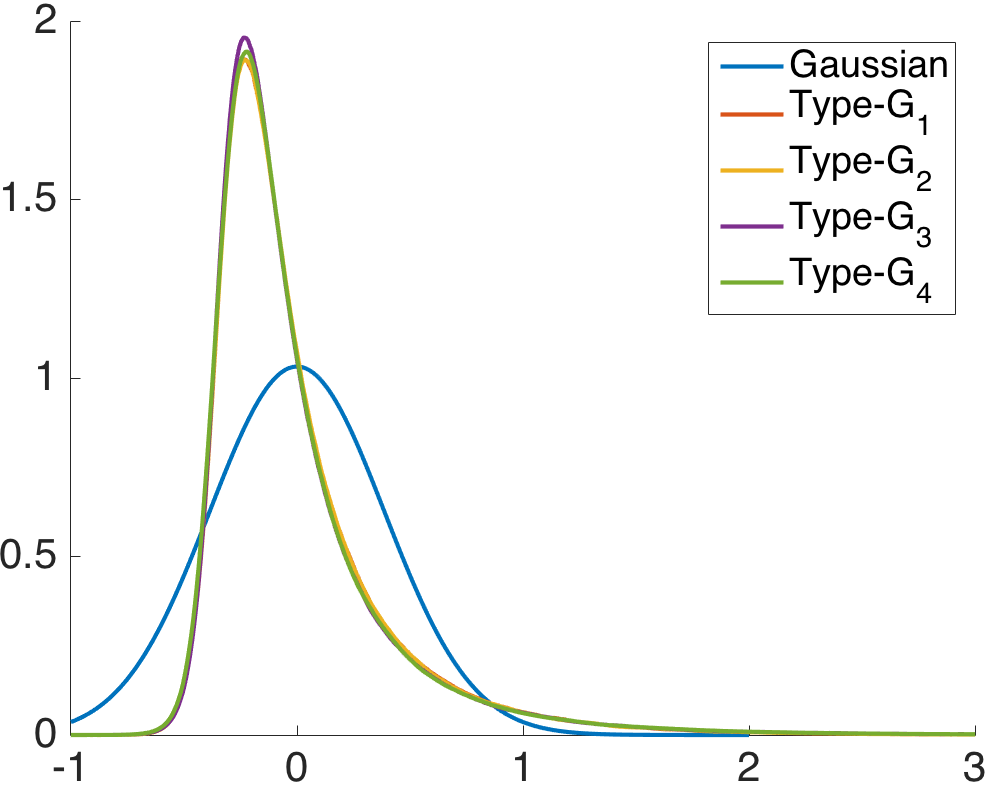}
				(a) $\pi(x_1(0))$
				\includegraphics[width=\textwidth]{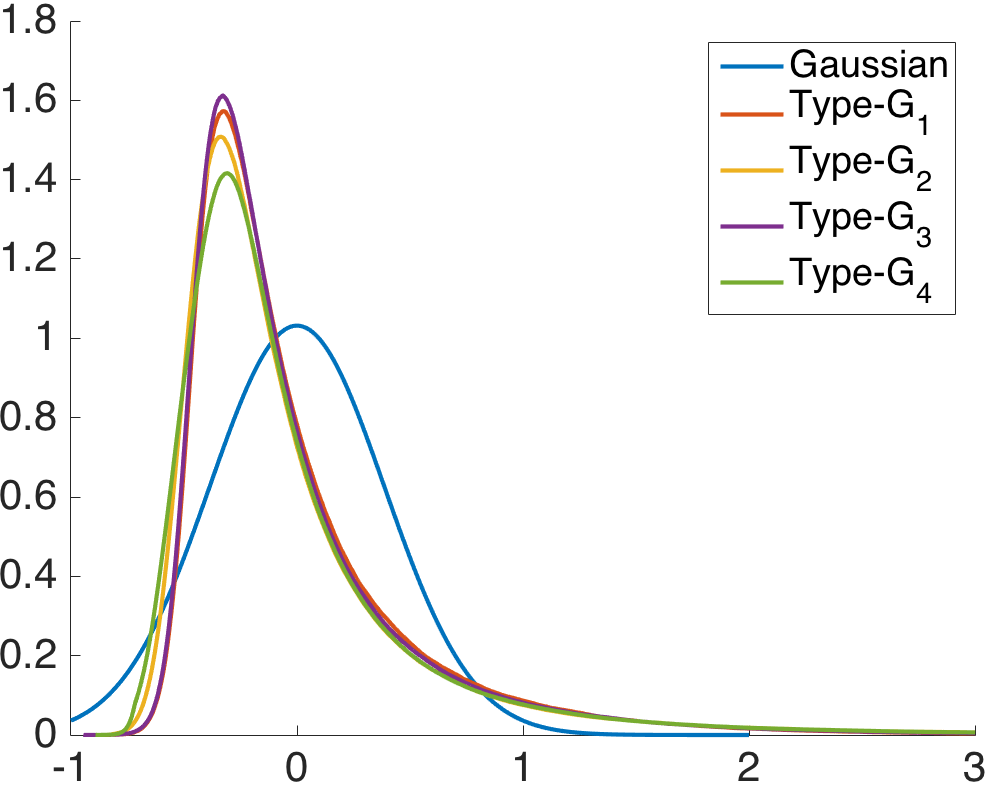}
				(d) $\pi(x_2(0))$
			\end{center}
		\end{minipage}
		\begin{minipage}[t]{0.32\linewidth}
			\begin{center}
				\includegraphics[width=\textwidth]{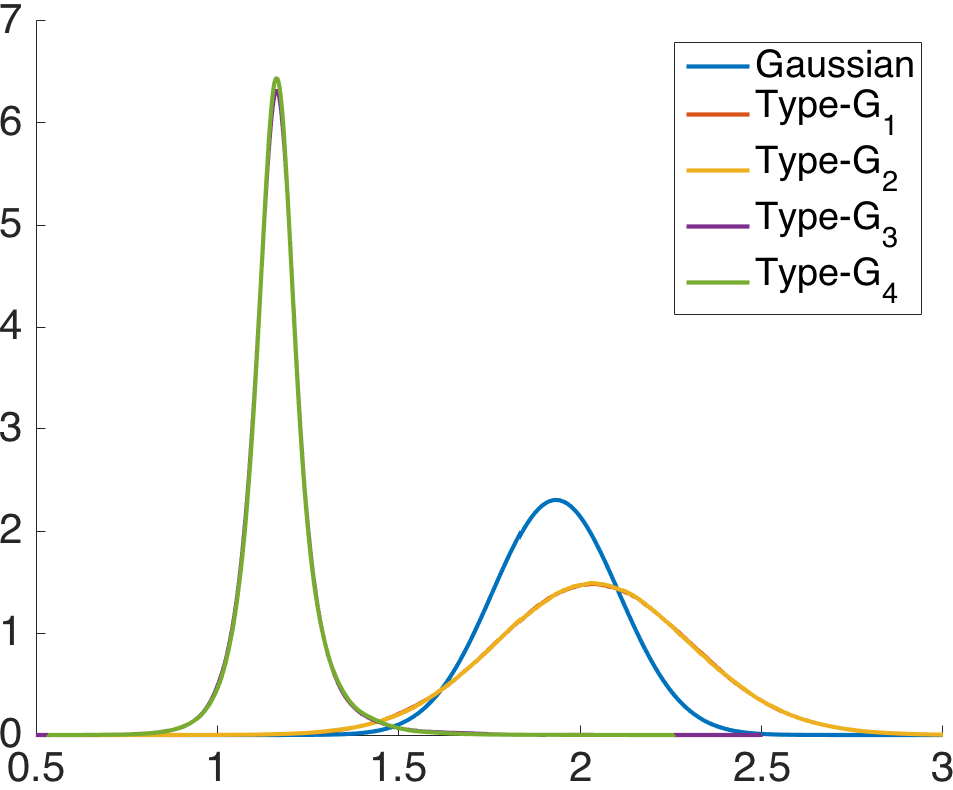}
				(b) $\pi(x_1(0)|\mv{y}),  \sigma_e = 10^{-3}$
				\includegraphics[width=\textwidth]{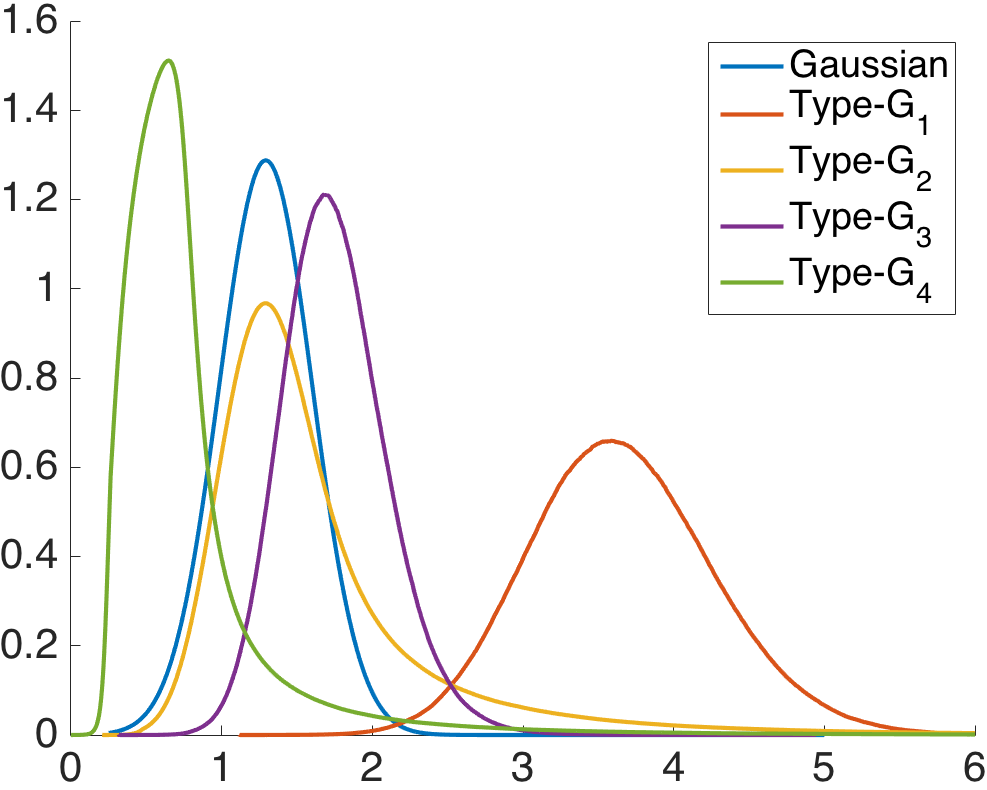}
				(e) $\pi(x_2(0)|\mv{y}), \sigma_e = 10^{-3}$
			\end{center}
		\end{minipage}
		\begin{minipage}[t]{0.32\linewidth}
			\begin{center}
				\includegraphics[width=\textwidth]{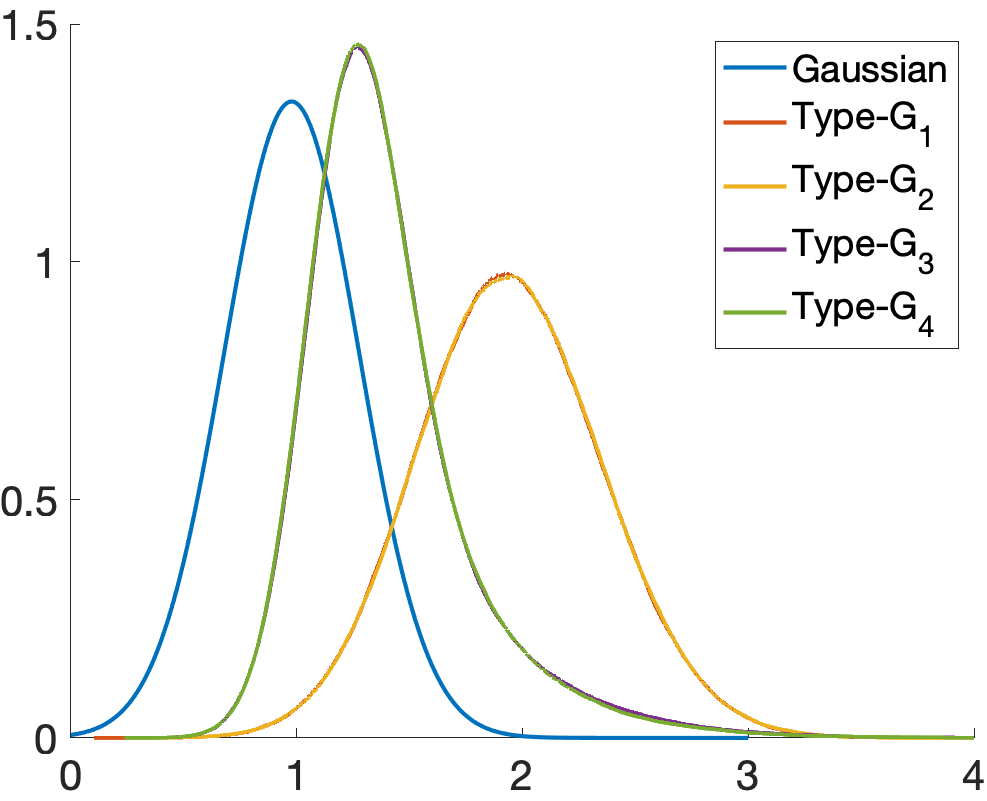}
				(c) $\pi(x_1(0)|\mv{y}), \sigma_e = 0.5$
				\includegraphics[width=\textwidth]{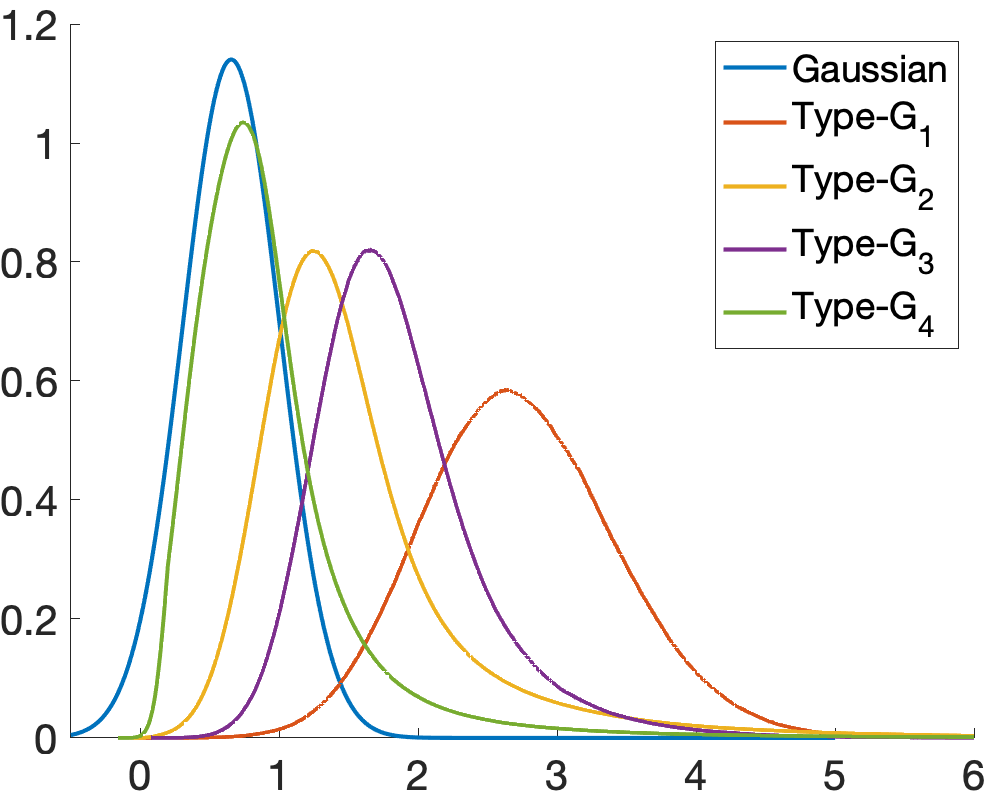}
				(f) $\pi(x_2(0)|\mv{y}), \sigma_e = 0.5$
			\end{center}
		\end{minipage}
	\end{center}
	\caption{Marginal distributions for $\mv{x}(0)$ based on the five bivariate models in Example \ref{example2}. Here $\mv{y} = \{x_1(-1) = 0, x_1(1) = 4\}$ and $\sigma_e$ denotes the measurement noise standard deviation.} 
	\label{fig:preddist}
\end{figure}

\begin{example}\label{example2}
	Let $\mv{x}_i(t)$, $i=1,\ldots,4$, be bivariate type G$_i$ NIG Mat\'ern-SPDE processes on $\R$ with $\alpha=2$, $\kappa = 1$, $\sigma = 0.1$, $\rho = 0.9$, and $\theta = 0$. The processes have the same operator matrix $\mv{\mathcal{K}}$ and we choose the parameters $\mv{\mu}$ and $\mv{\nignu}$ so that they have similar (univariate) marginal distributions, see Figure \ref{fig:preddist}, Panels (a) and (d), where the marginal distributions of a corresponding Gaussian process also is shown for reference. We predict the value of the four processes at $t=0$ based on two observations of the first dimension $y_{1} = x_{i,1}(-1) +\vep_{1} = 0$ and $y_2 = x_{i,1}(1) + \vep_{2} = 4$, where $\vep_{-1}$ and $\vep_{1}$ are independent $\pN(0,0.001^2)$ variables representing measurement noise. How the prediction is done is presented in Section \ref{sec:kriging}. The predictive distributions are shown in Panels (b) and (e). Even though the four processes have similar marginal distributions for $\mv{x}_i(0)$, their predictive  distributions are very different. For the prediction of the first dimension, the type G$_1$ and type G$_2$ processes have similar distributions, which is expected since they have the same marginal structures. The type G$_3$ and type G$_4$ also have equal marginal structures and therefore similar predictions, which are very different from the first two. For the prediction of the second dimension, we get different predictions for all models since they have different cross-dependence structures. In particular we can note the counter intuitive type G$_1$ prediction, where the prediction of the second dimension is larger than the first, even though there are no observations for this dimension. The same predictive distributions in the case when $\vep_{-1}$ and $\vep_1$ instead have variance $0.5^2$ are shown in Panels~(c) and (f), and one can note the same behaviour of the predictions for this case.
\end{example}

\section{Geostatistical modeling and estimation}\label{sec:model}
To use the multivariate type G fields for geostatistical applications, we need  to be able to include them in hierarchical models that include covariates and measurement noise. In this section, we formulate such a model and describe how to perform likelihood-based estimation of the model parameters and how to use the model for spatial prediction. 

We consider a standard geostatistical model where a latent field is specified using covariates for the mean, and the data consists of noisy observations of this latent field at some locations $\mv{s}_1, \ldots, \mv{s}_n$. 
Let $\data_{ki}$ be the $i$th observation of the $k$th dimension, defiened as $y_{ki} = \sum_{j=1}^K B_{kj}(\mv{s}_i)\beta_j + \proc_k(\mv{s}_i) + \vep_{ki}$ for $k=1,\ldots,p$,  
where the independent variables $\vep_{ki}\sim \pN(0,\sigma_{e,k}^2)$ represent the measurement noise. The functions $B_j(\mv{s})$ are covariates for the mean  and $\proc_k(\mv{s})$ is the $k$th variable of a mean-zero multivariate type G Mat\'ern-SPDE field $\mv{\proc}(\mv{s})$. Since the mean of $\mv{y}$ is modeled using covariates, 
we assume that the mixing variables in the type G construction are scaled so that they have unit expectation (if the expected value exists) and set $\gamma_k = -\mu_k$ to guarantee that  $\proc_k(\mv{s})$ has mean zero in the case that it has an expected value. 

Assuming that the smoothness parameters satisfy $\alpha_i/2\in\mathbb{N}$ for $i=1,\ldots,p$ and using the finite-dimensional representation of $\mv{\proc}(\mv{s})$ derived in Appendix \ref{sec:fem}, we have $\mv{\proc}(\mv{s}) = \sum_{j=1}^n\sum_{k=1}^p w_{jk}\mv{\varphi}_j^k(\mv{s})$. Here $\mv{\varphi}_j^k(\mv{s}) = \varphi_j(s)\mv{e}_k$ are p-dimensional basis functions, where $\mv{e}_k$ is the $k$th column in a $p \times p$ identity matrix, $\{\varphi_i\}$ are piecewise linear FE basis functions defined by a mesh on $\mathcal{D}$, and $\{w_{jk}\}$ are stochastic weights. The model is then
\begin{equation}\label{eq:mainmodel}
\begin{split}
\mv{v} &\sim \pi(\mv{v}), \\
\mv{\weight} | \mv{\var} &\sim  \pN \left( \mv{K}^{-1}(\mv{\mu}\otimes\mv{I}_n)(\mv{\var}-\mv{h}), \mv{K}^{-1} \diag(\mv{\var}) \mv{K}^{-\trsp}  \right),
\\
\mv{\data}_{k}|\mv{\weight} &\sim  \pN\left(  \mv{B} \mv{\beta} +  \mv{A}_k\mv{\weight}, \sigma_{e,k}^2\mv{I}\right), \, k = 1, \ldots,p,
\end{split}
\end{equation}
where $\mv{\data}_k$ denotes the vector of all $n$ observations of the $k$th dimension of the data, $\mv{w}$ is a vector with all stochastic weights, and $\mv{K}$ is a discretization of the operator matrix. The matrix $\mv{B}$ contains the covariates evaluated at the measurement locations and $\mv{A}_k = \diag(\mv{e}_k) \otimes \mv{A}$ where $\mv{A}$ is an observation matrix with elements $\mv{A}_{ij} = \varphi_{j}(\mv{s}_i)$. Finally, the distribution of the variance components, $\pi(\mv{v})$, depends on which model that is used, as described in Appendix~\ref{sec:fem}.

\subsection{Parameter estimation}\label{sec:estimation}
As is standard in the SPDE approach, we assume that the smoothness parameters are fixed and known. It should be noted that models with general smoothness parameters likely could be estimated from data using the rational SPDE approach \citep{bolin2017rational}. However, we leave the adaptation of this approach to the multivariate typee-G setting for future research.

Let $\mv{y} = (\mv{y}_1^\trsp,\ldots, \mv{y}_p^\trsp)^\trsp$ denote the vector of all observations in  \eqref{eq:mainmodel}, and let $\mv{\Psi}$ be the model parameters to be estimated.
There is no explicit expression for the likelihood distribution $\pi(\mv{\data}|\mv{\Psi})$. However, it is possible to compute maximum likelihood parameter estimates using Monte Carlo (MC) methods. This is computationally feasible because of two important properties of the model: Firstly, $\mv{w}|\mv{y},\mv{v},\mv{\Psi}$ is a Gaussian Markov random field (GMRF) and can thus be sampled efficiently. Secondly, $\mv{v}|\mv{w},\mv{y},\mv{\Psi}$ is a vector of independent variables and can thus be sampled in parallel. 

We use a stochastic gradient (SG) method \citep{kushner2003stochastic} to estimate the parameters. The idea of SG is that one only needs an asymptotically unbiased estimator (as the number of MC samples goes to infinity), $\mv{G}(\mv{\Psi})$, of the gradient of the likelihood in order to utilize an iterative procedure where one at iteration $i$ updates the parameters as $\mv{\Psi}^{(i)} = \lambda_{i} \mv{G}(\mv{\Psi}^{(i-1)} ) + \mv{\Psi}^{(i-1)}$. Here $\{\lambda_i\}$ is a sequence satisfying $\sum \lambda_{i} \rightarrow \infty$ and $\sum  \lambda_{i}^2 < \infty$, which ensures that the method converges to a stationary point of the likelihood \citep{kushner2003stochastic,andrieu2005stability}. 
To derive  the estimator of the gradient, we use Fisher's identity \citep{dempster1977maximum} to obtain
\begin{align}
\nabla_{\mv{\Psi}} \log \pi(\mv{\data}|\mv{\Psi}) &= \pE_{\mv{\var},\mv{\weight}}\left(\nabla_{\mv{\Psi}} \log \pi(\mv{\var},\mv{\weight}|\mv{\data},\mv{\Psi}) | \mv{\data},\mv{\Psi}\right) 
=\E_{ \mv{\var}} (\nabla_{\mv{\Psi}}\log \pi( \mv{\var}| \mv{\data},\mv{\Psi})| \mv{\data},\mv{\Psi}). \label{eq:Egrad}
\end{align}

Since $\pi(\mv{\weight}|\mv{\var},\mv{y},\mv{\Psi})$ is Gaussian, we have a closed-form expression for  $\nabla_{\mv{\Psi}}\log \pi( \mv{\var}| \mv{\data},\mv{\Psi})$, see the Appendix \ref{sec:gradients}, but there is no closed form expression for its expected value. We therefore use
$
\mv{G}(\mv{\Psi}) = \frac1{N}\sum_{i=1}^N \nabla_{\mv{\Psi}}\log\pi_{\mv{\Psi}}( \mv{\var}^{(i)} | \mv{\data},\mv{\Psi})
$
as a MC estimate of the expectation, where $\mv{\var}^{(i)}$ are samples from distribution $\pi(\mv{\var}|\mv{\data},\mv{\Psi})$. These samples are obtained using a Gibbs sampler (Algorithm \ref{alg:Gibbs} in Appendix \ref{sec:pseudo}) which samples $\pi(\mv{w}|\mv{y},\mv{v},\mv{\Psi})$ and $\pi(\mv{v}|\mv{w},\mv{y},\mv{\Psi})$ respectively. The sampling of $\pi(\mv{v}|\mv{w},\mv{y},\mv{\Psi})$ typically needs to be done with a general sampling method, such as a Metropolis Hastings algorithm. However, if $\pi(v)$ is a generalized inverse Gaussian (GIG) distribution, then the conditional distribution remains in the GIG family which can be sampled uniformly fast over the entire parameter space, see \cite{hormann2014generating}. The GIG distribution has density
$GIG(v; c, a, b) = \left(\frac{a}{b}\right)^{\frac{c}{2}}(2 K_c(\sqrt{ab}))^{-1} v^{c-1} e^{-\frac{1}{2}\left(av + bv^{-1}\right)}.$
For further details, including parameter ranges, see \cite{Jorgensen}.
The GIG distribution contains several known distributions as special cases, such as the gamma distribution, the inverse gamma distribution, and the IG distribution. Because of this, one can sample the variance components of the NIG distribution explicitly. The following example provides the conditional distributions for the NIG Mat\'ern-SPDE fields from Section \ref{sec:nig}.
\begin{example}
	For the NIG processes in Section \ref{sec:nig}, the distribution of the variance components $v$, $v_i$ and $v_k$ is  $IG(v;\nignu_1, \nignu_2) = GIG(v;-\frac{1}{2}, \nignu_1 , \nignu_2 )$. It can therefore be shown that the different type G constructions result in the following posterior distributions 
	\begin{align*}
	&\mbox{type G$_1$:} &\pi(v|\mv{E},\mv{\Psi}) &= \mbox{GIG}\left(v;-\frac{np+1}{2}, \nignu+ \sum_{k=1}^p \mu_k^2\mv{1}_n ^\trsp\mv{h}_k  , \nignu + \sum_{k=1}^p \left(\frac{\mv{\xi}_k}{\mv{h_k}}\right)^\trsp \mv{\xi}_k \right),\\
	&\mbox{type G$_2$:} &\pi(v_k|\mv{E},\mv{\Psi}) &= \mbox{GIG}\left(v_k;-\frac{n+1}{2},\nignu_k+  \mu_k^2\mv{1}_n ^\trsp\mv{h}_k  , \nignu_k + \left(\frac{\mv{\xi}_k}{\mv{h_k}}\right)^\trsp \mv{\xi}_k\right),\\
	&\mbox{type G$_3$:} &\pi(\mv{v}|\mv{E},\mv{\Psi}) &= \mbox{GIG}\left(\mv{v}; -\frac{p+1}{2} , \nignu  +  \sum_{k=1}^p \mu_k^2 , \mv{h}_k^2 \nignu + \sum_{k=1}^p \mv{\xi}_k^2 \right),\\
	&\mbox{type G$_4$:} &\pi(\mv{v}_k|\mv{E},\mv{\Psi}) &= \mbox{GIG}\left(\mv{v};-1, \mu_k^2 + \nignu_k , \mv{\xi}_k^2 + \mv{h}_k^2 \nignu_k\right),
	\end{align*}
	where $\mv{E}=[\mv{E}_1^{\trsp},\ldots,\mv{E}_p^{\trsp}]^{\trsp} = \mv{K} \mv{\weight}$ and $\mv{\xi}_k = \mv{E}_k + \mv{h}_k\mu_k$.
	For the two last densities it is explicitly understood that $GIG$ in vector form denotes product of independent GIG distributions with parameter values given by the values in the vectors. 
\end{example}



\subsection{Spatial prediction and evaluation of predictive performance}\label{sec:kriging}
In applications one is often interested in predictions of the latent field given data. The predictive distribution for the $k$th variable of the latent field, at a location $\mv{s}_0$, is $\pi(\proc_k(\mv{s}_0)|\mv{\data},\mv{\Psi})$. This distribution is often summarized using the mean as a point estimate, and the variance as a measure of uncertainty. To estimate these two quantities, let $\mv{A}_p = [\varphi_1(\mv{s}_0),\ldots, \varphi_n(\mv{s}_0)]$ and use the Gibbs sampler in Algorithm \ref{alg:Gibbs}, Appendix \ref{sec:pseudo}, to obtain samples, $\{\mv{v}^{i}\}_{i=1}^N$, from $\pi(\mv{v}|\mv{\data},\mv{\Psi})$. Based on these samples, we compute MC estimates $\pE(\proc_k(\mv{s}_0)|\mv{\data}) \approx \frac1{N} \sum_{i=1}^N \mv{A}_p\hat{\mv{\xi}}^{(i)}$ and 
$\pV(\proc_k(\mv{s}_0)|\mv{\data})\approx \frac1{N} \sum_{i=1}^N \mv{A}_p^{\trsp}(\hat{\mv{Q}}^{(i)})^{-1}\mv{A}_p$, where $\hat{\mv{\xi}}^{(i)}$ is the expected value of $\mv{w}|\data, \mv{v}^{(i)}$ and $\hat{\mv{Q}}^{(i)}$ is the corresponding precision matrix (see Appendix \ref{sec:gradients} for analytic formulas of these quantities).
The posterior median, which may be a more appropriate point estimator if the distribution is asymmetric, can similarly be estimated by the sample median of $\{\mv{A}_p\hat{\mv{\xi}}^{(i)}\}_{i=1}^N$. 

To evaluate a proposed model one also need to compute various goodness-of-fit measures, such as the  continuous ranked probability scores (CRPS) \citep{matheson1976scoring}. Let $y_k$ be an observation in the $k$th dimension at $\mv{s}_0$, and let $F$ denote the marginal CDF of  $\pi(y_k|\mv{\data}_{-0},\mv{\Psi})$, where $\mv{\data}_{-0}$ denotes all observations but $y_k$, then the (negatively oriented) CRPS value for this location can be computed as \citep{gneiting2007strictly}
\begin{equation}\label{eq:crpsE}
\mbox{CRPS}(F,y_k) = \pE(|Y^{(1)}_k - y_k|) - \frac1{2}\pE(|Y^{(1)}_k - Y^{(2)}_k|)
\end{equation}
where $Y^{(1)}_k$ and $Y^{(2)}_k$ are independent random variables with distribution $F$. For a Gaussian distribution this expression can be used to derive CRPS value analytically \citep[see e.g.][]{gneiting2007strictly}. For the multivariate type G SPDE-Mat\'ern fields, one option is to approximate the expected values in \eqref{eq:crpsE} by MC integration. Basing such an estimate on $N$ draws of $Y_k^{(1)}$ and $Y_k^{(2)}$ yields an estimate $\mbox{CRPS}_N(F,y)$. Unfortunately, $N$ often needs to be quite large to obtain good approximations with this estimator. The following proposition provides a more efficient way of approximating the CRPS value in the case of a general normal-variance mixture distribution.

\begin{prop}\label{crpsthm}
	Assume that the random variable $X$ is a normal-variance mixture with CDF $F(x) = \int \Phi\left(\frac{x-\mu(v)}{\sigma(v)}\right) \md F_v(v)$. Let $V_j^{(i)}, j=1,2, i=1,\ldots, N$ be independent draws from the mixing distribution $F_v$, and define $\mu_V = \pE(X|V)$, $\sigma_V^2 = \pV(X|V)$, and
	\begin{equation}\label{eq:M}
	M(\mu,\sigma^2) = 2\sigma\varphi\left(\frac{\mu}{\sigma}\right) + \mu\left(2\Phi\left(\frac{\mu}{\sigma}\right)-1\right),
	\end{equation}
	where $\varphi$ denotes the density function of a standard Gaussian distribution. Then
	\begin{align*}
	\mbox{CRPS}_N^{RB}(F,y)  &= \frac1{N}\sum_{i=1}^N \left[M\left(\mu_{V_1^{(i)}} - y,\sigma_{V_1^{(i)}}^2\right) 
	- \frac1{2}M\left(\mu_{V_1^{(i)}} - \mu_{V_2^{(i)}},\sigma_{V_1^{(i)}}^2 +\sigma_{V_2^{(i)}}^2\right)\right]
	\end{align*}
	satisfies $\pE(\mbox{CRPS}_N^{RB}(F,y))  = \mbox{CRPS}(F,y)$ and $\pV(\mbox{CRPS}_N^{RB}(F,y)) \leq \pV(\mbox{CRPS}_N(F,y))$.
\end{prop}
The $\mbox{CRPS}_N^{RB}$ estimator can be used for the type G fields since $\pi(\proc_k(\mv{s}_0)|\mv{\data}_{-0},\mv{v},\mv{\Psi})$ is Gaussian and since we easily can sample the variances $\mv{v}$ using the Gibbs sampler. 

To give an idea of the improvement that can be obtained by using the RB estimator, both estimators were used to compute the CRPS value for the final fold of the most general NIG model in the cross-validation study in Section~\ref{sec:applications}. Based on $N=10000$ samples, the MC variances of the two estimators were $\pV(\sqrt{N}\mbox{CRPS}_N^{RB}(F,y)) \approx 187$ and  $\pV(\sqrt{N}\mbox{CRPS}_N(F,y)) \approx 2225$. 

\section{Applications}\label{sec:applications}
In this section we illustrate for two different data sets how the multivariate Type G SPDE fields can be used for spatial modelling. The first data set consists of temperature and pressure measurements from the North American Pacific Northwest and was previously studied in \cite{gneiting2012matern} and  \cite{apanasovich2012valid}. The second data set consists of seawater temperature measurements taken at two different depths in the ocean.

For both data sets we assume the model $\mv{\data}_i = \mv{\beta} + \mv{\proc}(\mv{s}_i) + \mv{\vep}_i$ for the bivariate observations $\mv{y}_i$, where $\mv{\proc}(\mv{s}) = (\proc_1(\mv{s}), \proc_2(\mv{s}))^{\trsp}$ is a mean-zero random field,  $\mv{\beta} = (\beta_1,\beta_2)^{\trsp}$ is the expected value, and $\mv{\vep}_i$ are independent $\pN(\mv{0},\diag(\sigma_{1}^2,\sigma_{2}^2))$ variables representing measurement noise.  As reference models, we will for each data set use four Gaussian models for $\mv{x}(\mv{s})$. The first of these assumes that $x_1(\mv{s})$ and $x_2(\mv{s})$ are independent Gaussian Mat\'ern fields with covariance functions $C_{11}(\mv{h}) = \sigma_1^2\materncorr{\mv{h}}{\kappa_1}{\nu_1}$ and  $C_{22}(\mv{h}) = \sigma_2^2\materncorr{\mv{h}}{\kappa_2}{\nu_2}$ respectively. We also use the parsimonious Gaussian Mat\'ern field by \cite{gneiting2012matern} as well as two Gaussian Mat\'ern-SPDE models specified using \eqref{eq:general}, one lower-triangular and one independent model with $\rho= 0$.  For the applications, we focus on comparing the reference models to the type G$_4$ models and do not evaluate the simpler type G constructions. We do not consider the type G$_1$ and type G$_2$ models since the first data set does not have repeated measurements, and since one does not expect these models to improve the predictive performance compared to the Gaussian models because of Theorem~\ref{thm:krig}. We do not consider the type G$_3$ model since there is no specific reason for why a shared variance component would be beneficial for the considered data sets.

\begin{figure}[t]
	\begin{center}
		\includegraphics[width=0.4\textwidth]{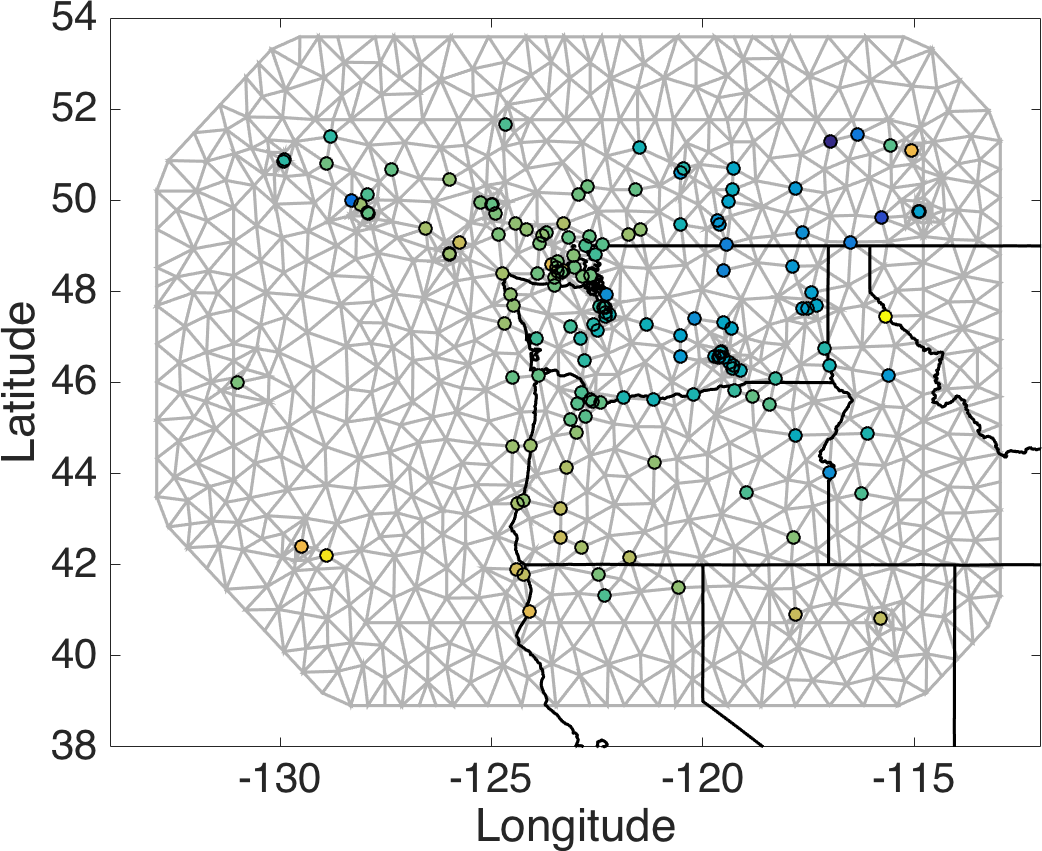}
		\raisebox{0.14\height}{\includegraphics[height=0.28\textwidth]{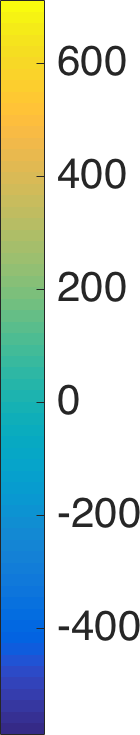}}
		\includegraphics[width=0.4\textwidth]{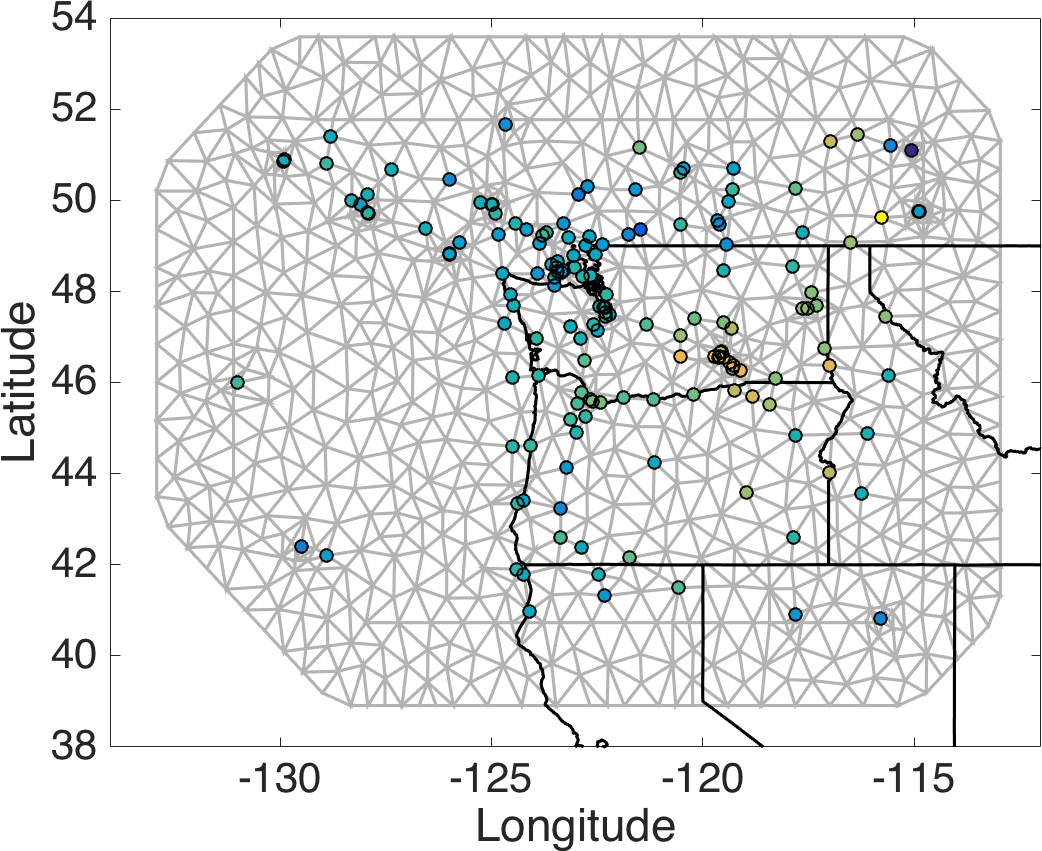}
		\raisebox{0.14\height}{\includegraphics[height=0.28\textwidth]{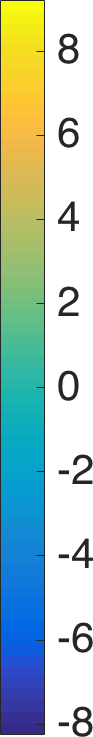}}
	\end{center}
	\caption{Measurements of pressure (left) and temperature (right) in the North American Pacific Northwest together with the mesh used for the SPDE models. The sample mean has been subtracted from the data in both cases.}
	\label{fig:data}
\end{figure}

\subsection{Temperature and pressure in the North American Pacific Northwest}
The data, shown in Figure \ref{fig:data}, consists of temperature and pressure observations, $\mv{\data}_i = (\data_{P}, \data_{T})_i^{\trsp}$ where $\data_{P}$ denotes pressure and $\data_{T}$  temperature, at 157 locations in the North American Pacific Northwest. 
Besides the four baseline models, we test four different type G Mat\'ern-SPDE models for the data. A Gaussian model for temperature seems adequate whereas the pressure data has short-range variations that is inflating the measurement noise variance (see parameter estimates in Appendix \ref{seq:parameter_estimates}), which possibly could be captured by the latent field if a non-Gaussian model was used. We therefore consider type G$_4$ models where the driving noise for the pressure is NIG distributed, whereas the driving noise for temperature is Gaussian. In order to investigate the effects of the operator matrix, we use one independent model, with $\rho = 0$, and two dependent models. The first of these is triangular with $\theta = 0$, and the second has a general operator where $\theta$ is estimated jointly with the other parameters. 

The mesh that is used for the discretization of the SPDE models is shown in Figure~\ref{fig:data}. It consists of  981 nodes and was built using \texttt{R-INLA} \citep{lindgren2015bayesian}. We fix the $\alpha$ parameters to $2$ for all SPDE models, which corresponds to $\nu = 1$ for the Mat\'ern covariances. 
The parameters of the Gaussian models are estimated using numerical maximisation of the log-likelihood function, whereas the gradient-based method from Section \ref{sec:estimation} is used for the non-Gaussian models. The gradient method is run 1000 iterations, using starting values obtained from the corresponding Gaussian model.  For the lower-triangular models, the estimation took $44$ seconds for the Gaussian model and $156$ seconds for the NIG model. These values were obtained using a \cite{Matlab} implementation of the algorithm on a Macbook Pro computer with a 2.6GHz Intel Core i7 processor. The parameter estimates for the different models are shown in Appendix~\ref{seq:parameter_estimates}. 

To compare the models, we perform a leave-one-out pseudo cross-validation study. For each observation location, the pressure and temperature values are predicted using the data from all 156 other locations using the models with parameters given in Appendix~\ref{seq:parameter_estimates}. For all models, the point estimates are computed using the expected values of the values at the held-out location conditionally on the data at all other locations. Using the posterior median as a predictor did not improve the predictive performance for this data, and we therefore omit those results. The predictive performance of the models is assessed using the median absolute error of the 157 predicted values, as well as the median CRPS. The resulting values are shown in Table~\ref{tab2}. One can note that the dependent NIG models have better predictive performance than the Gaussian models. Spatial predictions using the parsimonious Mat\'ern model and the general NIG model can be seen in Figure~\ref{fig:krig}.

\begin{table}
\centering
	\begin{tabular}{lcccccc}
		\toprule
		& Operator & Number of & \multicolumn{2}{c}{Pressure} & \multicolumn{2}{c}{Temperature} \\
		Model & matrix &  parameters & \multicolumn{2}{c}{(Pascal)} & \multicolumn{2}{c}{(degrees Celcius)} \\
		\cmidrule(r){1-3} \cmidrule(r){4-5} \cmidrule(r){6-7}
		& &    & MAE & CRPS & MAE & CRPS \\
		Independent Mat\'ern & -		& 10 & $41.632$ 		& $28.994$ 		& $0.956$ 	& $0.598$ \\
		Parsimonious Mat\'ern & -	& 10 & $39.068$ 		& $27.682$ 		& $0.921$ 	& $0.576$ \\
		Gaussian SPDE & Diagonal  	& 8 	& $38.624$ 		& $31.711$ 		& $0.917$ 	& $0.594$\\
		Gaussian SPDE & Triangular 		& 9 	& $38.856$ 		& $31.829$ 		& $0.915$ 	& $0.580$\\
		NIG SPDE & Diagonal  	& 10 &$39.101$ 		& $25.993$ 		& $0.847$ 	& $0.525$\\
		NIG SPDE & Triangular 		& 11 & $39.302$ 		& $25.776$ 		& $\mv{0.841}$ & $\mv{0.512}$\\
		NIG SPDE & General 		&12 	& $\mv{38.523}$ 	& $\mv{25.591}$ 	& $0.876$ 	& $0.514$\\
		\bottomrule
	\end{tabular}
	\caption{\label{tab2} Cross-validation results comparing the median absolute error (MAE) and median CRPS for the different models.} 
\end{table}

\begin{figure}[t]
	\begin{center}
		\begin{minipage}[t]{0.3\linewidth}
			\begin{center}
				Parsimonious Mat\'ern \\[2mm]
				\includegraphics[width=0.85\linewidth]{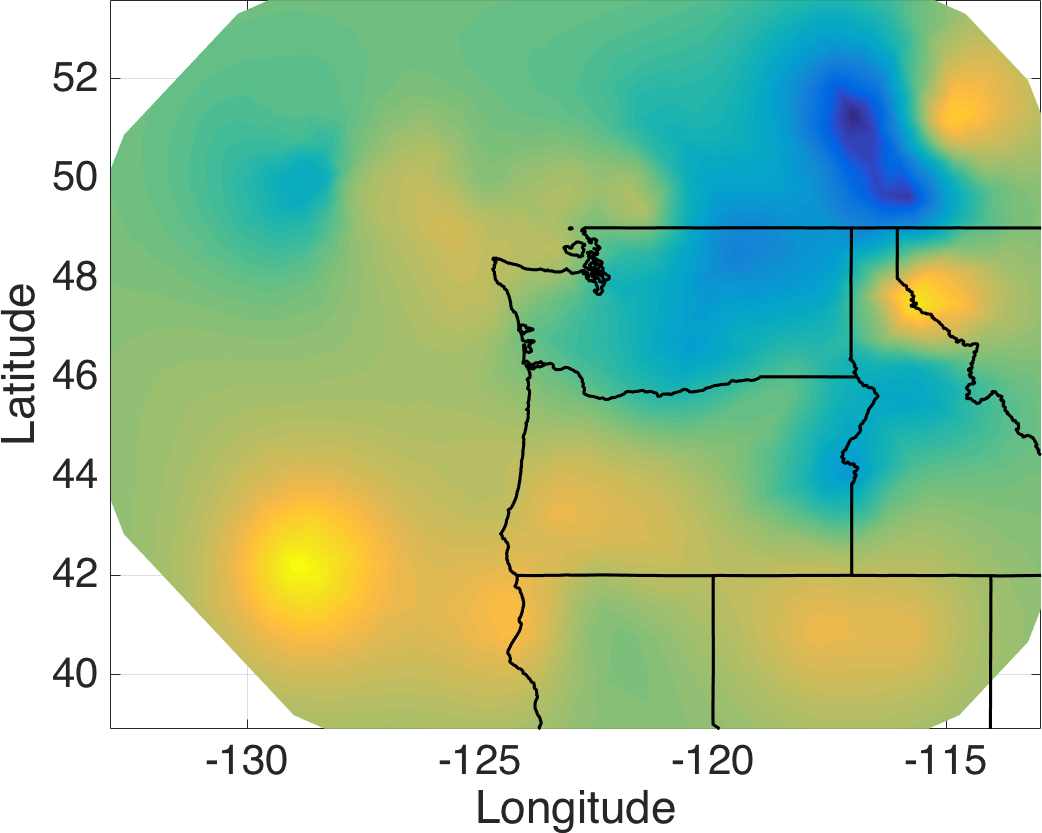}
				\raisebox{0.1\height}{\includegraphics[height=0.65\textwidth]{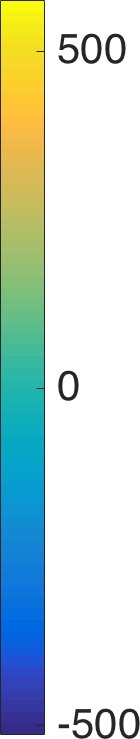}}\\[2mm]
				\includegraphics[width=0.85\linewidth]{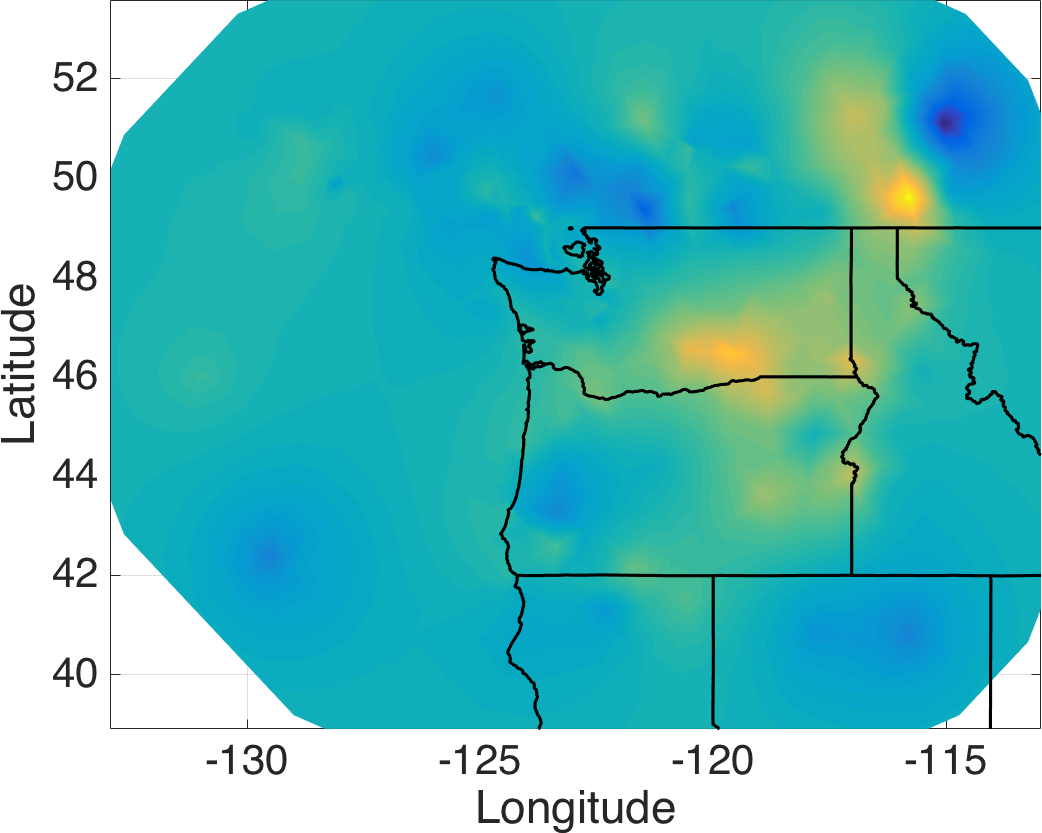}
				\raisebox{0.1\height}{\includegraphics[height=0.65\textwidth]{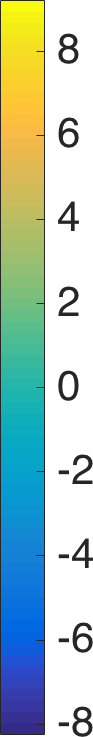}}
			\end{center}
		\end{minipage}
		\begin{minipage}[t]{0.3\linewidth}
			\begin{center}
				NIG General \\[2mm]
				\includegraphics[width=0.85\linewidth]{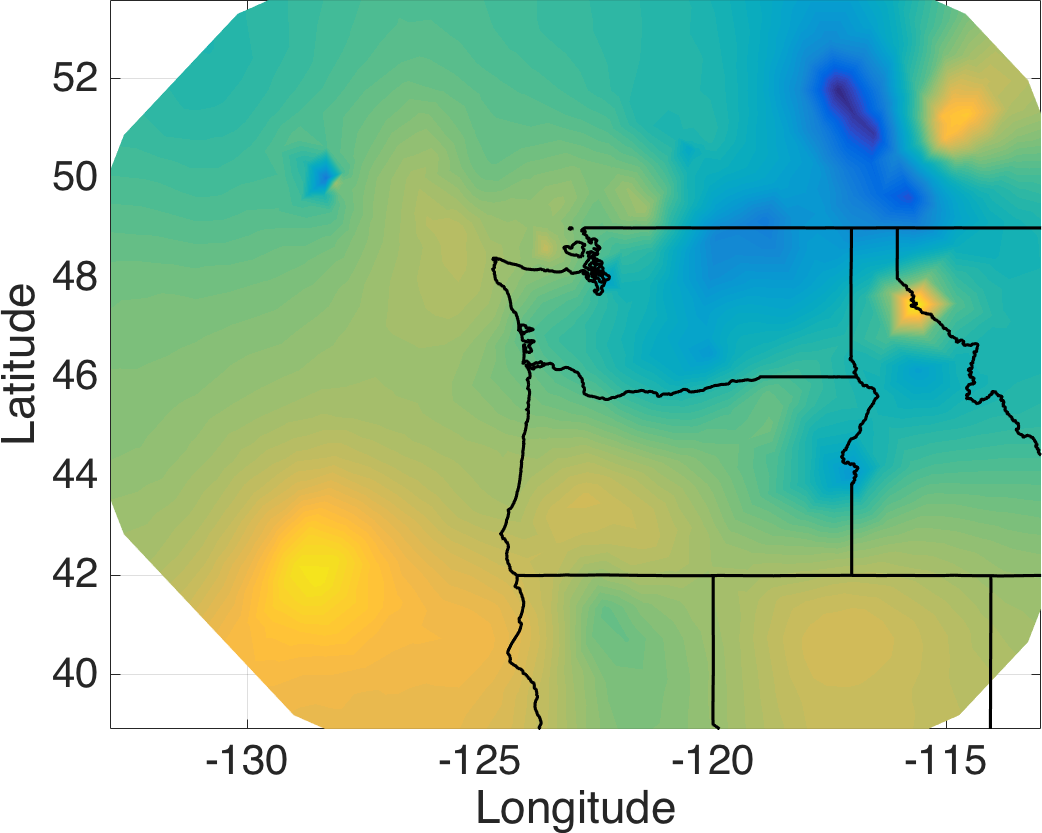}
				\raisebox{0.1\height}{\includegraphics[height=0.65\textwidth]{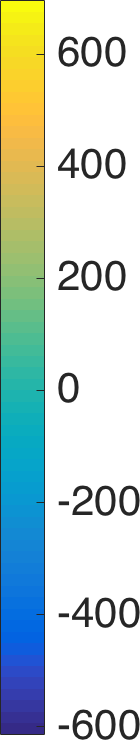}}\\[2mm]
				\includegraphics[width=0.85\linewidth]{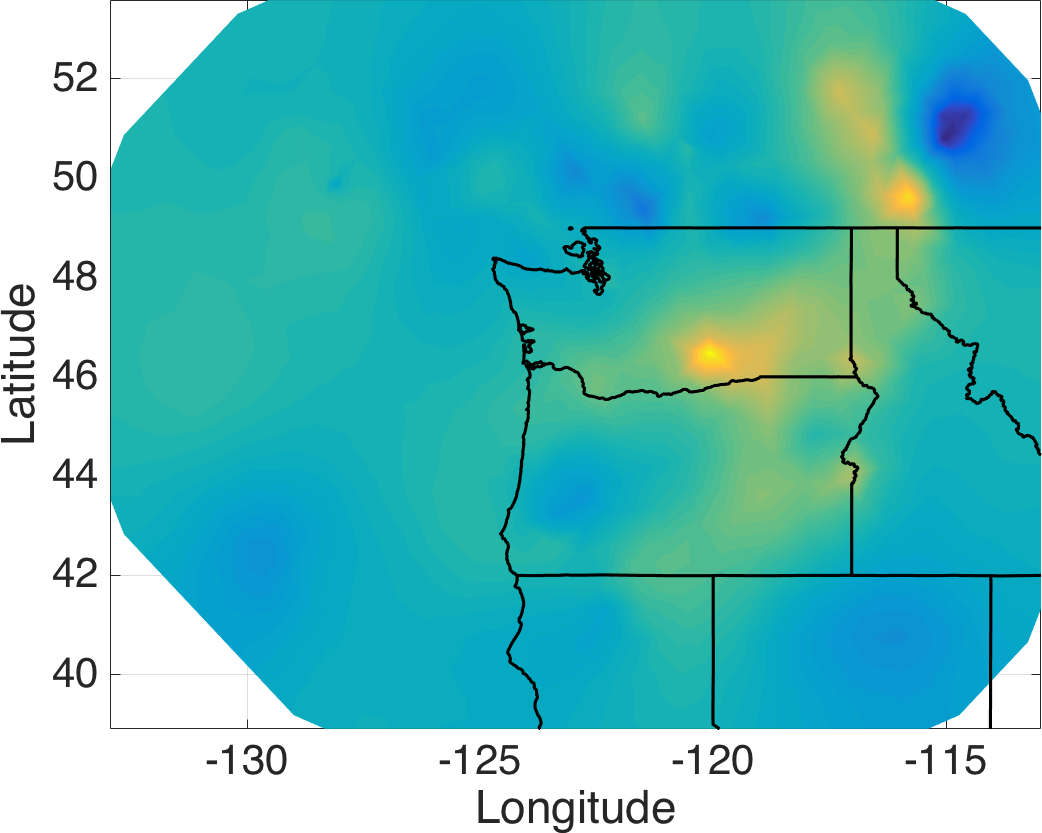}
				\raisebox{0.1\height}{\includegraphics[height=0.65\textwidth]{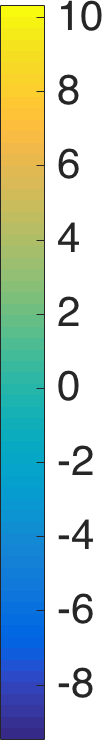}}
			\end{center}
		\end{minipage}
		\begin{minipage}[t]{0.3\linewidth}
			\begin{center}
				Difference \\[2mm]
				\includegraphics[width=0.85\linewidth]{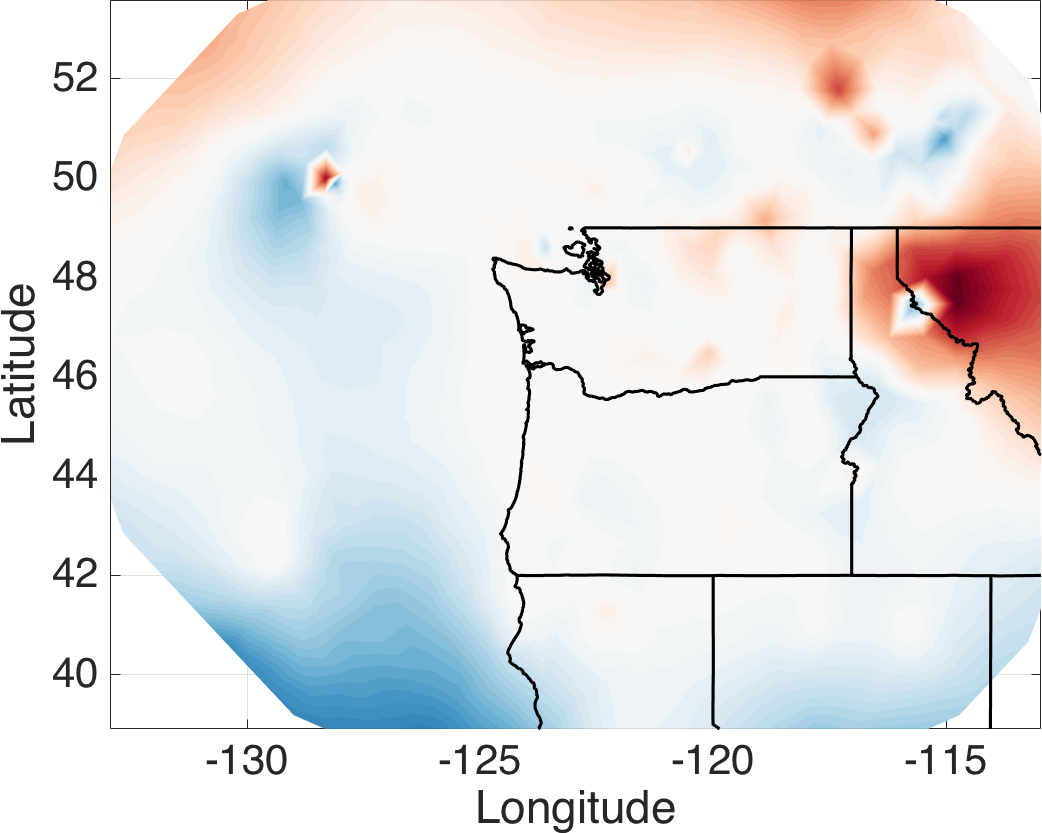}
				\raisebox{0.1\height}{\includegraphics[height=0.65\textwidth]{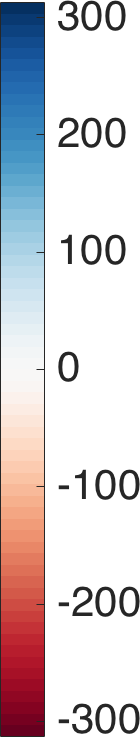}} \\[2mm]
				\includegraphics[width=0.85\linewidth]{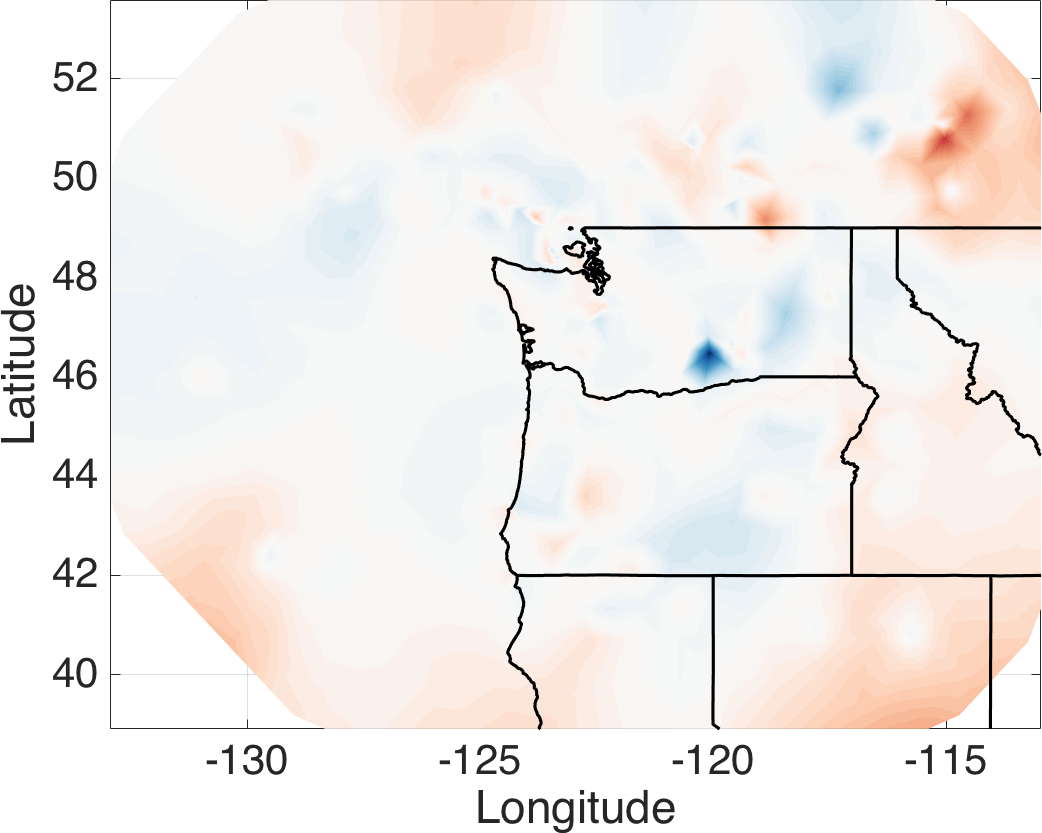}
				\raisebox{0.1\height}{\includegraphics[height=0.65\textwidth]{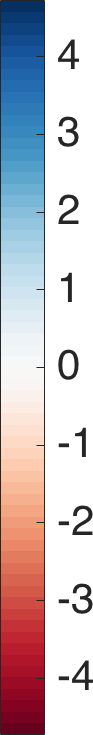}}
			\end{center}
		\end{minipage}
	\end{center}
	\caption{Estimates of pressure (top) and temperature (bottom) using the parsimonious Mat\'ern model and the NIG model with general operator matrix. The difference between the estimates is shown to the right.}
	\label{fig:krig}
\end{figure}
\subsection{Seawater temperatures}
We now consider Argo floats measurements of seawater temperature at two different depths. Since the measurements are sparse in space (and time) an important statistical task is to ``fill in the gaps'' through spatial interpolation. The data has been thoroughly analyzed from a statistical perspective by \cite{kuusela2018locally}, who noted that the data seem to be non-Gaussian at higher depths in certain areas. To investigate if the type G models could be useful for interpolation of this data, we choose two different depths, $300$ dbar and $1500$ dbar, and investigate if one could improve the joint prediction of those depths using the type G models. 

\begin{figure}[t]
	\begin{center}
		\includegraphics[width=0.5\textwidth]{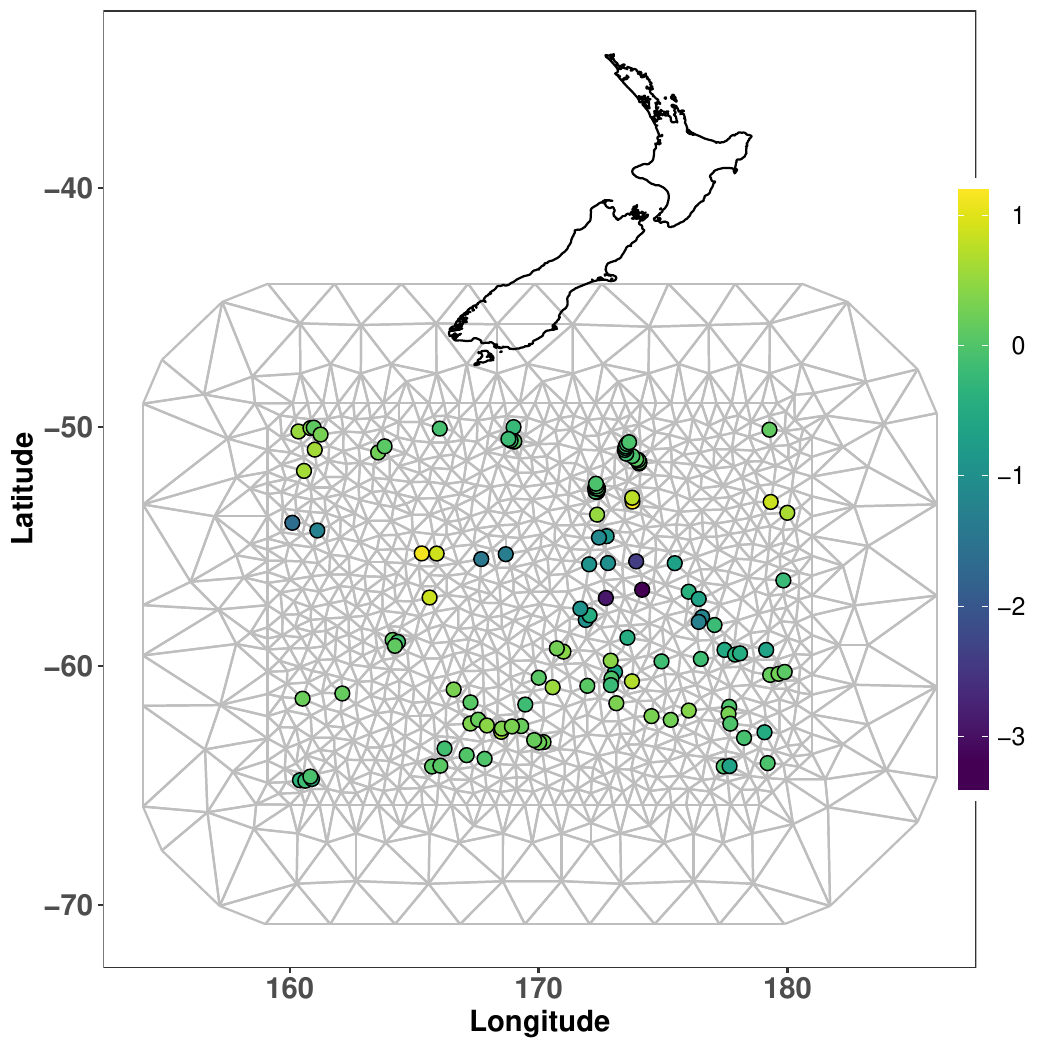}
	\end{center}
	\caption{Measurements of seawater temperature on depth $300$ db in 2016 together with the mesh used for the SPDE models. A mean field has been removed from the measurements.}
	\label{fig:dataArgo}
\end{figure}

We extract data from the month of February for three years ($2014-2016$). Since \cite{kuusela2018locally} showed that an analysis of the complete data set requires a non-stationary model, we focus on a limited spatial region south of New Zeeland to be able to use a stationary model. However, as for the Gaussian SPDE-based models, one could model non-stationarity by allowing the parameters in the operator to be spatially varying. The restriction results in a data set consisting of $312$ observations in total.
For each location we study the residuals after removing a seasonally varying mean field (the Roemmich-Gilson mean field, see \cite{kuusela2018locally}). Thus, we let $\mv{y}_i=\left(y_{300}(\mv{s}_i),y_{1500}(\mv{s}_i)\right)^{\trsp}$ for $i\in \{1,\ldots,n_t\}$ and $t\in \{2014,2015,2016\}$. Here $y_{300}$ is the residual at the depth of $300$ dbar and $y_{1500}$ is the residual at the depth of $1500$ dbar. We assume that data for the different years are independent. 

Besides the four baseline models, we test two non-Gaussian type G SPDE models. In the first, we assume that  $x_1(\mv{s})$ and $x_2(\mv{s})$ are independent univariate NIG SPDE fields. In the second, we use the type G$_4$ construction with a general operator matrix where both noise processes are NIG  distributed. For all SPDE-based models, we again fix the smoothness parameters $\alpha$ to $2$ and use the mesh shown in Figure~\ref{fig:dataArgo}, which also shows the available data for the year $2016$. 
The parameters are estimated using an R implementation of the proposed methods, available in the package \texttt{ngme}. The parameter estimates for the different models are shown in Appendix~\ref{seq:parameter_estimates}. 

To evaluate which of the tested models that performs best in terms of prediction, we again use leave-one-out pseudo cross-validation and compare median MAE and median CRPS. However, contrary to the previous application we now do the cross-validation by removing individual univariate observations instead of removing the bivariate observation pairs for each spatial location. The reason for this is that we here do not necessarily have observations of the fields at different depth in the same spatial locations. The results are shown in Table~\ref{tab:argoCrossval}, where we see that the multivariate Non-Gaussian model is clearly outperforming the other models. Thus, it seems as if one could increase the accuracy of the spatial interpolation of the Argo data using type G models. For future work it is therefore interesting to study the entire data set, where a more in-depth analysis would require a space-time model, see \cite{kuusela2018locally}.

\begin{table} 
	\centering
	\begin{tabular}{lcccccc}
		\toprule
		& Operator & Number of & \multicolumn{4}{c}{depth} \\
		Model & matrix &  parameters & \multicolumn{2}{c}{db $300$} & \multicolumn{2}{c}{db $1500$} \\
		\cmidrule(r){1-3} \cmidrule(r){4-5} \cmidrule(r){6-7}
		& &    & MAE & CRPS & MAE & CRPS \\
		Independent Mat\'ern & -		& 10 & $0.183$ 		& $0.154$ 		& $0.040$ 	& $0.029$ \\
		Parsimonious Mat\'ern & -		& 10 & $0.239$ 		& $0.162$ 		& $0.033$ 	& $\mv{0.024}$ \\
		Gaussian SPDE & Diagonal  	& 8 	& $\mv{0.187}$ 		& $0.182$ 		& $0.043$ 	& $0.032$\\
		Gaussian SPDE & Triangular  	& 9 	& $0.216$ 		& $0.169$ 		& $0.034$ 	& $0.025$\\
		NIG SPDE & Diagonal  	& 12 &$0.200$ 		& $0.159$ 		& $0.040$ 	& $0.028$\\
		NIG SPDE & General	&14 	& $0.201$ 	& $\mv{0.130}$ 	& $\mv{0.032}$ 	& $\mv{0.024}$\\
		\bottomrule
	\end{tabular}
\caption{\label{tab:argoCrossval} Cross-validation results comparing the median absolute error and median CRPS for different models for the Argo data. }
\end{table}

\section{Discussion}\label{sec:discussion}
There is a need for practically useful random field models with more general distributions than the Gaussian. Especially for multivariate data, finding good alternatives to Gaussian fields has been considered an open problem in the literature. We have introduced one such alternative by formulating a new class of multivariate random fields with flexible multivariate marginal distributions and covariance functions of Mat\'ern-type. The fields are constructed as solutions to SPDEs and can be used in a geostatistical setting where likelihood-based parameter estimation can be performed using a computationally efficient stochastic gradient algorithm. In fact, the models have the same computational advantages as their Gaussian counterparts, which facilitates applications to large data sets, although with additional cost due to MC sampling. 

Four different constructions of the non-Gaussian noise were considered, where the first two are closely related to existing approaches, such as factor-copula models and Student's t-fields. We showed that these constructions have significant disadvantages when used for spatial prediction, or on data without replicates. The more sophisticated constructions based on type G L\'evy noise do not have these disadvantages, and their combination of flexibility and computational efficiency should therefore make them attractive alternatives to Gaussian models for geostatistical applications.

The computational benefits of the finite dimensional approximations presented in Appendix \ref{sec:fem} are only available for fields with $\alpha/2\in\mathbb{N}$. This restriction of the smoothness parameters is often viewed as one of the main drawbacks of the SPDE approach, since the smoothness of the covariance function is important for the predictive performance. However, in many cases the distributional assumptions can be equally important. This was clearly shown in the application where the covariance-based models, which allow for arbitrary smoothness parameters, were outperformed by the non-Gaussian models with fixed smoothness parameters. Nevertheless, extending the approach to fields with general smoothness would increase the flexibility. As previously mentioned, this could likely be done using the rational SPDE approach \citep{bolin2017rational}, and extending that method to multivariate type G fields is thus an interesting topic for future research.

\begin{appendix}	
\section{Finite-dimensional representations}\label{sec:fem}
An advantage with the SPDE approach is that the finite element method can be used for computationally efficient approximations of the models. This was introduced by \cite{lindgren10} for Gaussian models and was  extended to SPDEs driven by type G L\'evy noise in \citep{bolin11}. In this section, we present a multivariate extension of this method. 

In the univariate case, the method is based on a basis expansion  $\proc(\mv{s}) = \sum_{j = 1}^n \weight_j\varphi_j(\mv{s})$,
where $\{\varphi_j\}$ is a collection of piecewise linear basis functions obtained by a triangulation of the (compact) spatial domain of interest $\mathcal{D}$. See Figure \ref{fig:data} for an example. Each node $\tilde{\mv{s}}_j$ in the triangulation defines a piecewise linear basis function $\varphi_j(\mv{s})$ with $\varphi_j(\tilde{\mv{s}}_j) = 1$ that is zero for all locations in triangles not directly connected to the node $\tilde{\mv{s}}_j$. For the multivariate extension, we assume that the SPDE is formulated using the representation in \eqref{eq:general}. Introduce p-dimensional basis functions $\mv{\varphi}_i^k(\mv{s}) = \varphi_i(s) \mv{e}_k$, where $\mv{e}_k$ is the $k$th column in a $p\times p$ identity matrix, and let $\mv{\proc}(s) = \sum_{j = 1}^n\sum_{k=1}^p \weight_{jk}\mv{\varphi}_j^k(\mv{s})$.

The distribution of the stochastic weights $\mv{\weight} = (\weight_{11},\ldots, \weight_{n1},\weight_{12},\ldots, \weight_{n2},\ldots,\weight_{np})^{\trsp}$ is calculated by augmenting the operators in \eqref{eq:general} with homogeneous Neumann boundary conditions and computing the weights using the Galerkin method. For  $\alpha=2$ and Gaussian noise, the result is $\mv{\weight} \sim \pN(\mv{0}, \mv{K}^{-1}\diag(\mv{h})\mv{K}^{-\trsp})$. Here $\mv{h} = \mv{1}_p\otimes (h_{1}, \ldots, h_n)^{\trsp}$ where $h_i = |\mathcal{D}_i|$ is the area of the region $\mathcal{D}_i = \{\mv{s} : \varphi_i(\mv{s}) \geq \varphi_j(\mv{s}) \,\forall j\neq i\}$. Further,
\begin{equation}\label{eq:K}
\mv{K} = (\mv{\BM}_p\otimes \mv{I}_n)\diag(\mv{L}_{\alpha_1}(\sigma_1,\kappa_1),\ldots,\mv{L}_{\alpha_p}(\sigma_p,\kappa_p)),
\end{equation}
is the discretized operator matrix where $\mv{I}_n$ denotes an identity matrix of size $n\times n$, and $\mv{L}_{\alpha_k}(\sigma_k,\kappa_k) = c_k (\mv{G} + \kappa_k^2\mv{C})$ is the discretized operator for the $k$th dimension. The matrices $\mv{C}$ and $\mv{G}$ have elements $C_{ii} = \scal{\varphi_i}{\varphi_i}$ and $G_{ij} = \scal{\nabla\varphi_i}{\nabla\varphi_j}$, respectively, where $\scal{f}{g}$ denotes the inner product on $\R^d$ and $\nabla$ is the gradient operator. 

In the type G case, the corresponding result is 
\begin{align*}
\mv{\weight}|\mv{\var} &\sim \pN(\mv{K}^{-1}((\diag(\mv{\gamma})\otimes \mv{I}_n)\mv{h}+(\diag(\mv{\mu})\otimes \mv{I}_n)\mv{\var}), \mv{K}^{-1}\diag(\mv{\var})\mv{K}^{-\trsp}),
\end{align*}
where $\mv{\var} = (\mv{\var}_1^{\trsp},\ldots, \mv{\var}_p^{\trsp})^{\trsp}$ and $\mv{1}_p$ is a vector with $p$ ones.
The vector  $\mv{\var}_k = (\var_1^k,\ldots, \var_n^k)$ is the discretized variance process for the $k$th dimension, with elements 
\begin{equation*}
\var_i^k  = \int \mathbb{I}(\mv{s}\in \mathcal{D}_i) v_k(d\mv{s}) = \begin{cases}
h_i v & \mbox{type G$_1$,}\\
h_i v_k & \mbox{type G$_2$,}\\
M_{v}(\mathcal{D}_i) & \mbox{type G$_3$,}\\
M_{v_k}(\mathcal{D}_i) & \mbox{type G$_4$},
\end{cases}
\end{equation*}
where $M_v(\cdot)$ denotes the random measure associated with $v$. The distribution of $\mv{v}$ is in general not explicit for type G$_3$ or type G$_4$, unless the distribution of $v_k(\mv{s})$ is closed under convolution. An example of a distribution that has this property is the IG distribution that is used in for the NIG process. 

\begin{example}
	The following equation summarizes the distribution of $\mv{v}$ for the different versions of the NIG processes from Section \ref{sec:nig}.
	\begin{equation}\label{eq:Vdist}
	\mv{v} \sim  \begin{cases}
	\mv{h} \otimes ( \mv{1}_K \otimes IG(\nignu^2,\nignu^2) )& \mbox{type G$_1$,}\\
	\mv{h} \otimes IG(\mv{\nignu}^2,\mv{\nignu}^2) & \mbox{type G$_2$,}\\
	\mv{1}_K  \otimes IG(\nignu^2,\nignu^2\mv{h}^2) & \mbox{type G$_3$,}\\
	IG(\mv{\nignu}^2 \otimes \mv{1}_n,\mv{\nignu}^2\otimes \mv{h}^2) & \mbox{type G$_4$}.
	\end{cases}
	\end{equation}
	Here the notation $\mv{v} \sim IG(\mv{a}, \mv{b})$ is a compact way of writing a vector with independent components $v_i \sim IG(a_i, b_i)$.
\end{example}

The discretization above assumes $\alpha_i = 2$. In the case of $\alpha_i/2\in\mathbb{N}$, each operator is an integer power of the operator for $\alpha_i = 2$ and the method can then be combined with the iterated finite element discretization by \cite{lindgren10} to obtain similar finite dimensional approximations with Markov properties. The only difference in this case is that $\mv{L}_{\alpha_k}(\sigma_k,\kappa_k) = c_k \mv{C}(\mv{C}^{-1}\mv{G} + \kappa_k^2\mv{I})^{\alpha_k}$.

\section{Gradients of the log-likelihood}\label{sec:gradients}
In this section, the gradients needed for the estimation method from Section~\ref{sec:estimation} are presented. The parameters we need the gradients for are $\mu_k$ and  $\sigma_k$ for $k=1,\ldots, p$, the regression parameters $\mv{\beta}$, the parameters of the differential operator matrix $\mv{K}$, as well as any parameters of $\pi(\mv{v})$. 

To simplify notation, let $[ \hat{\mv{E}}_1^{\trsp},\ldots, \hat{\mv{E}}_p^{\trsp}]^{\trsp} = \mv{K}\hat{\mv{\xi}}$, where 
$$
\hat{\mv{\xi}} = \hat{\mv{Q}}^{-1} \left(\sum_{k=1}^p\frac{1}{\sigma_{e,k}^2} \mv{A}_k^{\trsp} \mv{\data}_k + \mv{K}^{\trsp} \diag(\mv{\var})^{-1} (\mv{\mu}\otimes\mv{I}_n)(\mv{\var} - \mv{h})\right)
$$ 
is the posterior mean of $\mv{w}|\mv{v},\mv{\Psi}$ and $\hat{\mv{Q}}  = \mv{K}^{\trsp} \diag(\mv{\var})^{-1} \mv{K}+ \sum_{k=1}^p\frac{1}{\sigma_{e,k}^2} \mv{A}_k^{\trsp} \mv{A}_k$.
All gradients are obtained by first computing $\log\pi(\mv{\var}|\mv{\data}, \mv{\Psi}) = \log\int \pi(\mv{\var},\mv{w}|\mv{\data}, \mv{\Psi})\md\mv{w}$. 
This integral is straight-forward to compute since
\begin{align*}
\log\pi( \mv{\var},\mv{\weight} | \mv{\data},\mv{\Psi}) =& \sum_{k=1}^p\left(-m\log\sigma_{e,k} -\frac{1}{2\sigma_{e,k}^2} \left( \mv{\data}_k- \mv{A}_k\mv{\weight} - \mv{B}\mv{\beta} \right)^{\trsp} \left( \mv{\data}_k- \mv{A}_k \mv{\weight} - \mv{B} \mv{\beta} \right)\right)  \\
& - \frac{1}{2}  \left( \mv{K}\mv{\weight} - (\mv{\mu}\otimes\mv{I}_n)(\mv{\var}-\mv{h})\right)^{\trsp}  \diag(\mv{\var})^{-1} \left( \mv{K}\mv{\weight} - (\mv{\mu}\otimes\mv{I}_n)(\mv{\var}-\mv{h}) \right),   \\
&+ |\mv{K}| - \mv{1}^{\trsp}\log(\mv{v}) + \log(\pi_{\mv{\Psi}}(\mv{v}))  + \mbox{const.}, 
\end{align*}
Standard matrix calculus is then used to differentiate $\log\pi(\mv{\var}|\mv{\data}, \mv{\Psi})$ with respect to the parameters to obtain the required gradients. For brevity we omit the details of these computations and just present the results. The gradients for $\mu_k$, $\sigma_{e,k}$, and $\mv{\beta}$ are
\begin{align*}
\nabla_{\mu_k}\log \pi( \mv{\var}| \mv{\data}, \mv{\Psi}) &=  \left( -\mv{h}_k + \mv{\var}_k \right) ^{\trsp} \diag(\mv{\var}_k)^{-1} \left(  \hat{\mv{E}} - \left(- \mv{h}_k + \mv{\var}_k \right)  \mu_k \right), \\
\nabla_{\sigma_{e,k}}\log \pi(\mv{\var}| \mv{\data}, \mv{\Psi}) &=  - \frac{n}{\sigma_{e,k}} + \frac{1}{\sigma_{e,k}^3} \| \mv{\data}_k- \mv{A}_k\hat{\mv{\xi}} - \mv{B}\mv{\beta} \|^2 + \trace(\mv{A}^{\trsp}\mv{A}\hat{\mv{Q}}^{-1}), \\
\nabla_{\mv{\beta}}\log \pi(\mv{\var}| \mv{\data}, \mv{\Psi}) &=  \sum_{k=1}^p \frac{1}{\sigma_{e,k}^2} \left( \mv{\data}_k- \mv{A}_k\hat{\mv{\xi}} - \mv{B} \mv{\beta} \right)^{\trsp}  \mv{B} .
\end{align*}
For a parameter $\psi_K$ in the operator, the gradient is
\begin{align*}
\nabla_{\psi_K}\log \pi( \mv{\var}| \mv{\data}, \mv{\Psi}) =&  \trace(\mv{K}_{\psi_K}\mv{K}^{-1}) -   \hat{\mv{\xi}}^{\trsp} \mv{K}_{\psi_K}^{\trsp} \diag(\mv{\var})^{-1}\mv{K}\hat{\mv{\xi}}  - \trace(\mv{K}_{\psi_K}^{\trsp} \diag(\mv{\var})^{-1}\mv{K}\hat{\mv{Q}}^{-1}) \\
& +   \hat{\mv{\xi}}^{\trsp} \mv{K}_{\psi_K}^{\trsp} \diag(\mv{\var})^{-1}(\mv{\mu}\otimes\mv{I}_n)\left( -\mv{h} + \mv{\var} \right) ,
\end{align*}
where  $\trace(\cdot)$ denotes the matrix trace, and where $\mv{K}_{\psi_K}$ denotes the derivative of  $\mv{K}$ with respect to $\psi_K$. Using that $\mv{K}$ is on the form given in \eqref{eq:K}, one gets
$$
\mv{K}_{\psi_K} = \begin{cases}
(\mv{\BM}_{\psi_K}\otimes \mv{I}_n) \diag(\mv{L}_1,\ldots, \mv{L}_p) & \mbox{$\psi_K = \theta_i, \rho_{ij}$}, \\
-\sigma_j^{-1}(\mv{\BM}\otimes \mv{I}_n) (\mv{L}_j\otimes\diag(\mv{e}_j)) & \mbox{$\psi_K = \sigma_j$}, \\
\kappa_j^{-1}(\mv{\BM}\otimes \mv{I}_n) (\mv{L}_j(\alpha_j\kappa_j^2(\mv{C}^{-1}\mv{G}+\kappa_j^2\mv{I})^{-1}-\nu_j)\otimes\diag(\mv{e}_j)) & \mbox{$\psi_K = \kappa_j$}, \\
\end{cases}
$$
where $\mv{\BM}_{\psi_K}$ is the derivative of $\mv{\BM}$ with respect to $\psi_K$ and $\mv{L}_i$ denotes $\mv{L}_{\alpha_i}(\sigma_i,\kappa_i)$.

To take full advantage of the sparsity of the matrices, one should compute $\trace(\mv{A}^{\trsp}\mv{A}\hat{\mv{Q}}^{-1})$ and $\trace(\mv{K}_{\psi_K}^{\trsp} \diag(\mv{\var})^{-1}\mv{K}\hat{\mv{Q}}^{-1})$ without inverting $\hat{\mv{Q}}$. To do so,  note that $\mv{A}^{\trsp}\mv{A}$ and $\mv{K}_{\psi_K}^{\trsp} \diag(\mv{\var})^{-1}\mv{K}$ are sparse matrices with non-zero elements only at positions in the matrices where also $\hat{\mv{Q}}$ is non-zero. This means that it is enough to compute the elements of $\hat{\mv{Q}}^{-1}$ only at the positions where $\hat{\mv{Q}}$ is non-zero, which can be done efficiently using the method by \cite{RueMartino07}.

Finally, the expression for the gradient of the parameters for the distribution of $\mv{v}$ depends on which distribution that is used. The following example gives the results for the NIG processes.
\begin{example}
	For the NIG processes in Section \ref{sec:nig}, the gradient of the likelihood with respect to the parameter $\nignu$ in the type G$_1$ and type G$_3$ cases is
	\begin{equation*}
	\nabla_{\nignu}\log \pi( \mv{\var}|\mv{\data}, \mv{\Psi}) = 
	\begin{cases}
	\frac{1}{2\nignu} - \frac{1}{2}\left(\var  +\var^{-1} \right)  +   1& \mbox{type G$_1$,} \\
	\frac{n}{2\nignu} - \frac{1}{2}\left(\mv{\var}   + \mv{h}^2 \cdot \mv{\var}^{-1} \right) \mv{1}  +   \mv{h}^{\trsp}  \mv{1}  & \mbox{type G$_3$,} 
	\end{cases} 
	\end{equation*}
	and the gradient of the likelihood with respect to the parameters $\nignu_k,k=1,\ldots,p$ in the type G$_2$ and type G$_4$ cases is
	\begin{equation*}
	\nabla_{\nignu_k}\log \pi( \mv{\var}|\mv{\data}, \mv{\Psi}) = 
	\begin{cases}
	\frac{1}{2\nignu_k} - \frac{1}{2}\left(\var_k  +\var_k^{-1} \right)   +   1 & \mbox{type G$_2$,} \\
	\frac{n}{2\nignu_k} - \frac{1}{2}\left(\mv{\var}_k   + \mv{h}^2_k \cdot \mv{\var}^{-1}_k \right) \mv{1}  +   \mv{h}_k^{\trsp}  \mv{1}  & \mbox{type G$_4$.}
	\end{cases} 
	\end{equation*}
\end{example}

\section{Pseudo-code for the sampling methods}\label{sec:pseudo}

Algorithm \ref{alg:Gibbs} describes one iteration of the Gibbs sampler that is used to generate the samples used for parameter estimation and prediction. On Line 4 and Line 5 of the algorithm one should not compute the inverse $\hat{\mv{Q}}^{-1}$ but instead use an efficient sampling method for GMRFs based on sparse Cholesky factorization \cite[see][]{rue1}. 
The general form of the distribution of $\mv{v}$ given $\mv{E}=[\mv{E}_1^{\trsp},\ldots,\mv{E}_p^{\trsp}]^{\trsp} = \mv{K} \mv{\weight}$ is shown in Algorithm \ref{algV}, where one can see how the different type G models affect how $v$ is sampled.

\begin{algorithm}[h]
	\footnotesize
	\caption{Gibbs sampler}
	\label{alg:Gibbs}
	\begin{algorithmic}[1]
		\Procedure{GIBBS}{$\mv{y}, \mv{B}, \mv{\var}, \mv{\Psi}, \mv{A}_1, \ldots, \mv{A}_p, \mv{h}$,typeG}
		\State $\mv{K}  \gets BuildOperator(\mv{\Psi})$ (Construct $\mv{K}$ as outlined in Appendix \ref{sec:fem})
		\State $\hat{\mv{Q}}  \gets \mv{K}^{\trsp} \diag(\mv{\var})^{-1} \mv{K}+ \sum_{k=1}^p\frac{1}{\sigma_{e,k}^2} \mv{A}_k^{\trsp} \mv{A}_k$
		\State  $\hat{\mv{\xi}} \gets  \hat{\mv{Q}}^{-1} \left(\sum_{k=1}^p\frac{1}{\sigma_{e,k}^2} \mv{A}_k^{\trsp} \left(\mv{\data}_k -\mv{B}\mv{\beta}\right) + \mv{K}^{\trsp}\diag(\mv{\var})^{-1}(\mv{\mu}\otimes\mv{I}_n)(\mv{\var}-\mv{h})\right)$
		\State Sample $\mv{\weight}   \sim \pN(\hat{\mv{\xi}} , \hat{\mv{Q}}^{-1})$
		\State $[\mv{E}_1^{\trsp},\ldots,\mv{E}_p^{\trsp}]^{\trsp} \gets \mv{K} \mv{\weight}$
		\State Sample $\mv{\var}   \sim \pi(\mv{\var}|\mv{E}_1,\ldots, \mv{E}_p,\mv{\Psi})$ using Algorithm \ref{algV} 
		\State \Return $\{\mv{\weight}, [\mv{\var}_1^{\trsp},\ldots,\mv{\var}_p^{\trsp}]^{\trsp},\hat{\mv{\xi}},\hat{\mv{Q}} \}$
		\EndProcedure
	\end{algorithmic}
\end{algorithm}

\begin{algorithm}[h]
	\footnotesize
	\caption{Variance sampler}
	\label{algV}
	\begin{algorithmic}[1]
		\Procedure{SampleV}{$ \mv{\Psi}, \mv{E}_1, \ldots, \mv{E}_p, \mv{h}$,typeG}
		\If{typeG=1}
		\State Sample $\var   \sim   \pi(\var) \prod_{i=1}^m \prod_{k=1}^p \pN(E_{ik}; h_{ik}(v - 1)  \mu_k),h_{ik} v)$
		\ffor{$k=1,\dots,p$} $\mv{v_k} \gets \mv{h}_k \var;$ \Endffor
		\ElsIf{typeG=2}
		\For{$k=1,\dots,p$}
		\State Sample $\var_k     \sim   \pi(\var) \prod_{i=1}^m  \pN(E_{ik}; h_{ik}(v_k - 1)  \mu_k),h_{ik} v_k)$
		\State $\mv{v_k} \gets \mv{h}_k \var_k$
		\EndFor
		\ElsIf{typeG=3}
		\ffor{$i=1,\dots,m$}  
		Sample $\var_i   \sim   \pi(v_{i}) \prod_{k=1}^p \pN(E_{ik}; (v_i - h_{ik}) \mu_k, v_i)$
		\Endffor
		\ffor{$k=1,\dots,p$} $\mv{v}_k \gets \mv{\var};$	\Endffor
		\ElsIf{typeG=4}
		\For{$k=1,\dots,p$}
		\ffor{$i=1,\dots,m$}  Sample $\var_{ik}   \sim  \pi(v_{ik})\pN(E_{ik}; (v_i - h_{ik}) \mu_k, v_i) $ \Endffor
		\EndFor
		\EndIf
		\State \Return $\{[\mv{\var}_1^{\trsp},\ldots,\mv{\var}_p^{\trsp}]^{\trsp} \}$
		\EndProcedure
	\end{algorithmic}
\end{algorithm}

\section{Parameter estimates for the applications}\label{seq:parameter_estimates}

The parameter estimates for the two covariance-based models in the first application are shown in Table~\ref{tab:cov}, and the parameter estimates for the SPDE models are shown in Table~\ref{tab1}. The main reason for the differences between our parameter estimates and those by \cite{gneiting2012matern} and \cite{apanasovich2012valid} is that they assumed $\mv{\beta} = \mv{0}$ whereas we estimate this parameter jointly with the other parameters. The reason for doing this is that the comparison with the type G models otherwise could be considered to be unfair, since the type G models allow for skewness that could capture some of the effects that cause the non-zero estimates of the means.

\begin{table}
	\centering
	\begin{tabular}{lccccccccccccc}
		\toprule
		Model &  $\beta_1$ & $\beta_2$ & $\sigma_1$  & $\sigma_2$ & $\kappa_1$ & $\kappa_2$  & $\nu_1$ & $\nu_2$ & $\rho$ & $\sigma_{1e}$ & $\sigma_{2e}$ \\
		\midrule
		Independent   	& 136 & -0.53 & 218 & 2.64 & 5.54  & 0.89 & 20 	&  0.58 	& - 		& 71.8 & 0.00  \\
		Parsimonious 	& 150 & -0.48 & 216 & 2.56 & 1.03  & 1.03 & 1.36 & 0.60 	& -0.46 & 68.5 & 0.00  \\
		\bottomrule
	\end{tabular}
	\caption{\label{tab:cov} Parameter estimates for the covariance-based models. For the independent model, the value of $\nu_P$  was limited to the interval $0 \leq \nu_p \leq 20$ for numerical stability. }
\end{table}

\begin{table}[t]
	\centering
		\begin{tabular}{lcccccccccccccc}
			\toprule
			Noise &  $\beta_1$ & $\beta_2$ & $\sigma_1$  & $\sigma_2$ & $\kappa_1$ & $\kappa_2$  & $\rho$ & $\sigma_{1e}$ & $\sigma_{2e}$ & $\theta$ & $\mu_1$ & $\nignu_1$\\
			\midrule
			GG  	& $154$ & $-0.55$ & $211$ & $2.56$ & $0.74$ & $1.11$ & $(0)$     & $61.4$ & $0.58$ & $-$        & $-$         & $-$\\
			GG 	& $149$ & $-0.52$ & $202$ & $2.48$ & $0.82$ & $1.26$ & $-0.52$ & $60.5$ & $0.52$ & $-$       & $-$          & $-$\\
			NG	& $148$ & $-0.48$ & $222$ & $2.74$ & $0.72$ & $1.12$ & $(0)$     & $45.4$ & $0.75$ & $(0)$    & $-0.014$ & $0.21$\\
			NG	& $140$ & $-0.42$ & $212$ & $2.73$ & $0.74$ & $1.19$ & $-0.42$ & $45.3$ & $0.74$ & $(0)$    & $-0.053$ & $0.21$\\
			NG	& $147$ & $-0.59$ & $220$ & $2.87$ & $0.77$ & $1.18$ & $-0.42$ & $42.3$ & $0.72$ & $-0.89$ & $-0.065$ & $0.21$\\
			\bottomrule
	\end{tabular}
	\caption{\label{tab1} Parameter estimates for the SPDE models. Dashes and parentheses respectively indicates that the parameters are not present and not estimated. GG denotes a Gaussian model whereas NG denotes that NIG noise is used for pressure and Gaussian noise for temperature. }
\end{table}

The parameter estimates for the SPDE-based models for the Argo data are shown in Table~\ref{tab:argospde}, whereas Table~\ref{tab:argocov} shows the parameter estimates for the covariance-based models.

\begin{table}[t]
	\centering
	\resizebox{\linewidth}{!}{
		\begin{tabular}{lcccccccccccccc}
			\toprule
			Noise &  $\beta_1$ & $\beta_2$ & $\tau_1$  & $\tau_2$ & $\kappa_1$ & $\kappa_2$  & $\rho$ & $\sigma_{1e}$ & $\sigma_{2e}$ & $\theta$ &  
			$\mu_1$ & $\mu_2$ & $\nignu_1$ & $\nignu_2$ \\
			\midrule
			GG  	& $0.00$  & $0.01$ & $0.89$ & $5.02$ & $1.14$ & $1.15$ & $(0)$    & $0.39$ & $0.03$ & $-$      & $-$        & $-$         & $-$       & $-$\\
			GG 	& $-0.01$ & $0.01$ & $0.39$ & $2.23$ & $1.18$ & $1.09$ & $0.97$ & $0.35$ & $0.02$ & $-$      & $-$        & $-$         & $-$       & $-$\\
			NN	& $0.01$  & $0.05$ & $0.74$ & $5.56$ & $1.88$ & $1.39$ & $(0)$    & $0.04$ & $0.01$ & $(0)$   & $-0.03$ & $-0.11$ & $0.09$  & $0.73$\\
			NN	& $0.06$  & $0.04$ & $0.37$ & $2.99$ & $1.55$ & $0.96$ & $1.31$ & $0.04$ & $0.03$ & $0.04$ &$-0.02$  & $-3.24$ & $0.27$ & $6.61$\\
			\bottomrule
	\end{tabular}}
	\caption{\label{tab:argospde} Parameter estimates for the SPDE models for the Argo data. Dashes and parentheses respectively indicates that the parameters are not present and not estimated. GG denotes a Gaussian model whereas NN denotes a NIG model. }
\end{table}

\begin{table}[t]
	\centering
	\begin{tabular}{lccccccccccc}
		\toprule
		Model &  $\beta_1$ & $\beta_2$ & $\sigma_1$  & $\sigma_2$ & $\kappa_1$ & $\kappa_2$  & $\nu_1$ & $\nu_2$ & $\rho$ & $\sigma_{1e}$ & $\sigma_{2e}$ \\
		\midrule
		Independent   & $0.00$  & $0.00$ & $0.72$ & $0.13$ & $1.32$ & $1.49$ & $1.01$ & $1.08$ & -           & $0.41$ & $0.03$\\
		Parsimonious & $-0.01$ & $0.00$ & $0.80$ & $0.14$ & $0.63$ & $0.63$ & $0.32$ & $0.51$ & $0.65$ & $0.24$ & $0.00$\\
		\bottomrule
	\end{tabular}
	\caption{\label{tab:argocov} Parameter estimates for the covariance-based models for the Argo data. }
\end{table}

\section{Proofs}\label{sec:proofs}
Most of the proofs are based on that the fractional operator 
$(\kappa^2-\Delta)^{\alpha/2}$ on $\mathbb{R}^d$  is defined through its Fourier transform \citep[see][]{lindgren10}, 
$(\mathcal{F}((\kappa^2-\Delta)^{\alpha/2}\varphi)(\mv{k}) = (\kappa^2+\|\mv{k}\|)^{\alpha/2}(\mathcal{F}(\varphi))(\mv{k})$. 
The operator is well-defined for example if $\varphi$ is a tempered distribution. This is important for the definition of the SPDE in \eqref{eq:model} since the right-hand side is white noise, which does not have pointwise meaning. Thus, the equation \eqref{eq:model} is understood in the weak sense, $(\kappa^2-\Delta)^{\alpha/2}X(\varphi)=\noise(\varphi)$,  where $\varphi$ is a function in an appropriate space of test functions, and $\noise(\varphi) = \int\varphi(\mv{s})\mathcal{M}(\md\mv{s})$. The kernel of the operator $\mathcal{K} = (\kappa^2-\Delta)^{\frac{\alpha}{2}}$ is non-empty for $\alpha\geq 2$ and there is therefore an implicit assumption on boundary conditions \citep[see][]{lindgren10}.

\begin{proof}[Proof of Proposition \ref{thm1}]
	Due to the mutual independence of the noise processes, the power spectrum of driving noise is $\mv{S}_{\mathcal{M}} = (2\pi)^{-d}\mv{I}$.  Let
	$$
	\mv{\mathcal{H}}(\mv{k}) = \mathcal{F}(\mv{\mathcal{K}})(\mv{k}) = \mv{D} \mathcal{F}(\diag(\mathcal{L}_1,\ldots, \mathcal{L}_p)) = \mv{D} \mv{\mathcal{H}}_D(\mv{k}),
	$$
	where $\mv{\mathcal{H}}_D(\mv{k})$ is a diagonal matrix with elements $\mv{\mathcal{H}}_D(\mv{k})_{ii} = \mathcal{F}(\mathcal{L}_i) = (\kappa_i^2 + \|\mv{k}\|)^{\alpha_i/2}$. 
	The power spectrum of $\mv{\proc}$ can then be written as 
	\begin{align}
	\mv{S}_{\mv{\proc}}(\mv{k}) &= 
	(2\pi)^{-d}\mathcal{H}_D(\mv{k})^{-1}\mv{\BMi}\mv{\BMi}^{T}\mathcal{H}_D(\mv{k})^{-1}. \label{eq:Xspectrum}
	\end{align}
	Evaluating a single element of $\mv{S}_{\mv{\proc}}(\mv{k})$ gives
	$$
	(\mv{S}_{\mv{\proc}}(\mv{k}))_{ij} = \frac{\sum_{k=1}^p \BMi_{ik}\BMi_{jk}}{(2\pi)^d}\frac1{(\kappa_i^2 + \|\mv{k}\|)^{\alpha_i/2}(\kappa_j^2 + \|\mv{k}\|)^{\alpha_j/2}}.
	$$
	It is well-known that \citep{lindgren10}
	$$
	\mathcal{F}^{-1}\left(\frac1{(2\pi)^d}\frac{1}{ (\kappa^2 + \|\mv{k}\|)^{\alpha}}\right)(\mv{h}) = \frac{\Gamma(\nu)}{(4\pi)^{d/2}\Gamma(\alpha)\kappa^{2\nu}}\materncorr{\mv{h}}{\kappa_i}{\nu_i}
	$$
	which together with the expression for $(\mv{S}_{\mv{\proc}}(\mv{k}))_{ii}$ completes the proof. 
\end{proof}

\begin{proof}[Proof of Proposition \ref{thm2}]
	By the representation of the multivariate Mat\'ern-SPDE in Remark \ref{cor2}, we have that the covariance function of $\mv{\proc}$ depends on $\mv{\BM}$ only through the expression $\mv{\BMi}\mv{\BMi}^{\trsp} = (\mv{\BM}^{\trsp}\mv{\BM})^{-1}$. It is therefore clear that $\mv{\BM}$ and $\hat{\mv{\BM}}$ will generate the same covariance structure if and only if $\mv{\BM}^{\trsp}\mv{\BM} = \hat{\mv{\BM}}^{\trsp}\hat{\mv{\BM}}$. 
	
	If we assume $\mv{\BM} = \mv{Q}\hat{\mv{\BM}}$, then 
	$
	\mv{\BM}^{\trsp}\mv{\BM} = \hat{\mv{\BM}}^{\trsp}\mv{Q}^{\trsp}\mv{Q}\hat{\mv{\BM}} = \hat{\mv{\BM}}^{\trsp}\hat{\mv{\BM}},
	$ 
	since $\mv{Q}$ is orthogonal, and the models therefore have the same covariance structure. 
	Conversely, assume that $\mv{\BM}$ and $\hat{\mv{\BM}}$ generate the same covariance structure. We then have that 
	$\mv{\BM}^{\trsp}\mv{\BM} = \hat{\mv{\BM}}^{\trsp}\hat{\mv{\BM}}$. Since $\hat{\mv{\BM}}$ is invertible, we can define $\mv{Q} = \mv{\BM}\hat{\mv{\BM}}^{-1}$ which is orthogonal, since $\mv{Q}^{\trsp}\mv{Q} = (\mv{\BM}\hat{\mv{\BM}}^{-1})^{\trsp}\mv{\BM}\hat{\mv{\BM}}^{-1} = \mv{I}$, and satisfies $\mv{Q}\hat{\mv{\BM}} = \mv{\BM}\hat{\mv{\BM}}^{-1}\hat{\mv{\BM}} = \mv{\BM}$. 
	
	Finally, for any multivariate Mat\'ern-SPDE, the matrix  $\mv{\BM}^{\trsp}\mv{\BM} $ is by definition symmetric and positive definite. We can therefore define a Mat\'ern-SPDE model with triangular dependence matrix $\hat{\mv{\BM}} = \chol(\mv{\BM}^{\trsp}\mv{\BM})$. Because of the properties of the Cholesky factor, $\hat{\mv{\BM}}$ is the unique upper-triangular matrix with positive diagonal elements satisfying $\hat{\mv{\BM}}^{\trsp}\hat{\mv{\BM}} = \mv{\BM}^{\trsp}\mv{\BM}$. 
\end{proof}

\begin{proof}[Proof of Proposition \ref{thm3}]
	We only have to show that  $\Cov(\proc_i(\mv{s}),\proc_j(\mv{t}))= 0$ since the variables are dependent by construction. Since $\mv{\rho} = 0,$ $\proc(\mv{s})$ is the solution to
	$
	\mv{Q}(\mv{\theta})\diag(c_1\mathcal{L}_1,\cdots, c_p\mathcal{L}_p)\mv{\proc}(\mv{s}) = \mv{\mathcal{M}}_{3},
	$
	or equivalently
	$
	\diag(c_1\mathcal{L}_1,\cdots, c_p\mathcal{L}_p)\mv{\proc}(\mv{s}) = \mv{\mathcal{M}}_{Q},
	$
	where
	$$
	\mv{\mathcal{M}}_{Q} = \sum_{k=1}^{\infty}g(e_k)^{\frac{1}{2}}\mathbb{I}(\mv{s} \geq \mv{s}_k) \mv{Q}(\mv{\theta})^{-1}\mv{Z}_k = \sum_{k=1}^{\infty}g(e_k)^{\frac{1}{2}}\mathbb{I}(\mv{s} \geq \mv{s}_k) \mv{Z}_{Q,k}.
	$$
	Since $\mv{Q}(\mv{\theta})^{-1}\mv{Q}(\mv{\theta})^{-\trsp} = \mv{I}$ it follows that $\mv{Z}_{Q,k} \sim N(\mv{0},\mv{I})$.
	From \citep{bolin11}  we have 
	$
	\proc_i(\mv{s}) = \int G_i(\mv{s},\mv{v})  \mv{\mathcal{M}}_{Q,i}(d\mv{v}),
	$ 
	for $i=1,\ldots,p$. Here $G_i(\mv{s},\mv{v})$ is the Green function of $c_i \mathcal{L}_i$, and $\mv{\mathcal{M}}_{Q,i}(\mv{s})$ is the $i$th value of the vector $\mv{\mathcal{M}}_{Q}(\mv{s})$. Since the elements in the vector $\mv{\mathcal{M}}_{Q}(\mv{s})$ are uncorrelated it follows that the elements of $\mv{\proc}(\mv{s}) = [\proc_1(\mv{s}), \ldots,  \proc_d(\mv{s})]^{\trsp}$ are uncorrelated.
\end{proof}

The proof of the Theorem \ref{thm:krig} builds on the following lemma, which shows that the posterior distribution of $v$ contracts to a point.
\begin{lemma}
	\label{lem:contraction}	
	Let Assumption \ref{ass1} hold and assume that $\pi(v)$ has mean one, is bounded and absolutely continuous with respect to the Lebesgue measure. Then $\pi(v|X_{1:n})  \overset{p}{\to}\delta_{K_0}(v)$ as $n \rightarrow\infty$.
\end{lemma}

\begin{proof}
	In the following, $C$ is a generic positive constant that changes from line to line.
	Let $k_0 := k_0(x)$ be the realisation of $K_0$ determined by the realisation of $x$ which generates the data. Let  $B_{k,n}  = (b^l_{k,n},b^u_{k,n})=(k_0-   n^{-1/2+k\epsilon}, k_0 +   n^{-1/2+k\epsilon})$ where $0<\epsilon<1/8$. To prove the lemma it suffices to show that
	$$
	\frac{\int_{B_{2,n}^c} \pi(v|\mv{x}_{1:n}) dv }{\int_{B_{2,n}} \pi(v|\mv{x}_{1:n}) dv} \rightarrow 0\quad \mbox{ as } n\rightarrow \infty.
	$$
	By the mean value theorem we have
	$$
	\frac{\int_{B_{2,n}^c} \pi(v|\mv{x}_{1:n}) dv }{\int_{B_{2,n}} \pi(v|\mv{x}_{1:n}) dv} \leq \frac{\int_{B_{2,n}^c} \pi(v|\mv{x}_{1:n}) dv }{\int_{B_{1,n}} \pi(v|\mv{x}_{1:n}) dv} 
	\leq \frac{\sqrt{n}\int_{B_{2,n}^c} \pi(v|\mv{x}_{1:n}) dv }{ \pi(v_d|\mv{x}_{1:n})},
	$$
	for some $v_d\in B_{1,n}$. By boundedness and absolute continuity of $\pi(v)$  (which implies that $\frac{\pi(v)}{\pi(v_d)}$ is bounded from above) 
	$$
	\frac{\sqrt{n}\int_{B_{2,n}^c} \pi(v|\mv{x}_{1:n}) dv }{ \pi(v_d|\mv{x}_{1:n})} \leq C \sqrt{n} \int_{B^c_{2,n}} e^{-\frac{n}{2}\left(\frac{c_n}{v} + \log(v) -  \frac{c_n}{v_d} - \log(v_d) \right)} dv,
	$$
	where $c_n=\frac1{n}\mv{x}^{\trsp}_{1:n} \mv{C}_n^{-1} \mv{x}_{1:n}$. We will now show that the right-hand side goes to zero if we condition on the event $A_n = \{c_n\in B_{0.5,n}\}$. We first bound the integral as
	\begin{align*}
	\sqrt{n} \int_{B^c_{2,n}} &e^{-\frac{n}{2}\left(\frac{c_n}{v} + \log(v) -  \frac{c_n}{v_d} - \log(v_d) \right)} dv \\ 
	&\leq \sqrt{n} \int_{B^c_{2,n} \cap [0,n]} e^{-\frac{n}{2}\left(\frac{c_n}{v} + \log(v) -  \frac{c_n}{v_d} - \log(v_d) \right)} dv 
	+ \sqrt{n} \int_{n}^\infty e^{-\frac{n}{2}\left(\frac{c_n}{v} + \log(v) -  \frac{c_n}{v_d} - \log(v_d) \right)}
	\\ 
	&:= (I) + (II).
	\end{align*} 
	To bound (II), let $k$ be a constant such that $\log(k) > \frac{c_n}{v_d} + \log(v_d)$ for all $n$ (this is possible since we are in $A_n$) then
	$$
	(II)  \leq  \sqrt{n} \int_{n}^\infty e^{-\frac{n}{2}\left(\log(v) -  \log(k) \right)} dv =\frac{n^{-\frac{n-1}{2}}}{n/2-1} k^{n/2} \rightarrow 0 \quad \mbox{as $n\rightarrow \infty$.}
	$$
	To bound (I), note that $f(v) = \frac{c_n}{v} + \log(v) $ takes its minimum at $c_n$, and is increasing above and below $c_n$.  Thus, for $i = \argmax_{j\in \{l,u\}} f(b^j_{2,n})$ we have $f(v) \geq f(b^i_{2,n})> f(b^i_{1,n}) > f(v_d) $, for all $v\in B_{2,n}^c$, and therefore 
	$$
	(I) \leq C n^{3/2} e^{-\frac{n}{2}\left(f(b^i_{2,n})-f(b^i_{1,n})\right)}.
	$$
	Assume for simplicity that $i=u$ (the calculation for $i=l$ follows from similar arguments). We split the exponent into two parts $f(b^u_{2,n})-f(b^u_{1,n}) = (\frac{c_n}{b^u_{2,n}} -  \frac{c_n}{b^u_{1,n}}) + (\log(b^u_{2,n}) - \log(b^u_{1,n}))$.
	For the first part we have
	\begin{align*}
	\frac{c_n}{b^u_{2,n}} -  \frac{c_n}{b^u_{1,n}} &= \frac{c_n n^{-1/2}(n^{\epsilon}-n^{2\epsilon})}{(k_0 + n^{-1/2+\epsilon})(k_0 + n^{-1/2+2\epsilon})} \\
	&\geq  \frac{c_n n^{-1/2}(n^{\epsilon}-n^{2\epsilon})}{(k_0 + n^{-1/2+2\epsilon})^2} \geq \frac{c_n n^{-1/2}(n^{\epsilon}-n^{2\epsilon})}{k_0^2 + n^{-1+4\epsilon}},
	\end{align*}
	while for the second part
	\begin{align*}
	\log(b^u_{2,n}) - \log(b^u_{1,n}) &= \log(1 + \frac{n^{-1/2+2\epsilon}}{k_0}) - \log(1 + \frac{n^{-1/2+\epsilon}}{k_0}) \\
	&= \frac{n^{-1/2+2\epsilon}}{k_0} + \mathcal{O}(n^{-1+4\epsilon}) -  \frac{n^{-1/2+\epsilon}}{k_0} +  \frac{n^{-1+2\epsilon}}{k_0^2} + \mathcal{O}(n^{-3/2+3\epsilon})\\
	&= \frac{n^{-1/2}}{k_0}(n^{2\epsilon} - n^{\epsilon}) +  \frac{n^{-1+2\epsilon}}{k_0^2} + \mathcal{O}(n^{-1+4\epsilon}) .
	\end{align*}
	
	Hence
	$C \sqrt{n}  e^{-\frac{n}{2}\left(f(b^i_{2,n})-f(b^i_{1,n})\right)} \leq  C \sqrt{n}  e^{-\frac{n^{4\epsilon}}{4k_0}} \to 0$ as $n\rightarrow \infty$. 
	Finally, by Assumption~\ref{ass1} and the Chebyshev inequality, $P(A_n) \rightarrow 1$, which completes the proof. 
\end{proof}

\begin{proof}[Proof of Theorem \ref{thm:krig}]
	Let $k_0 := k_0(x)$ and $k_1 := k_1(x)$ be the realisations of $K_0$ and $K_1$ respectively, determined by the realisation of $x$ which generates the data.
	Take $\epsilon>0$ and define $A_n=\{\mv{c}_{0,1:n} \mv{C}_n^{-1} \mv{x}_{1:n}\in [k_1 - \frac{\epsilon}{\sqrt{n}},k_1 + \frac{\epsilon}{\sqrt{n}}]\}$.  Conditioning on the event $A_n$ and using the triangle inequality yields
	\begin{align*}
	|	\pi_{G_1,x_0}(\cdot|\mv{x}_{1:n}) - N(\cdot; k_1, k_0 k_2)| \leq&  \left|\pi_{G_1,x_0}(\cdot|\mv{x}_{1:n}) - \int N(\cdot; k_1, vk_2)\pi(v) dv  \right| \\
	&+\left|\int N(\cdot; k_1, vk_2)\pi(v) dv-  N(\cdot; k_1, k_0 k_2) \right|.
	\end{align*}
	By equation \eqref{eq:kriging_mean} and the continuous mapping theorem, the first term on the right-hand side converges to zero since we have conditioned on the event $A_n$, and the second term converges to zero by Lemma \ref{lem:contraction}. Under  Assumption \ref{ass1} and using Chebyshev inequality it follows that $P(A_n) \rightarrow 1$, which completes the proof. 
\end{proof}

\begin{proof}[Proof of Lemma \ref{eq:Lemma}]
	To prove the result we need to verify that Assumption~\ref{ass1} is satisfied.
	We first establish some properties about $x(\mv{s})$ which we will use to verify the assumptions. Note that the distribution of $x(\mv{s}_1),\ldots,x(\mv{s}_n)|v$ is $N(0,v\mv{C}_n)$ where $\mv{C}_n$ is a positive definite matrix for all $n$. Let $\mv{C}^{1/2}_n$ denote the Cholesky factor of $\mv{C}_n$, and let $\mv{z}_{1:n} =\frac{1}{\sqrt{v}}\mv{C}_n^{-1/2} \mv{x}_{1:n}$ where by assumption $z_i$ are i.i.d $N(0,1)$. 
	
	To establish \eqref{eq:mean_est} and \eqref{eq:var_est} note that
	$$\mv{x}_{1:n}^{\trsp}\mv{C}_n^{-1}\mv{x}_{1:n}^{\trsp} = \mv{z}_{1:n}^{\trsp} v^{1/2} \mv{C}_n^{1/2} \mv{C}_n^{-1} v^{1/2}\mv{C}_n^{1/2} \mv{z}_{1:n}^{\trsp} =  v \sum_{i=1}^n z_{i}^2.
	$$
	Hence, by the law of large numbers, \eqref{eq:mean_est} and \eqref{eq:var_est} are satisfied with $K_0=v$.
	
	For \eqref{eq:kriging_mean} and \eqref{eq:kriging_var}, note that $\sigma_n := c_0 - \mv{c}^{T}_{0,1:n}\mv{C}_n^{-1}\mv{c}_{0,1:n}$ is the variance of the Kriging predictor (the variance of the best linear predictor), thus $\{\sigma_n\}$ is a decreasing sequence in $[0, C_2]$. Therefore $\{\sigma_n\}$  must converge to a point, implying equation \eqref{eq:kriging_var}. Finally we need to establish that
	$$
	\mv{c}_{0,1:n} \mv{C}_n^{-1} \mv{x}_{1:n} = \sqrt{v} \mv{c}_{0,1:n}\mv{C}_n^{-1/2}  \mv{z}_{1:n} \overset{p}{\to} K_1.
	$$
	Since $\mv{C}^{1/2}_n$ is the Cholesky factor of $\mv{C}_n$, we have that $(\mv{C}_n^{-1/2})_{1:n-1,1:n-1}=\mv{C}_{n-1}^{-1/2}$ \cite[see for instance][Section 2.2.4]{pourahmadi2011covariance}. Thus, the limit $\tilde{c} = \lim_{n\rightarrow \infty} \mv{c}_{0,1:n}\mv{C}_n^{-1/2}$ exists. By  \eqref{eq:kriging_var} it follows that $\tilde{c}\in l^2$ and hence that $\sum_{i=n}^\infty \tilde{c}^2_{i}  \rightarrow 0$ as $n \to  \infty$. Thus  $\pV[\mv{c}_{0,1:n} \mv{C}_n^{-1} \mv{x}_{1:n} - K_1]= \pV[\sqrt{v}\sum_{i=n}^\infty \tilde{c}_i z_{i}] = \pE[v]\sum_{i=n}^\infty \tilde{c}^2_{i}\to 0$ as $n\rightarrow \infty$. 
\end{proof}

\begin{proof}[Proof of Proposition \ref{thm:charf}]
	To derive the CF, $\phi_{\proc(\mv{s})}(\mv{u})$, of $\mv{x}(s)$, note that Remark~\ref{cor2} shows that the SPDE in \eqref{eq:general}, for $p>1$, can be formulated as 
	$$
	\mv{\proc}(\mv{s})
	= \sum_{k=1}^p
	\diag(\mathcal{L}^{-1}_1\mathcal{N}_k(\mv{s}),\ldots, \mathcal{L}^{-1}_p \mathcal{N}_k(\mv{s})
	)\begin{bmatrix}
	\BMi_{1k}, \\
	\vdots\\
	\BMi_{pk}
	\end{bmatrix}  = \sum_{k=1}^p \begin{bmatrix} R_{1k}x^1_k(\mv{s}) \\
	\vdots \\
	R_{pk}x^p_k(\mv{s}) 
	\end{bmatrix},
	$$
	where $\proc_k^r(\mv{s})=\mathcal{L}^{-1}_r \mathcal{N}_k(\mv{s})$.
	The right-hand side is a sum of independent random variables, and thus $\phi_{\proc(\mv{s})}(\mv{u}) = \prod_{i=1}^p\phi_{k}(\mv{u})$ where  $\phi_{k}(\mv{u})$ is the CF of $\begin{bmatrix} R_{1k}x^1_k(\mv{s}) &
	\ldots & 
	R_{pk}x^p_k(\mv{s}) 
	\end{bmatrix}^{\trsp}$. 
	In order to derive $\phi_{k}(\mv{u})$ we first derive the CF for $\proc_k^r(\mv{s})$. From \citep{bolin11} it follows that  $\proc_k^r(\mv{s}) = \int G_r(\mv{s},\mv{t}) \mathcal{N}_k(d\mv{t})$, where the kernel is given by the Green's function of the operator $\mathcal{L}_r$:
	$$
	G_r(\mv{s},\mv{t})  = \frac{\Gamma\left(\frac{\alpha_r-d}{2}\right)}{(4\pi)^{d/2} \Gamma(\frac{\alpha_r}{2}) \kappa_r^{\alpha_1 - d}} \materncorr{\|\mv{s}-\mv{t}\|}{\kappa_r}{\frac{\alpha_r-d}{2}}.
	$$
	Using that the CF of the univariate NIG noise $\dot{\mathcal{N}}_k(A)$ is 
	\begin{align}
	\label{eq:charfuncM}
	\phi_{\dot{\mathcal{N}}_k(A)}(u) = \exp \left(  i\gamma m(A) u_k + m(A) \sqrt{\nignu_k} \left(\sqrt{\nignu_k}- \sqrt{\nignu_k + u^2 - 2i\mu_k u}  \right)\right),
	\end{align}
	and Proposition 2.6 in \citep{rajput1989spectral} it follows that the CF of  $\proc_k^r(\mv{s})$  is
	\begin{align*}
	\phi_{\proc_k^r(\mv{s})}(u) = \exp \left(  - i\gamma_k u \int  G_r(\mv{s},\mv{t}) d\mv{t} +   \sqrt{\nignu_k} \int  (\nignu_k - \sqrt{\nignu_k  + \mu_k^2- (\mu_k + iG_r(\mv{s},\mv{t})u)^2}   d\mv{t}\right).
	\end{align*}
	
	To complete the proof we need derive  $\phi_{k}(\mv{u})$. Note that the random variable $Y(\mv{s}) = \sum_{r=1}^p u_{r}R_{rk}\proc^r_k(\mv{s})$ has CF $\phi_{Y(\mv{s})}(h) = \phi_{k}(\mv{u}h)$ and since  $Y(\mv{s}) =  \int \sum_{r=1}^p R_{rk}G_r(\mv{s},\mv{t})u_r \mathcal{N}_k(d\mv{t})$ it follows that
	\begin{align*}
	\phi_{k}(\mv{u}) = \phi_{Y(\mv{s})}(1) = \exp \left( \right.& - i\gamma_k  \int \sum_{r=1}^p R_{rk}G_r(\mv{s},\mv{t})u_r d\mv{t}  +   \\
	&\left.  \sqrt{\nignu_k} \int  \nignu_k - \sqrt{\nignu_k  + \mu_k^2- (\mu_k + i \sum_{r=1}^pR_{rk}G_r(\mv{s},\mv{t})u_r )^2}   d\mv{t}\right). 
	\end{align*}
\end{proof}

\begin{proof}[Proof of Proposition \ref{crpsthm}]
	If $Z \sim \pN(\mu,\sigma^2)$, then $|Z|$ has a folded normal distribution with mean $M(\mu,\sigma^2)$ defined in  \eqref{eq:M}. Let $X_1$ and $X_2$ be two independent variance mixture variables with CDF $F$ and let $V_1$ and $V_2$ be their corresponding mixing variables. Introduce $\tilde{X}_1 = X_1 -y$ and $\tilde{X}_2 = X_1 - X_2$ and note that there exist variables $\mu_1,\mu_2,\sigma_1^2,$ and $\sigma_2^2$, depending on $V_1$ and $V_2$, such that $\tilde{X}_1 |V_1 \sim \pN(\mu_1 - y, \sigma_1^2)$ and $\tilde{X}_2 | V_1,V_2 \sim \pN(\mu_1 - \mu_2, \sigma_1^2 + \sigma_2^2)$.
	By the law of total expectation
	\begin{align}\label{eq:crps22}
	\mbox{CRPS}(F,y) &= \pE_{V_1}(\pE(|\tilde{X}_1 - y|\mid V_1)) - \frac1{2}\pE_{V_1}(\pE_{V_2}(\pE(|\tilde{X}_1 - \tilde{X}_2| \mid V_1, V_2))) \notag\\
	&= \pE_{V_1}(M(\mu_1 - y,\sigma_1^2)) - \frac1{2}\pE_{V_1}(\pE_{V_2}(M(\mu_1 - \mu_2,\sigma_1^2+\sigma_2^2))).
	\end{align}
	We have that $\pE(\mbox{CRPS}_N^{RB}(F,y)) =  \mbox{CRPS}(F,y)$ since $\mbox{CRPS}_N^{RB}(F,y)$ is a standard MC estimator of  \eqref{eq:crps22}. Furthermore,  $\mbox{CRPS}_N^{RB}(F,y) = \pE(\mbox{CRPS}_N(F,y)|\mv{V}_1,\mv{V}_2)$ where $\mv{V}_j = (V_j^{(1)},\ldots V_j^{(N)})$ for $j=1,2$. Thus, $\mbox{CRPS}_N^{RB}(F,y)$ is a Rao-Blackwell estimator and by the Law of total variation  $\pV(\mbox{CRPS}_N^{RB}(F,y)) \leq \pV(\mbox{CRPS}_N(F,y))$. 
\end{proof}
	\end{appendix}
	
\section*{Acknowledgment}
This work has been supported by the Swedish Research Council under grant
No. 2016-04187 and the Knut and Alice Wallenberg Foundation (KAW 20012.0067). The authors thank Holger Rootz\'en, the editors, and the anonymous reviewers for valuable comments on the manuscript. We also thank Mikael Kuusela for helping with the Argo data.

	\bibliographystyle{chicago}
	\bibliography{multfield}

\begin{thebibliography}{}

\bibitem[\protect\citeauthoryear{Andrieu, Moulines, and Priouret}{Andrieu
  et~al.}{2005}]{andrieu2005stability}
Andrieu, C., {\'E}.~Moulines, and P.~Priouret (2005).
\newblock Stability of stochastic approximation under verifiable conditions.
\newblock {\em {SIAM} J.\ Control Optim.\/}~{\em 44\/}(1), 283--312.

\bibitem[\protect\citeauthoryear{Apanasovich, Genton, and Sun}{Apanasovich
  et~al.}{2012}]{apanasovich2012valid}
Apanasovich, T.~V., M.~G. Genton, and Y.~Sun (2012).
\newblock A valid {M}at{\'e}rn class of cross-covariance functions for
  multivariate random fields with any number of components.
\newblock {\em J.\ Amer.\ Statist.\ Assoc.\/}~{\em 107\/}(497), 180--193.

\bibitem[\protect\citeauthoryear{B{\'a}rdossy}{B{\'a}rdossy}{2006}]{bardossy2006copula}
B{\'a}rdossy, A. (2006).
\newblock Copula-based geostatistical models for groundwater quality
  parameters.
\newblock {\em Water Resour.\ Res.\/}~{\em 42\/}(11), W11416.

\bibitem[\protect\citeauthoryear{Barndorff-Nielsen, Kent, and
  S{\o}rensen}{Barndorff-Nielsen et~al.}{1982}]{Barndorff1982}
Barndorff-Nielsen, O., J.~Kent, and M.~S{\o}rensen (1982).
\newblock Normal variance-mean mixtures and z distributions.
\newblock {\em Internat.\ Statist.\ Review\/}~{\em 50\/}(2), 145--159.

\bibitem[\protect\citeauthoryear{Barndorff-Nielsen}{Barndorff-Nielsen}{1997}]{barndorff1997normal}
Barndorff-Nielsen, O.~E. (1997).
\newblock Normal inverse {G}aussian distributions and stochastic volatility
  modelling.
\newblock {\em Scand.\ J.\ Statist.\/}~{\em 24\/}(1), 1--13.

\bibitem[\protect\citeauthoryear{Bolin}{Bolin}{2014}]{bolin11}
Bolin, D. (2014).
\newblock Spatial {M}at\'ern fields driven by non-{G}aussian noise.
\newblock {\em Scand.\ J.\ Statist.\/}~{\em 41}, 557--579.

\bibitem[\protect\citeauthoryear{Bolin and Kirchner}{Bolin and
  Kirchner}{2019}]{bolin2017rational}
Bolin, D. and K.~Kirchner (2019).
\newblock The rational {SPDE} approach for {G}aussian random fields with
  general smoothness.
\newblock {\em J.\ Comput.\ Graph.\ Statist. (in press)\/}.

\bibitem[\protect\citeauthoryear{Bolin and Lindgren}{Bolin and
  Lindgren}{2011}]{bolin09b}
Bolin, D. and F.~Lindgren (2011).
\newblock Spatial models generated by nested stochastic partial differential
  equations, with an application to global ozone mapping.
\newblock {\em Ann.\ Appl.\ Statist.\/}~{\em 5\/}(1), 523--550.

\bibitem[\protect\citeauthoryear{Dempster, Laird, and Rubin}{Dempster
  et~al.}{1977}]{dempster1977maximum}
Dempster, A.~P., N.~M. Laird, and D.~B. Rubin (1977).
\newblock Maximum likelihood from incomplete data via the {EM} algorithm.
\newblock {\em J.\ Roy.\ Statist.\ Soc.\ Ser.\ B Stat.\ Methodol.\/}~{\em
  39\/}(1), 1--38.

\bibitem[\protect\citeauthoryear{Du, Leonenko, Ma, and Shu}{Du
  et~al.}{2012}]{du2012hyperbolic}
Du, J., N.~Leonenko, C.~Ma, and H.~Shu (2012).
\newblock Hyperbolic vector random fields with hyperbolic direct and cross
  covariance functions.
\newblock {\em Stoch.\ Anal.\ Appl.\/}~{\em 30\/}(4), 662--674.

\bibitem[\protect\citeauthoryear{Genton and Kleiber}{Genton and
  Kleiber}{2015}]{genton2015cross}
Genton, M.~G. and W.~Kleiber (2015).
\newblock Cross-covariance functions for multivariate geostatistics.
\newblock {\em Stat.\ Sci.\/}~{\em 30\/}(2), 147--163.

\bibitem[\protect\citeauthoryear{Genton, Padoan, and Sang}{Genton
  et~al.}{2015}]{genton2015multivariate}
Genton, M.~G., S.~A. Padoan, and H.~Sang (2015).
\newblock Multivariate max-stable spatial processes.
\newblock {\em Biometrika\/}~{\em 102\/}(1), 215--230.

\bibitem[\protect\citeauthoryear{Gneiting, Kleiber, and Schlather}{Gneiting
  et~al.}{2010}]{gneiting2012matern}
Gneiting, T., W.~Kleiber, and M.~Schlather (2010).
\newblock Mat{\'e}rn cross-covariance functions for multivariate random fields.
\newblock {\em J.\ Amer.\ Statist.\ Assoc.\/}~{\em 105\/}(491), 1167--1177.

\bibitem[\protect\citeauthoryear{Gneiting and Raftery}{Gneiting and
  Raftery}{2007}]{gneiting2007strictly}
Gneiting, T. and A.~E. Raftery (2007).
\newblock Strictly proper scoring rules, prediction, and estimation.
\newblock {\em J.\ Amer.\ Statist.\ Assoc.\/}~{\em 102\/}(477), 359--378.

\bibitem[\protect\citeauthoryear{Gr{\"a}ler}{Gr{\"a}ler}{2014}]{graler2014modelling}
Gr{\"a}ler, B. (2014).
\newblock Modelling skewed spatial random fields through the spatial vine
  copula.
\newblock {\em Spat.\ Stat.\/}~{\em 10}, 87--102.

\bibitem[\protect\citeauthoryear{H{\"o}rmann and Leydold}{H{\"o}rmann and
  Leydold}{2014}]{hormann2014generating}
H{\"o}rmann, W. and J.~Leydold (2014).
\newblock Generating generalized inverse gaussian random variates.
\newblock {\em Statistics and Computing\/}~{\em 24\/}(4), 547--557.

\bibitem[\protect\citeauthoryear{Hu, Simpson, Lindgren, and Rue}{Hu
  et~al.}{2013}]{hu2013multivariate}
Hu, X., D.~Simpson, F.~Lindgren, and H.~Rue (2013).
\newblock Multivariate {G}aussian random fields using systems of stochastic
  partial differential equations.
\newblock {\em Preprint, arXiv:1307.1379\/}.

\bibitem[\protect\citeauthoryear{Hu and Steinsland}{Hu and
  Steinsland}{2016}]{hu2016spatial}
Hu, X. and I.~Steinsland (2016).
\newblock Spatial modeling with system of stochastic partial differential
  equations.
\newblock {\em Wiley Interdisciplinary Reviews: Comput.\ Statist.\/}~{\em
  8\/}(2), 112--125.

\bibitem[\protect\citeauthoryear{J{\o}rgensen}{J{\o}rgensen}{1982}]{Jorgensen}
J{\o}rgensen, B. (1982).
\newblock {\em Statistical properties of the generalized inverse Gaussian
  distribution}.
\newblock Lecture Notes in Statistics. Springer-Verlag.

\bibitem[\protect\citeauthoryear{Kazianka and Pilz}{Kazianka and
  Pilz}{2010}]{kazianka2010copula}
Kazianka, H. and J.~Pilz (2010).
\newblock Copula-based geostatistical modeling of continuous and discrete data
  including covariates.
\newblock {\em Stoch.\ Environ.\ Res.\ Risk Assess.\/}~{\em 24\/}(5), 661--673.

\bibitem[\protect\citeauthoryear{Krupskii, Huser, and Genton}{Krupskii
  et~al.}{2018}]{krupskii2016factor}
Krupskii, P., R.~Huser, and M.~G. Genton (2018).
\newblock Factor copula models for replicated spatial data.
\newblock {\em J.\ Amer.\ Statist.\ Assoc.\/}~{\em 113\/}(521), 467--479.

\bibitem[\protect\citeauthoryear{Krupskii and Joe}{Krupskii and
  Joe}{2015}]{krupskii2015structured}
Krupskii, P. and H.~Joe (2015).
\newblock Structured factor copula models: Theory, inference and computation.
\newblock {\em J.\ Multivar.\ Anal.\/}~{\em 138}, 53--73.

\bibitem[\protect\citeauthoryear{Kushner and Yin}{Kushner and
  Yin}{2003}]{kushner2003stochastic}
Kushner, H.~J. and G.~Yin (2003).
\newblock {\em Stochastic approximation and recursive algorithms and
  applications}, Volume~35.
\newblock Springer Science \& Business Media.

\bibitem[\protect\citeauthoryear{Kuusela and Stein}{Kuusela and
  Stein}{2018}]{kuusela2018locally}
Kuusela, M. and M.~L. Stein (2018).
\newblock Locally stationary spatio-temporal interpolation of argo profiling
  float data.
\newblock {\em Proceedings of the Royal Society A\/}~{\em 474\/}(2220),
  20180400.

\bibitem[\protect\citeauthoryear{Lindgren and Rue}{Lindgren and
  Rue}{2015}]{lindgren2015bayesian}
Lindgren, F. and H.~Rue (2015).
\newblock Bayesian spatial modelling with {R-INLA}.
\newblock {\em J.\ Statist.\ Software\/}~{\em 63\/}(19), 1--25.

\bibitem[\protect\citeauthoryear{Lindgren, Rue, and Lindstr\"{o}m}{Lindgren
  et~al.}{2011}]{lindgren10}
Lindgren, F., H.~Rue, and J.~Lindstr\"{o}m (2011).
\newblock An explicit link between {G}aussian fields and {G}aussian {M}arkov
  random fields: the stochastic partial differential equation approach (with
  discussion).
\newblock {\em J.\ Roy.\ Statist.\ Soc.\ Ser.\ B Stat.\ Methodol.\/}~{\em 73},
  423--498.

\bibitem[\protect\citeauthoryear{Ma}{Ma}{2013a}]{ma2013mittag}
Ma, C. (2013a).
\newblock Mittag-{L}effler vector random fields with {M}ittag-{L}effler direct
  and cross covariance functions.
\newblock {\em Ann.\ Inst.\ of Statist.\ Math.\/}~{\em 65\/}(5), 941--958.

\bibitem[\protect\citeauthoryear{Ma}{Ma}{2013b}]{ma2013student}
Ma, C. (2013b).
\newblock Student's t vector random fields with power-law and log-law decaying
  direct and cross covariances.
\newblock {\em Stoch.\ Anal.\ Appl.\/}~{\em 31\/}(1), 167--182.

\bibitem[\protect\citeauthoryear{Mat\'{e}rn}{Mat\'{e}rn}{1960}]{matern60}
Mat\'{e}rn, B. (1960).
\newblock Spatial variation.
\newblock {\em Meddelanden fr{\aa}n statens skogsforskningsinstitut\/}~{\em
  49\/}(5).

\bibitem[\protect\citeauthoryear{Matheson and Winkler}{Matheson and
  Winkler}{1976}]{matheson1976scoring}
Matheson, J.~E. and R.~L. Winkler (1976).
\newblock Scoring rules for continuous probability distributions.
\newblock {\em Manag.\ Sci.\/}~{\em 22\/}(10), 1087--1096.

\bibitem[\protect\citeauthoryear{MATLAB}{MATLAB}{2015}]{Matlab}
MATLAB (2015).
\newblock {\em 8.6.0.267246 (R2015b)}.
\newblock Natick, Massachusetts: The MathWorks Inc.

\bibitem[\protect\citeauthoryear{Pourahmadi}{Pourahmadi}{2011}]{pourahmadi2011covariance}
Pourahmadi, M. (2011).
\newblock Covariance estimation: The glm and regularization perspectives.
\newblock {\em Stat.\ Sci.\/}, 369--387.

\bibitem[\protect\citeauthoryear{Rajput and Rosinski}{Rajput and
  Rosinski}{1989}]{rajput1989spectral}
Rajput, B.~S. and J.~Rosinski (1989).
\newblock Spectral representations of infinitely divisible processes.
\newblock {\em Probab.\ Theory Related Fields\/}~{\em 82\/}(3), 451--487.

\bibitem[\protect\citeauthoryear{R{\o}islien and Omre}{R{\o}islien and
  Omre}{2006}]{roislien06}
R{\o}islien, J. and H.~Omre (2006).
\newblock T-distributed random fields: A parametric model for heavy-tailed
  well-log data.
\newblock ~{\em 38\/}(7), 821--849.

\bibitem[\protect\citeauthoryear{Rosi\'{n}ski}{Rosi\'{n}ski}{1991}]{rosinski91}
Rosi\'{n}ski, J. (1991).
\newblock On a class of infinitely divisible processes represented as mixtures
  of {G}aussian processes.
\newblock In {\em Stable Processes and Related Topics}, Volume~25 of {\em
  Progress in Probability}, pp.\  27--41. Boston: Birkhauser.

\bibitem[\protect\citeauthoryear{Rubio and Steel}{Rubio and
  Steel}{2018}]{rubio2018flexible}
Rubio, F. and M.~Steel (2018).
\newblock Flexible linear mixed models with improper priors for longitudinal
  and survival data.
\newblock {\em Electron.\ J.\ Stat.\/}~{\em 12\/}(1), 572--598.

\bibitem[\protect\citeauthoryear{Rue and Held}{Rue and Held}{2005}]{rue1}
Rue, H. and L.~Held (2005).
\newblock {\em {Gaussian} {Markov} Random Fields; Theory and Applications},
  Volume 104 of {\em Monographs on Statistics and Applied Probability}.
\newblock Boca Raton, FL: Chapman \& Hall/CRC.

\bibitem[\protect\citeauthoryear{Rue and Martino}{Rue and
  Martino}{2007}]{RueMartino07}
Rue, H. and S.~Martino (2007).
\newblock Approximate {Bayesian} inference for hierarchical {Gaussian} {Markov}
  random field models.
\newblock {\em J.\ Statist.\ Plann.\ and Inference\/}~{\em 137\/}(10),
  3177--3192.

\bibitem[\protect\citeauthoryear{Wallin and Bolin}{Wallin and
  Bolin}{2015}]{Wallin15}
Wallin, J. and D.~Bolin (2015).
\newblock Geostatistical modelling using non-{G}aussian {M}at\'ern fields.
\newblock {\em Scand.\ J.\ Statist.\/}~{\em 42}, 872--890.

\bibitem[\protect\citeauthoryear{Whittle}{Whittle}{1963}]{whittle63}
Whittle, P. (1963).
\newblock Stochastic processes in several dimensions.
\newblock {\em Bull.\ Internat.\ Statist.\ Inst.\/}~{\em 40}, 974--994.

\end{thebibliography}
	
\end{document}